%% file: main.tex
\documentclass{LMCS}

\def\dOi{10(2:7)2014}
\lmcsheading%
{\dOi}
{1--33}
{}
{}
{Nov.~\phantom02 2010}
{Jun.~\phantom06, 2014}
{}

\ACMCCS{[{\bf Theory of computation}]: Logic; Formal languages and
  automata theory---Grammars and context-free
  languages\,/\,Formalisms---Rewrite systems; Semantics and reasoning---Program semantics\,/\,Program reasoning}

\input{packages}

\input{macros}

\begin{document}

       
\title[Infinitary Term Rewriting for Weakly Orthogonal Systems]%
      {Infinitary Term Rewriting for\\ Weakly Orthogonal Systems\\ \normalfont{Properties and Counterexamples}}

\author[Endrullis]{J\"{o}rg Endrullis\rsuper a}
\address{{\lsuper{a,b,c,d}}%
  VU University Amsterdam,
  Department of Computer Science,
  De Boelelaan 1081a,
  1081 HV Amsterdam,
  The Netherlands
}
\email{\{j.endrullis, c.a.grabmayer, r.d.a.hendriks, j.w.klop\}@vu.nl}

\author[Grabmayer]{Clemens Grabmayer\rsuper b}
\address{\vspace{-18 pt}}

\author[Hendriks]{Dimitri Hendriks\rsuper c} 
\address{\vspace{-18 pt}}

\author[Klop]{\\Jan Willem Klop\rsuper d}
\address{\vspace{-18 pt}}

\author[van~Oostrom]{Vincent van Oostrom\rsuper e}
\address{{\lsuper e}%
  Utrecht University,
  Department of Philosophy,
  Janskerkhof 13/13a,
  3512 BL Utrecht,
  The Netherlands
}
\email{vincent.vanoostrom@phil.uu.nl}


\begin{abstract}
We present some contributions to the theory of infinitary rewriting
for weakly orthogonal term rewrite systems, in which critical pairs may
occur provided they are trivial. 

We show that the infinitary unique normal form property ($\UNinf$) fails by an example
of a weakly orthogonal TRS with two collapsing rules. 
By translating this example, we show that $\UNinf$ also fails 
for the infinitary $\lambda\beta\eta$-calculus.

As positive results we obtain the following:
Infinitary confluence, and hence $\UNinf$, holds
for weakly orthogonal TRSs that do not contain collapsing rules. 
To this end we refine the compression lemma.
Furthermore, we establish the triangle and diamond properties
for infinitary multi-steps (complete developments) in weakly orthogonal TRSs, 
by refining an earlier cluster-analysis for the finite case. 

\end{abstract}

\maketitle

\section{Introduction}\label{sec:intro}


While the theory of infinitary term rewriting is well-developed
for orthogonal rewrite systems, much less is known about
infinitary rewriting in non-orthogonal systems, in which
critical pairs between rules may occur. 
In this paper we address a simple 
                                  weakening of orthogonality:
weakly orthogonal systems, in which critical pairs may occur
provided that they are trivial.
%
%
Thus conceptually, weakly orthogonal systems only deviate 
little from orthogonal ones. 
And whereas, in the case of finitary rewriting, only a few 
rewrite properties known in orthogonal systems turned out to fail (e.g.\ head normalization), 
most such properties have been established to hold~\cite{kete:klop:oost:2004a}.
But this required the development of a rewrite theory specific to weakly orthogonal systems with
tailor-made notions and techniques. 
%

We show that the infinitary rewrite theory known for orthogonal systems
fails dramatically 
                   for weakly orthogonal systems.
In Section~\ref{sec:pscountex}, we present and analyze a counterexample to the 
infinitary unique normal form property $\UNinf$ in term rewriting systems (TRSs).
In Section~\ref{sec:beta:eta} we translate this example into the $\lambda\beta\eta$-calculus, 
and in this way obtain also a counterexample to $\UNinf$ in
this paradigmatic example of a weakly orthogonal higher-order rewrite system.
%

In the remaining sections we show that,
under simple restrictions,
much of the theory of infinitary rewriting in orthogonal systems
can be regained:
we establish the triangle property, and hence the diamond property, for developments
in weakly orthogonal TRSs without collapsing rules.
An important ingredient of the proofs is a refinement of the compression lemma
(Section~\ref{sec:compression}). 

This paper extends our RTA~2010 contribution~\cite{endr:grab:hend:klop:oost:2010}.
We have elaborated the results in more detail and have filled in missing proofs.
The main change concerns Section~\ref{sec:multi} where we strengthen the results of~\cite{endr:grab:hend:klop:oost:2010}
in two ways:
We give a proof of the \emph{triangle property}, 
and we include an alternative proof of the diamond property 
for the multi-step reduction where the common reduct is obtained \emph{effectively}.
The (tri)angle property is a strong form of the diamond property where the joining term only depends on the initial~term.

\section{Basic Definitions}\label{sec:prelims}

For a general introduction to infinitary rewriting (predominantly for the case of orthogonal systems)
we refer to~\cite[Ch.12]{terese:2003}, \cite{klop:vrij:2005,kenn:klop:slee:vrie:1995}.
In this section we gather definitions of the basic notions. 


Infinite terms can be introduced in several ways, see~\cite{kete:2006} for an overview.
Here we choose the most concrete definition.
We define terms as partial mappings from the set of positions $\nat^\ast$ 
to the alphabet symbols of some first-order signature $\asig$.
An alternative definition of infinite terms is 
by the completion of the metric space of finite terms 
with the usual metric based on the familiar notion of distance that yields distance 
$2^{-(n+1)}$ for a pair of terms that are identical up to and including level $n$ 
from the root, but then have a difference (see Definition~\ref{def:metric}).

We will consider a finite or infinite term as a function
on a prefix-closed subset of $\nat^\ast$ taking values in a first-order signature.
A \emph{signature $\asig$} is a finite set of symbols~$f$
each having a fixed \emph{arity $\arity{f} \in \nat$}.
Let $\avars$ be a set of symbols, 
called \emph{variables}, such that $\avars\cap\asig = \setemp$.
Then, a \emph{term over $\asig$} 
is a partial map $t\funin{\nat^\ast \pto \asig\cup\avars}$ such that
the \emph{root} is defined, $\funap{t}{\posemp} \in \asig\cup\avars$, 
and for all $\apos\in\nat^\ast$ and all $i\in\nat$ we have 
$\funap{t}{\posconcat{p}{i}} \in \asig\cup\avars$
if and only if $\funap{t}{\apos} \in \asig$ of arity $n$ and $1 \le i \le n$.
The set of (not necessarily well-founded) terms over $\asig$ and $\avars$ 
is denoted by $\iter{\asig,\avars}$.
Usually we will write $\iter{\asig}$ 
for the set of 
terms 
over $\asig$ and countably infinite set of variables, 
which is assumed to be fixed as underlying the definition of terms.
By $\vars{\atrm}$ we denote the set of terms that occur in a term $\atrm$.

The set of \emph{positions $\pos{t}$ of a term $t\in\iter{\asig}$} is the domain of $t$, 
that is, the set of values $\apos\in\nat^\ast$ such that $\funap{t}{\apos}$ is defined:
$\pos{t} \defdby \{\apos\in\nat^\ast \where \funap{t}{\apos}\in\asig\cup\avars\}$.
Note that, by the definition of terms, the set $\pos{t}$ is prefix-closed. 
A term $t$ is called \emph{finite} if the set $\pos{t}$ is finite.
We write $\ter{\asig,\avars}$ or $\ter{\asig}$ for the set of finite terms.
For positions $\apos\in\pos{t}$ we use $\subtrmat{t}{\apos}$ to denote 
the \emph{subterm of $t$ at position $\apos$},
defined by $\funap{\subtrmat{t}{\apos}}{\bpos} \defdby \funap{t}{\posconcat{\apos}{\bpos}}$ 
for all $\bpos\in\nat^\ast$.
For terms $s$ and $t$ and a position~$p \in \pos{s}$ we denote by $s[t]_p$ the term obtained from $s$
by replacing the subterm at position~$p$ by $t$.

For a symbol $f \in \asig$ of arity $n$ 
and terms $t_1,\ldots,t_n \in \iter{\asig}$
we write $\funap{f}{t_1,\ldots,t_n}$ to denote the term $t$ defined by
$\funap{t}{\posemp} = f$, and $\funap{t}{\posconcat{i}{\apos}} = \funap{t_i}{\apos}$ for all 
$1 \le i \le n$ and $\apos\in\nat^\ast$. 
For \emph{constants} $c\in\asig$, i.e., with $\arity{c} = 0$, 
we simply write $c$ instead of $\funap{c}{\nix}$.
We use $x,y,z,\ldots$ to range over variables.
%
%


A \emph{substitution} is a map $\asubst \funin \avars \to \iter{\asig,\avars}$.
For terms $t \in \iter{\asig,\avars}$ and substitutions $\asubst$ 
we define $\subst{\asubst}{t}$ 
as the result of replacing each $x\in\avars$ in $t$ by $\funap{\asubst}{x}$.
Formally, $\subst{\asubst}{t}$ is defined, for all $\apos \in \nat^\ast$, by:
$\funap{ \subst{\asubst}{t} }{\apos} = \funap{ \funap{\asubst}{\funap{t}{\apos_0}} }{ \apos_1 }$
if there exist $\apos_0, \apos_1 \in \nat^\ast$ such that 
$p = \posconcat{\apos_0}{\apos_1}$ and $\funap{t}{\apos_0} \in \avars$,
and $\subst{\asubst}{t}(\apos) = \funap{t}{\apos}$, otherwise.
Let $\contexthole$ be a fresh 
symbol, $\contexthole\not\in\asig\cup\avars$.
A \emph{context} $\acontext$ is a term in $\iter{\asig,\avars\cup\{\contexthole\}}$
that contains precisely one occurrence of $\contexthole$.
By $\contextfill{\acontext}{s}$ we denote the term $\subst{\asubst}{\acontext}$
where $\funap{\asubst}{\contexthole} = s$ and $\funap{\asubst}{x} = x$ for all $x \in \avars$.

\begin{definition}
  An \emph{infinitary term rewriting system (iTRS)} is a pair $\atrs = \pair{\asig}{\asetofrules}$
  consisting of a first-order signature $\asig$
  and a set $\asetofrules$ of infinitary rewrite rules over $\asig$ (and a set of variables $\avars$):
  an \emph{infinitary rewrite rule} is a pair $\pair{\ell}{r}$, usually written as $\ell \to r$,
  where $\ell\in\ter{\asig}$ and $r\in\iter{\asig}$,
  and 
  such that for \emph{left-hand side $\ell$} and \emph{right-hand side $r$}
  we have $\funap{\ell}{\posemp} \not\in \avars$ and $\vars{r}\subseteq\vars{\ell}$.
  
  We call an iTRS a \emph{term rewriting system (TRS)} 
  when the right-hand side of each of its rules is a finite term.
  %
\end{definition}

  In this paper we restrict to \iTRS{s} with \emph{finitely many rules}:
  all iTRSs are subjected to this restriction without explicit mention henceforth.
  While many of our results hold for general TRS{s} and \iTRS{s}, the definition of an effective orthogonalization procedure
  and its use in Section~\ref{sec:multi} assume this restriction.

\begin{definition}\normalfont\label{def:metric}
  On the set of terms $\iter{\asig}$ we define a \emph{metric}
  $\ametric$ by $\metric{s}{t} = 0$ whenever $s \equiv t$,
  and $\metric{s}{t} = 2^{-k}$ otherwise, 
  where $k \in \nat$ is the least length of all positions $p \in \nat^\ast$ 
  such that $\funap{s}{p} \neq \funap{t}{p}$.
\end{definition}


An iTRS $\atrs$ induces a rewrite relation on the set of terms as follows.
\begin{definition}\label{def:strong:conv}\normalfont
  Let $\atrs = \pair{\asig}{\asetofrules}$ be an iTRS, 
  $s,t\in\iter{\asig}$ terms and $\apos \in \nat^\ast$ a position. 
  We write $s \redrat{\atrs\!}{\apos} t$ if there exist
  a rule $\ell \to r \in R$, a substitution $\asubst$ 
  and a context $\acontext$ with $\trmat{\acontext}{\apos} = \contexthole$
  such that 
  $s \equiv \contextfill{\acontext}{\subst{\asubst}{\ell}}$
  and
  $t \equiv \contextfill{\acontext}{\subst{\asubst}{r}}$.
  We write $s \redr{\atrs\!} t$ if $s \redrat{\atrs\!}{\apos} t$ for some $\apos\in\nat^\ast$;
  we call the length of $\apos$ the \emph{depth} of the rewrite step.


   A \emph{(strongly continuous) transfinite rewrite sequence of length $\alpha$}, where $\alpha$ an ordinal,
   is a sequence of rewrite steps 
   $(t_{\beta} \redrat{\atrs}{\apos_{\beta}} t_{\beta+1})_{\beta < \alpha}$
   such that for every limit ordinal $\lambda < \alpha$ we have that if 
   $\beta$ approaches $\lambda$ from below, then:
   \begin{enumerate}
     \item 
       the distance $\metric{t_\beta}{t_\lambda}$ tends to $0$ 
       and, moreover,
     \item 
       the depth of the rewrite action 
       tends to infinity.
   \end{enumerate}
   We write $s \infred$ (or $s \to^\alpha$) for a transfinite rewrite sequence (of length $\alpha$) starting from $s$.
      
   A transfinite rewrite sequence of length $\alpha$ is called \emph{strongly convergent}
   if either $\alpha$ is not a limit ordinal or property~(ii) holds also
   when the limit ordinal $\lambda = \alpha$ is approached by ordinals $\beta$ from below,
   which guarantees that the \emph{limit} $\lim_{\beta\to\alpha} t_{\beta}$ of the sequence exists in $\iter{\asig}$. 
   %
   We will indicate strongly convergent rewrite sequences of length $\alpha$ from source term $\aitrm{0}$ 
   and with limit $\atrm$ 
   by $\aitrm{0} \infred \atrm$, 
   or by
   $\aitrm{0}\infred_\atrs \atrm$ (to emphasize the underlying term rewrite system $\atrs$),
   or by $\aitrm{0}\red^{\alpha} \atrm$ (to explicitly indicate the length $\alpha$ of the sequence). 
   We write $\aitrm{0}\red^{\le \alpha} \atrm$ if $\aitrm{0}\red^{\beta} \atrm$ for some $\beta \le \alpha$.
   A transfinite rewrite sequence that is not strongly convergent is called 
   \emph{divergent}. 
\end{definition}%

\begin{remark}
  We comment on Definition~\ref{def:strong:conv}:
  \begin{enumerate}
    \item 
      For rewrite sequences $\aitrm{0}\red^{\alpha} \atrm$ of limit ordinal length $\alpha$,
      the target term~$t$ is formally not part of the rewrite sequence.
      We view the convergence towards $t$ as a property of the sequence; 
      note that the limit term $t$ is unique.
    \item All proper initial segments of a divergent reduction
      are strongly convergent.
    \item The length of transfinite rewrite sequences is always countable 
      (see~\cite{kenn:klop:slee:vrie:1995}). 
     In the sequel we will use the familiar fact that
     countable limit ordinals have cofinality~$\omega$. 
  \end{enumerate}
\end{remark}

\begin{remark}
  In this paper we are concerned with `strongly continuous' rewrite sequences.
  The notion of `weakly continuous' rewrite sequences is obtained 
  by dropping requirement~(ii) in Definition~\ref{def:strong:conv},
  namely the condition that the depth of the rewrite action must tend to infinity.
  The notion of strongly continuous rewrite sequences leads to a more satisfying rewriting theory,
  for example, 
  symbol occurrences can be traced over limit ordinals.
  %
  For this reason, strongly continuous rewrite sequences 
  are the standard notion in the infinitary rewriting literature, see~\cite{terese:2003}.
  
  For an example of a weakly continuous rewrite sequence, we consider the rewrite system:
  \begin{align*}
    a(x) &\to a(b(x)) &
    b(x) &\to c(x)
  \end{align*}
  Then $a(x) \to^\omega a(b^\omega) \to a(c(b^\omega))$ in $\omega+1$ steps.
  Observe that this rewrite sequence cannot be compressed to yield one of length $\le \omega$,
  that is, there exists no weakly continuous reduction $a(x) \to^{\le \omega} a(c(b^\omega))$ 
  of length $\le \omega$.
  For a similar example see~\cite{kenn:klop:slee:vrie:1995}.
\end{remark}

\begin{definition}[Critical pairs]\normalfont\label{def:criticalpair}
  Let $\rho_1 \funin \ell_1 \to r_1$ and $\rho_2 \funin \ell_2 \to r_2$ be rules over $\asig$.
  Then $\rho_1$ has a \emph{critical pair} with $\rho_2$
  if there exists a non-variable position $\apos \in \pos{\ell_1}$ such that $\trmat{\ell_1}{\apos} \in \asig$,
  and $\subtrmat{\ell_1}{\apos}$ and $\ell_2$ have a common instance, that is,
  $\subst{\asubst}{(\subtrmat{\ell_1}{\apos})} \synteq \subst{\bsubst}{\ell_2}$
  for some substitutions $\sigma$ and $\tau$.
  Let $\asubst$ and $\bsubst$ be substitutions
  for which $\subst{\asubst}{(\subtrmat{\ell_1}{\apos})} \synteq \subst{\bsubst}{\ell_2}$
  is the unique (up to renaming of variables) most general such common instance.
  Without loss of generality, let $\asubst$ be minimal in the sense that
  $\dom{\asubst} = \vars{\subtrmat{\ell_1}{\apos}}$
  and let the variables introduced be fresh, that is,
  $\vars{\subst{\asubst}{(\subtrmat{\ell_1}{\apos})}} \cap \vars{\repsubtrmat{\ell_1}{\contexthole}{\apos}} = \setemp$.
  Then $\pair
        {\repsubtrmat{(\subst{\asubst}{\ell_1})}{\subst{\bsubst}{r_2}}{\apos}}
        {\subst{\asubst}{r_1}}$
  is called \emph{critical pair of $\rho_1$ with (inner rule) $\rho_2$}.
  
  A critical pair $\pair{\atrm}{\btrm}$ is called \emph{trivial} if
  $\atrm \synteq \btrm$.
  %
\end{definition}

The name `critical pair' arises from the fact that there is a
peak, a pair of diverging steps, of the form: 
  \begin{gather*}
    \repsubtrmat{\subst{\asubst}{\ell_1}}{\subst{\bsubst}{r_2}}{\apos}
    \redi_{\rho_2} \subst{\asubst}{\ell_1} \red_{\rho_1}
    \subst{\asubst}{r_1}
  \end{gather*}
where the pattern of rules $\rho_1$ and $\rho_2$ overlap.
We also denote critical pairs by this pair of steps.

\begin{definition}[Orthogonal, weakly orthogonal iTRSs]
  Let $\atrs = \pair{\asig}{\asetofrules}$ be an iTRS. 
  
  $\atrs$ is called \emph{left-linear} if no rule $\rho \funin \ell \to r$ of $\atrs$
  contains two or more occurrences of the same variable in its left-hand side $\ell$.
  
  $\atrs$ is called \emph{orthogonal} if it is left-linear, and if it does not contain critical pairs.
  And $\atrs$ is called \emph{weakly orthogonal} if it is left-linear, and if all of its critical pairs are trivial.
\end{definition}

\begin{definition}[Root-active terms]
  Let $\atrs = \pair{\asig}{\asetofrules}$ be an iTRS. 
  A term $\atrm\in\iter{\asig}$ is called \emph{root-active}
  if there exists an transfinite rewrite sequence starting at $\atrm$
  in which infinitely many rewrite steps take place at the root position $\epsilon$.
\end{definition}

\begin{definition}[Infinitary normalization, unique normalization, confluence, Church--Rosser]\label{def:UNi}\label{def:inf:rew:properties}
  Let $\atrs = \pair{\asig}{\asetofrules}$ be an iTRS.
  Let $\sired$ be the infinitary rewrite relation on $\iter{\asig}$ induced by $\atrs$,
  and let $\siconvR{\atrs} \defdby ( \siredi \cup \sired )^*$ the conversion relation 
  belonging to $\sired$.\footnote{%
      An alternative notion of infinitary equational reasoning was introduced recently 
      in \cite[Section~4]{endr:hans:hend:polo:silv:2013}.
      There $\siconvR{\atrs}$ is defined as the greatest fixed point of the equation
      $E = \sredi_\posemp \cup \overline{E} \cup \sred_\posemp$ where, for a relation $R\subseteq\iter{\asig}\times\iter{\asig}$, 
      $\overline{R} = \{\pair{f(\vec{s})}{f(\vec{t})} \where s_i R t_i \, (1 \le i \le \arity{f})\}$.
  }

  The \emph{infinitary}  properties 
  \emph{strong normalization}~$\SNinf$,
  \emph{weak normalization}~$\WNinf$,
  \emph{confluence}~$\CONFinf$,
  \emph{Church--Rosser}~$\CRinf$,
  \emph{normal form property}~$\NFredinf{\sred}$ \emph{with respect to reduction},
  \emph{normal form property}~$\NFinf$,
  \emph{unique normalization}\/~$\UNredinf{\sred}$ \emph{with respect to reduction},
  and
  \emph{unique normalization}~$\UNinf$
  of $\atrs$ are defined as follows: 
    %
    \begin{itemize}[label=$\WNinf$\quad]
      \vspace*{0.55ex}
    \item
      [$\SNinf$\,] \mbox{}
      $ \forall \atrm\in\iter{\asig} \:
        (\text{all infinite rewrite sequences from $\atrm$ are strongly convergent}) $.
      \vspace*{0.55ex}
    \item
      [$\WNinf$\,] \mbox{}
      $ \forall \atrm\in\iter{\asig} \:
        (\text{$\atrm$ has an infinite normal form}) $.
      \vspace*{0.55ex}
    \item
      [$\CONFinf$\,] \mbox{}
      $
        \forall \atrm, \aitrm{1}, \aitrm{2}\in\iter{\asig}\:
           (
           \aitrm{1} \iredi  \atrm  \ired \aitrm{2} 
             \;\implies\;
           \exists \btrm\in\iter{\asig}.\; \aitrm{1} \ired \btrm \iredi \aitrm{2}
           )
      $.
      \vspace*{0.55ex}
    \item
      [$\CRinf$\,] \mbox{}
      $
        \forall \aitrm{1}, \aitrm{2}\in\iter{\asig}\:
           (
           \aitrm{1} \iconvR{\atrs} \aitrm{2} 
             \;\implies\;
           \exists \btrm\in\iter{\asig}.\; \aitrm{1} \ired \btrm \iredi \aitrm{2}
           )
      $.
      \vspace*{0.55ex}
    \item%
      [$\NFredinf{\sred}$\,] \mbox{}
      $ \forall \atrm, \btrm,\ctrm\in\iter{\asig}\:
          ( 
          \atrm  \iredi \ctrm \ired  \btrm 
            \;\logand\;
          \text{$\btrm$ normal form}
             \;\implies\;
          \atrm \ired \btrm 
          )$. 
      \vspace*{0.55ex}
    \item%
      [$\NFinf$\,] \mbox{}
      $ \forall \atrm, \btrm\in\iter{\asig}\:
           (
           \atrm  \iconvR{\atrs} \btrm 
            \;\logand\;
          \text{$\btrm$ normal form}
             \;\implies\;
          \atrm \ired \btrm 
          )$.
      \vspace*{0.55ex}
    \item
      [$\UNredinf{\sred}$\,] \mbox{}
      $
        \forall \atrm, \aitrm{1}, \aitrm{2}\in\iter{\asig}.\:
          (
          \aitrm{1} \iredi  \atrm  \ired \aitrm{2} 
            \wedge
            \text{$\aitrm{1},\, \aitrm{2}$ normal forms}  
              \;\implies\;
                \aitrm{1} \synteq \aitrm{2}  
                )
      $.
      \vspace*{0.55ex}
    \item
      [$\UNinf$\,] \mbox{}
      $
        \forall \aitrm{1}, \aitrm{2}\in\iter{\asig}\;
          ( 
          \aitrm{1} \iconvR{\atrs} \aitrm{2} 
            \;\logand\; 
            \text{$\aitrm{1},\, \aitrm{2}$ normal forms}  
              \;\implies\;
                \aitrm{1} \synteq \aitrm{2}
                )  
      $.
      \vspace*{0.55ex}
    \end{itemize}
    %
    %
\end{definition}


\noindent With the exception of its first and third items,  
the following proposition is an easy consequence of the interdependencies known for the finitary analogues
of the rewrite properties introduced in this definition. 

\begin{proposition}
  For all iTRSs $\atrs = \pair{\asig}{\asetofrules}$ with induced rewrite relation $\sred$ 
  the following implications hold between infinitary rewrite properties for $\atrs\,$:
  \begin{enumerate}[label=(\roman*)]
    \item 
      $ \SNinf \implies \WNinf $.      
    \item $ \CRinf \logequiv \CONFinf $.
    \item 
      $ \NFinf \logequiv \NFredinf{\sred} $.
    \item $ \UNinf \logequiv \UNredinf{\sred} $.
    \item $ \CONFinf \implies \CRinf \implies \NFredinf{\sred} \implies \NFinf \implies \UNredinf{\sred} \implies \UNinf $.
  \end{enumerate}  
\end{proposition}

\begin{proof}
  For the first item, see \cite{klop:vrij:2005}.
  The non-trivial direction $ \NFredinf{\sred} \implies \NFinf $  of (iii), which occurs again in (v),
  can be shown by induction on the number of peaks 
                                                   in a conversion that witnesses $\atrm \iconvR{\atrs} \btrm$
  where $\btrm$ is a normal form.
  The other implications follow
  from known relationships between corresponding finitary properties of abstract reduction systems, see \cite[Ch.1]{terese:2003},
  by the following observation:
  for every iTRS~$\atrs$ with induced rewrite relation~$\sred$ and infinite rewrite relation $\sired$,
  a property $\CONFinf$, $\CRinf$, $\NFinf$, $\UNredinf{\sred}$, or $\UNinf$ holds for $\atrs$ 
  if and only if, respectively, the corresponding property 
  $\CONF$ (confluence), $\CR$, $\NF$, $\UNred{\sred}$, $\UN$, 
  holds for the abstract reduction system $\pair{\iter{\asig}}{\sired}$.
\end{proof}

\section{A Counterexample to \texorpdfstring{$\UNinf$}{UNinfty} for Weakly Orthogonal Systems}\label{sec:pscountex}
\newcommand{\rasi}{\xi}%
\newcommand{\rafn}{\zeta}%

In~\cite{kenn:klop:slee:vrie:1995} it has been 
shown that infinitary unique normalization ($\UNinf$) holds
for orthogonal term rewrite systems
(see also \cite{klop:vrij:2005}).
In 
   sharp contrast to this, 
we will now demonstrate that the property $\UNinf$
does \emph{not} generalize to weakly orthogonal TRSs.
The following simple counterexample can be used:
for the signature consisting of the unary symbols $\spre$ and $\ssuc$, consider the rewrite rules
$\pre{\suc{x}} \to x$ and $\suc{\pre{x}} \to x$.
Clearly this TRS is weakly orthogonal.

Employing the obvious correspondence between TRSs with only unary function symbols
and string rewrite systems (SRSs), in the sequel we consider the corresponding SRS:
\begin{align*}
  \spre\ssuc &\to \wordemp &
  \ssuc\spre &\to \wordemp
\end{align*}
where $\wordemp$ is the empty word.
%
If $u$ is a finite word, we write $u^\omega$ for the infinite word $u u u \cdots$.
Using $\ssuc$ and $\spre$ we have infinite words such as 
$\rafn = (\spre\ssuc)^\omega$.
Note that $\ssuc^\omega$ and $\spre^\omega$ are the only infinite normal forms,
and that $\rafn$ only reduces to itself.

Given an infinite \psword~$w$ we can plot in a graph 
the surplus number of $\ssuc$'s of $w$ when stepping through the word $w$ from left to right,
see e.g.\ Figure~\ref{fig:SP:graph}.
The graph is obtained by counting $\ssuc$ for $+1$ and $\spre$ for $-1$.
We define $\funap{\mrm{sum}}{w,n}$ as the result of this counting up to depth $n$ in the word $w$
(if $w$ is finite we define $\funap{\mrm{sum}}{w} = \funap{\mrm{sum}}{w,\lstlength{w}}$):
\begin{align*}
  \funap{\mrm{sum}}{w,0} &= 0 &
  \funap{\mrm{sum}}{\wordemp,n} &= 0  \\
  \funap{\mrm{sum}}{\ssuc w,n+1} &= \funap{\mrm{sum}}{w,n} + 1 &
  \funap{\mrm{sum}}{\spre w,n+1} &= \funap{\mrm{sum}}{w,n} - 1
\end{align*}
For $w= (\ssuc\spre)^\omega$ 
the graph takes values, consecutively, $1,0,1,0,\ldots$, 
for $w=\ssuc^\omega$ it takes $1,2,3,\ldots$,
and for $w=\spre^\omega$ we have $-1,-2,-3,\ldots$.

\begin{figure}[ht!]
  \begin{center}
    \input{figs/SPgraph}
  \end{center}
  \caption{\textit{%
    Graph for the oscillating \psword\ 
    $\pscountex = \spre^1\,\ssuc^2\,\spre^3\,\cdots$\,. 
  }}
  \label{fig:SP:graph}
\end{figure}

We define the \emph{$\ssuc$-norm} $\snorm{w}$ and \emph{$\spre$-norm} $\pnorm{w}$ of $w$:
\begin{align}
  \snorm{w} &= \sup_{n \in \nat} \funap{\mrm{sum}}{w,n} &
  \pnorm{w} &= \sup_{n \in \nat} (- \funap{\mrm{sum}}{w,n})
\end{align}
So the $\ssuc$-norm ($\spre$-norm) of $(\ssuc\spre)^\omega$
is $1$ ($0$), of $\ssuc^\omega$ it is $\infinity$ ($0$), and of $\spre^\omega$ it is $0$ ($\infinity$).

\begin{lemma}\label{lem:collapse}
  Let $w$ be a finite \psword, and let $z = \funap{\mrm{sum}}{w}$.
  Then $w \mred \ssuc^z$ if $z \ge 0$, and $w \mred \spre^{-z}$ if $z < 0$.
\end{lemma}
\begin{proof}
  For finite words $u, v$ we have that $u \to v$ implies $\funap{\mrm{sum}}{u} = \funap{\mrm{sum}}{v}$.
  Moreover, $\to$ is normalising,
  and the only normal forms are of the form $\ssuc^k$ and $\spre^k$ for $k \ge 0$.
\end{proof}

\begin{proposition}\label{prop:norm}
  \hfill
  \begin{enumerate}
    \item 
      $w \ired \ssuc^\omega$ 
      if and only if $\snorm{w} = \infinity$,
    \item $w \ired \spre^\omega$ if and only if $\pnorm{w} = \infinity$.
  \end{enumerate}
\end{proposition}
\begin{proof}
  We consider only (i) as case (ii) can be treated analogously.

  We start with the direction `$\Leftarrow$'
  From $\snorm{w} = \infinity$ it follows that $w = w_1 w_2 \cdots$ 
  with finite words $w_1$, $w_2$, \ldots{} such that $\funap{\mrm{sum}}{w_i} = 1$ for all $i \in \nat$.
  Then $w_i \mred \ssuc$ for all $i \in \nat$ by Lemma~\ref{lem:collapse} and hence $w \ired \ssuc^\omega$.

  For `$\Rightarrow$' we argue as follows.
  By compression there is a rewrite sequence $w = w_0 \to w_1 \to w_2 \to \cdots$
  of length $\omega$ with limit $\ssuc^\omega$.
  Consequently, for every $n \in \nat$ there exists $i\in \nat$ 
  such that $\ssuc^n$ is a prefix of $w_i$, and hence, $\snorm{w_i} \ge n$.  
  Moreover, we have $\snorm{w_0} \ge \snorm{w_1} \ge \snorm{w_2} \ge \cdots$
  since removing $\ssuc\spre$ or $\spre\ssuc$ cannot increase the norm \mbox{$\snorm{{\cdot}}$}
  (it either stays constant or decreases by $1$).
  As a consequence we obtain that $\snorm{w_0} \ge n$ holds for every $n\in\nat$.
  It follows that $\snorm{w} = \snorm{w_0} = \infinity$.
\end{proof}

Note that in Proposition~\ref{prop:norm}
$w \ired \ssuc^\omega$ 
can always be achieved using the rule $\spre\ssuc \to \wordemp$ only.
And likewise the rule $\ssuc\spre \to \wordemp$ for $w \ired \spre^\omega$.

Now let us take a word $\pscountex$ with $\snorm{\pscountex} = \infinity$ 
\emph{and} $\pnorm{\pscountex} = \infinity$\,!
Then by the previous proposition $\pscountex$ reduces to 
both $\ssuc^\omega$ and $\spre^\omega$, both normal forms.
Hence $\UNinf$ fails.
Indeed, such a term $\pscountex$ can be found:
\begin{align*}
  \pscountex & =
  \spre \, \ssuc\ssuc \, \spre\spre\spre \, \ssuc\ssuc\ssuc\ssuc \, 
  \spre\spre\spre\spre\spre \, \ssuc\ssuc\ssuc\ssuc\ssuc\ssuc \, \cdots
\end{align*}
The graph for this word is displayed in Figure~\ref{fig:SP:graph}.
If we only apply rule $\spre\ssuc \to \wordemp$
the $\spre$-blocks are absorbed by the larger $\ssuc$-blocks to their right,
leaving the normal form $\ssuc^\omega$.
Likewise, applying only $\ssuc\spre \to \wordemp$ yields $\spre^\omega$.

We find that $\pscountex \ired w$ for every infinite \psword{} $w$,
and furthermore, the following generalisation.

\begin{proposition}
  Every infinite \psword{} that reduces to both $\ssuc^\omega$ and $\spre^\omega$
  reduces to any infinite \psword{}.
\end{proposition}

\begin{proof}
  Suppose that $w$ is an infinite \psword{} with 
  $\spre^\omega \iredi w \ired \ssuc^\omega$,
  and let $u$ be the infinite \psword{} we want to obtain. 
  By applying Proposition~\ref{prop:norm} to $w$,
  we find: $\pnorm{w} = \snorm{w} = \infinity$. 
  This allows us to choose a partition $w = w_1 w_2 \cdots$ of $w$
  into finite words $w_1, w_2, \ldots$ such that, for all $i\in\nat$, it holds that
  $\funap{\mrm{sum}}{w_i} = 1$ if $\nth{u}{i} = \ssuc$, and 
  $\funap{\mrm{sum}}{w_i} = -1$ if $\nth{u}{i} = \spre$.
  From this, by Lemma~\ref{lem:collapse} we obtain rewrite sequences $w_i \mred \nth{u}{i}$, for all $i\in\nat$.
  By consecutively performing the corresponding finite rewrite sequences on the subwords $w_i$ of $w$, 
  we obtain a strongly convergent rewrite sequence that witnesses $w \ired u$.
\end{proof}

Hence, not only is $\pscountex$ a counterexample to $\UNinf$ for weakly orthogonal rewrite systems, 
but also, $\pscountex$ rewrites to $(\spre\ssuc)^\omega$, a word which has no normal form.
Thus, in contrast to orthogonal systems~\cite{terese:2003,klop:vrij:2005}, 
for individual terms the property $\WNinf$ of infinitary weak normalization
is not preserved under infinite rewriting.


\begin{figure}[ht!]
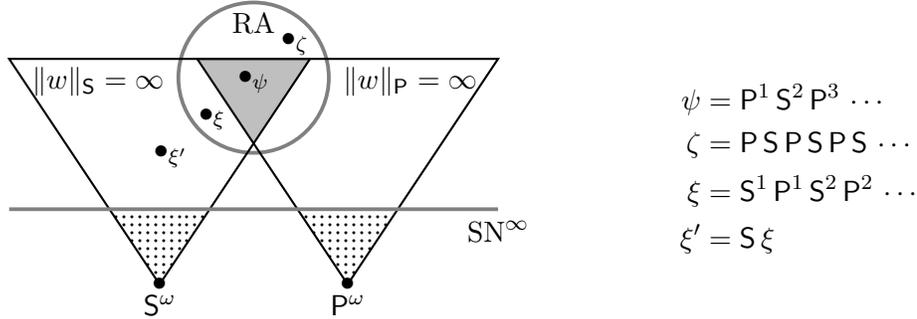

  \begin{minipage}{.50\textwidth}
  \begin{center}
    \include{figs/venn}
  \end{center}
  \end{minipage}
  \begin{minipage}{.40\textwidth}
  \begin{center}
  \begin{align*}
   \pscountex & = \spre^1\,\ssuc^2\,\spre^3\,\cdots \\
   \rafn      & = \spre\,\ssuc\,\spre\,\ssuc\,\spre\,\ssuc\,\cdots \\
   \rasi      & = \ssuc^1\,\spre^1\,\ssuc^2\,\spre^2\,\cdots \\
   \rasi'     & = \ssuc\,\rasi 
  \end{align*}
  \end{center}
  \end{minipage}
  \caption{\textit{Venn diagram of infinite \psword{s}.}}
  \label{fig:SP:Venn}
\end{figure}
Figure~\ref{fig:SP:Venn} shows a more detailed analysis of various classes of \psword{s}.
By Proposition~\ref{prop:norm} an infinite word $w$ reduces to $\ssuc^\omega$ iff $\snorm{w} = \infty$,
and to $\spre^\omega$ iff $\pnorm{w} = \infty$.
The shaded non-empty intersection 
($\snorm{w} = \pnorm{w} = \infty$)
contains the counterexample word~$\pscountex$ mentioned above.
All words in this intersection are root-active (RA),
that is, every \mbox{$\mred$-reduct} can be reduced to a redex (at the root).
However, there are also other root-active words.
For example $\rasi = \ssuc \, \spre \, \ssuc^2 \, \spre^2 \, \ssuc^3 \, \spre^3 \cdots$
is a root-active word
which reduces to $\ssuc^\omega$ but not to $\spre^\omega$ (i.e., $\pnorm{\rasi} = 0 < \infty$ and $\snorm{\rasi} = \infty$).
The word $\rasi' = \ssuc\,\rasi$ (a reduct of $\rasi$) 
is not root-active but still not $\SNinf$, yet it reduces to $\ssuc^\omega$.
An example of a root-active term
which reduces only to itself (implying that $\snorm{\rasi}$ and $\pnorm{\rasi}$ are finite)
is $\rafn = (\spre \, \ssuc)^\omega$.
The dotted part consists of words with the property of infinitary strong normalization, 
normalizing to $\ssuc^\omega$, or $\spre^\omega$, respectively.
For instance $(\ssuc\,\ssuc\,\spre)^\omega$ is in the left dotted triangle.

The root-active words can be characterized as follows.
\begin{proposition}\label{prop:root-active:SPterms}
  An infinite \psword~$w$ is root-active if and only if $w$ is the concatenation
  of infinitely many finite `zero-words' $w_1,w_2,w_3,\ldots$, that is, 
  words $w_i$ with $\funap{\mrm{sum}}{w_i} = 0$. 
\end{proposition}
\begin{proof}
  The direction `$\Leftarrow$' is obvious.
  For `$\Rightarrow$' assume that $w$ is root active.
  Then $w$ admits a rewrite sequence containing infinitely many root steps.
  We label all $\ssuc$'s and $\spre$'s in $w$ by numbering them
  from left to right, so e.g.\ the labelled $w$ could be:
  $\ssuc_0 \, \ssuc_1 \, \spre_2 \, \ssuc_3 \, \spre_4 \, \spre_5 \, \cdots$
  Let $w_i$ be the prefix of $w$ of length $i$.
  For every $i \in \nat$, if there is a rewrite sequence $w \to^* \ssuc_iw'$ or $w \to^* \spre_iw'$ 
  for some $w'$, then $w_i$ must be a zero-word (as $w_i$ has been rewritten to $\wordemp$).
  Since $w$ is root-active, there are infinitely many $i \in \nat$
  such that $w \to^* \ssuc_i\cdots$ or $w \to^* \spre_i\cdots$.
  Thus there are indices $0 = i_0 < i_1 < i_2 < \cdots$ such that $w_{i_k}$ is a zero-word for every $k \in \nat$.
  For every $i \in \nat$, there exists a word $v_i$ such that $w_{i+1} = w_i v_i$,
  and since $w_{i}$ and $w_{i+1}$ are zero-words it follows that $v_i$ is a zero-word.
  The claim follows since $w = v_0 v_1 v_2\cdots$.
\end{proof}
As a consequence of this proposition,
an infinite \psword~$w$  is root-active 
if and only if
$\funap{\mrm{sum}}{w,n} = 0$ for infinitely many $n$,
and hence, if 
$ ((\liminf)_{n\to\infty} \abs{\funap{\mrm{sum}}{w,n}}) = 0 $.
%

\begin{corollary}
  For an infinite \psword~$w$ we have $\funap{\SNinf}{w}$ if and only if each value $\funap{\mrm{sum}}{w,n}$
  for $n=0,1,\dots$ occurs only finitely often. 
\end{corollary}
\begin{proof}
  For the `only if' direction, assume we have 
  $\funap{\mrm{sum}}{w,n_i} = k$ for $i = 1,2,\ldots$ such that $n_i < n_{i+1}$.
  For $i \ge 1$, let $w(i)$ denote the $i$-th letter of $w$.
  For every $i \in \nat$, we define $v_i = w(n_i + 1)w(n_i + 2) \cdots w(n_{i+1})$,
  that is, $v_i$ is the factor of $w$ starting at the $(n_i+1)$-th letter
  and ending at (including) the $n_{i+1}$-st letter.
  Then, for every $i\in\nat$, $v_i$ is a zero-word, i.e., $\funap{\mrm{sum}}{v_i} = 0$. 
  By Proposition~\ref{prop:root-active:SPterms} 
  we obtain that $v_1 v_2 \cdots$ is a root-active word.
  Hence $w$ is not $\SNinf$. 
  
  For the other direction, assume that $w$ is not $\SNinf$. 
  Then there exists $d\in\nat$ such that $w$ admits a rewrite sequence with
  infinitely many rewrite steps at depth $d$, see~\cite{kenn:klop:slee:vrie:1995}.
  It follows (take $d$ minimal) that there is a word $v$ with $w \mred v$ such that $v$ has a root-active suffix $v'$.
  By Proposition~\ref{prop:root-active:SPterms} the word $v'$ is again a concatenation of zero-words.
  Since $w \mred v$ is finite, apart from a finite prefix these zero-words are already present in $w$:
  $w = w' v_m v_{m+1} \cdots$ for some prefix $w'$ and $m \geq 1$. 
  For $i \ge m$ let $\apos_i$ be the depth of the displayed $v_i$ in $w$.
  Then $\funap{\mrm{sum}}{w,\apos_m} = \funap{\mrm{sum}}{w,\apos_{m+1}} = \cdots$. 
\end{proof}
It follows that
$\funap{\SNinf}{w}$ holds if and only if
$ ((\liminf)_{n\to\infty} \abs{\funap{\mrm{sum}}{w,n}}) = \infty $,
and hence,
if 
$\lim_{n \to \infty} \funap{\mrm{sum}}{w,n} 
  \in \{ \infty, -\infty \}$.
The `only if'-part follows since 
if $\funap{\SNinf}{w}$ holds,
then each value $\funap{\mrm{sum}}{w,n}$ for $n=0,1,\dots$ occurs only finitely often,
and hence for every $m\in\mathbb{Z}$, 
either $\funap{\mrm{sum}}{w,n}$ will eventually (for large enough $n$) stay above $m\in\nat$ or eventually stay below $m\in\nat$.

\section{A Counterexample to \texorpdfstring{$\UNinf$}{UNinfty} of the
  Infinitary \texorpdfstring{$\lambda\beta\eta$}{lambda-beta-eta}-Calculus}\label{sec:beta:eta}

We give a translation of the word 
$\pscountex = \spre^1\,\ssuc^2\,\spre^3\,\cdots$ 
from the previous section
into an infinite $\lambda$-term which then forms a counterexample to
the infinitary unique normal form property $\UNinf$ 
for $\ilbe$, the infinitary $\lambda\beta\eta$-calculus.
The infinitary $\lambda\beta\eta$-calculus~\cite{seve:vrie:2002,seve:vrie:2005}
is a well-known example 
of a weakly orthogonal higher-order term rewrite system.

The set $\lamter$ of (potentially) infinite $\lambda$-terms
is coinductively defined by:
\begin{align}
  M \coBNFis x \BNFor M M \BNFor \mylam{x}{M}  
  \tag{$\lamter$}
\end{align}
Here $\coBNFis$ is used to indicate that the grammar has to be interpreted coinductively, 
that is, instead of the least we take the greatest fixed point of the underlying functor.
Alternatively, infinite \mbox{$\lambda$-terms} can be defined in a similar vein 
as we defined infinite first-order terms in Section~\ref{sec:prelims}.
We let the set of variables be an uncountably infinite set.
This guarantees that no term can contain all variables.
The reason is that $\beta$-reduction may require $\alpha$-conversion and fresh names for binders.

The rewrite rules of $\ilbe$ are:
\begin{align}
  (\mylam{x}{M})N &\to \lsubst{M}{x}{N} \tag{$\beta$}\\
  \mylam{x}{M x} &\to M &&\text{if $x$ is not free in $M$} \tag{$\eta$}
\end{align}
where $\lsubst{M}{x}{N}$ denotes the result of substituting $N$ 
for all free occurrences of $x$ in $M$.
The $\ilbe$-calculus allows for two critical pairs\footnotemark[1]\footnotetext[1]{%
  We use the notation of infinitary $\lambda$-calculus,
  but we view the rule schemes ($\beta$) and ($\eta$) as rules of a second-order HRS,
  thereby obtaining a formal notion of critical pairs (\cite[Def.~11.6.10]{terese:2003}).
  Likewise, CRSs can be viewed as second-order HRSs.
}:
\begin{align*}
  M x \stackrel{\beta}{\redi} (\mylam{x}{M x} ) x \stackrel{\eta}{\to} M x
  &&
  \mylam{x}{\lsubst{M}{y}{x}} \stackrel{\beta}{\redi} \mylam{x}{(\mylam{y}{M}) x} \stackrel{\eta}{\to} \mylam{y}{M} 
\end{align*}
As we have that $\mylam{x}{\lsubst{M}{y}{x}}$ and $\mylam{y}{M}$ 
are equal modulo renaming of bound variables, 
both of these critical pairs are trivial. Hence $\ilbe$ is weakly orthogonal.

We translate infinite \psword{s} to $\lambda$-terms.
\begin{definition}\label{def:pslam}
  We define $\spstolam \funin {\{\spre,\ssuc\}}^\omega \to \lamter$
  by $\pstolam{w} = \pstolami{w}{0}$, for all $w\in{\{\spre,\ssuc\}}^\omega$,
  where $\pstolami{w}{i}$ is defined coinductively, for all $i \in \ints$, 
  as follows:
  \begin{align*}
  \pstolami{\spre w}{i}
  & = \pstolami{w}{i-1} \, x_{i}
  &
  \pstolami{\ssuc w}{i}
  & = \mylam{x_{i+1}}{\pstolami{w}{i+1}}
  \end{align*}
\end{definition}
\noindent
The translation of $\pscountex$ is the $\lambda$-term $\pstolam{\pscountex}$, 
displayed in the middle of Figure~\ref{fig:betaeta}.
\input{figs/betaeta}
This term has two normal forms (corresponding to $\ssuc^\omega$ and $\spre^\omega$),
as indicated in the figure.
\begin{remark}
  Note that the $\lambda$-term $\pstolam{\pscountex}$ 
  as well as one of its normal forms
  contains infinitely many bound variables.
  We remark that it is possible to define a translation
  from infinite \psword{s} to $\lambda$\nb-terms 
  such that $\pstolam{\pscountex}$ has as normal forms 
  $A = \mylam{x_1}{A}$ (using only one index for the abstractions) 
  and $B = B \, x_0$ (using only one free variable).
\end{remark}

While $\pstolam{\pscountex}$ cannot be generated from a finite $\lambda$\nb-term
(it has infinitely many free variables), 
the finite term $W W I$
where $W = \mylam{w f}{f (w w (\mylam{a b c}{f (a b c)}) x_0)}$ and $I = \mylam{a}{a}$
exhibits a similar behaviour, 
reducing both to $A = \mylam{x}{A}$ and $B = B x_0$.
This can be seen as follows:
Let $V_n = \mylam{v_1\ldots v_n}{(v_1 \ldots v_n)}$.
First note that
$W W I \red_\beta^{2}  I (W W (\mylam{a b c}{I (a b c)}) x_0) \red_\beta^{2}  W W V_3 x_0$.
Then we get:
\begin{align*}
    W W V_3 x_0 
    & \red_\beta^{2}
    V_3 (W W (\mylam{a b c}{V_3 (a b c)}) x_0) x_0
    \red_\beta^3
    \mylam{v_3}{W W V_5 x_0 x_0 v_3} \\
    & \red_\beta^6
    \mylam{v_3 v_5}{W W V_7 x_0 x_0 x_0 v_3 v_5} 
    \ired_\beta 
    \mylam{v_3 v_5 v_7 \ldots}
    =_\alpha A
    \\
    W W V_3 x_0 & \red_\eta^{2}  (W W I) x_0 \ired_{\beta\eta} B
\end{align*}
Note that the number of bound variables needed along
the reduction from $W W (\mylam{a}{a})$ to $A$ is unbounded, 
but that $A$ can be written using only a single one. 
We conjecture that it holds for every counterexample to $\UNinf$
in the infinitary $\lambda\beta\eta$\nb-cal\-cu\-lus
that during the rewrite process to one of the normal forms unboundedly many variables are needed.

The translation given in Definition~\ref{def:pslam}
lifts $\spre\ssuc \to \wordemp$ to $\beta$, 
and $\ssuc\spre \to \wordemp$ to $\eta$. 
\begin{lemma}\label{lem:pstolam}
  An application of the rule $\spre\ssuc \to \wordemp$ at depth $k$ 
  in an infinite \psword{} $w$ corresponds to a $\beta$-step 
  in $\ilbe$ at depth $k$ in $\pstolami{w}{i}$.
  Similarly so for the rule $\ssuc\spre \to \wordemp$ and the $\eta$-rule.
  These correspondences are indicated in the following diagrams: 
  \begin{center}
    \input{figs/ps2lam}
  \end{center}
\end{lemma}

The counterexample to the infinitary unique normal form property
$\UNinf$ for infinitary $\lambda\beta\eta$-calculus ($\ilbe$) establishes a striking contrast
to the situation for infinitary $\lambda\beta$-calculus ($\ilb$). 
In the latter, infinitary confluence breaks down, but infinitary normal forms stay unique. 
However, when the $\eta$-rule is added, and the infinitary perspective is 
maintained, then `everything' breaks down dramatically: not only infinitary 
confluence, but also unique infinitary normal forms.

The counterexample displays an interesting phenomenon from the point of 
view of the three main semantics of $\lambda\beta$-calculus, to wit, 
the B\"{o}hm Tree (BT), L\'{e}vy--Longo Tree (LLT), and Berarducci Tree (BeT) semantics. 
The middle term in Figure~\ref{fig:betaeta} 
has an infinite `spine', that is, 
a branch consisting of steps down or to the left. 
Such branches signal a term without head normal form in the BT semantics.
To the left and to the right, the terms are infinite weak head normal forms;
such terms are meaningful in the LLT and BeT semantics. 
Thus the counterexample also shows that L\'{e}vy--Longo trees and Berarducci trees for $\ilbe$ are not unique.
By contrast, B\"{o}hm trees for $\ilbe$ are unique,
and B\"ohm reduction can be employed to restore infinitary confluence and unique normal forms,
see~\cite{seve:vrie:2002};
there it has also been observed that 
L\'{e}vy--Longo trees and Berarducci trees are not unique for $\ilbe$.

From the perspective of combinatory reduction systems (CRSs, see~\cite{terese:2003})
the $\eta$-rule has many undesirable properties:
(i)~it is undecidable whether an infinite term is an $\eta$-redex,
since it is undecidable whether an infinite term contains a variable freely;
(ii)
single-step $\eta$-reduction is not lower semi-continuous: 
if $t$ $\eta$-reduces to $u$, then for a given $\epsilon > 0$ 
we cannot always find a $\delta > 0$ such that anything 
within $\delta$-distance of $t$ $\eta$-reduces to something
within $\epsilon$-distance of $u$;
(iii)~the $\eta$-rule is not fully-extended, 
and various existing results for orthogonal infinite CRSs require fully-extendedness, 
see~\cite{kete:simo:2009}.

\section{A Refinement of the Compression Lemma}\label{sec:compression}

As a preparation for Section~\ref{sec:confluence} we will prove the following lemma,
which is a refined version of the compression lemma in left-linear \iTRS{s}.
In its original formulation 
(e.g.\ see \cite[Theorem~12.7.1,\hspace*{2.5pt}p.\hspace*{1.5pt}689]{terese:2003}) 
the compression lemma states that every strongly convergent rewrite sequence from $s$ to $t$ in left-linear \iTRS{}s
can be compressed to a strongly convergent rewrite sequence from $s$ to $t$ of length 
at most $\omega$.
To see why left-linearity is a necessary condition, 
consider the following example of a non-left-linear TRS from~\cite{ders:kapl:plai:1991}:
\begin{align*}
  a \to g(a) && b \to g(b) && f(x,x) \to c
\end{align*}
In this TRS every rewrite sequence of the form $f(a,b) \ired f(g^\omega,g^\omega) \red c$ of length $\ge \omega + 1$
cannot be compressed to one of length $\omega$.

%

The refined version adds that compression can be carried out in such
a way that the minimal depth of steps stays the same.
This version can then be applied to show that also rewrite sequences 
that are not strongly convergent can be compressed.
We recall that a rewrite sequence of ordinal length $\alpha$ 
is strongly convergent if for each limit ordinal $\lambda \leq \alpha$ 
the depth of the contracted redexes tends to infinity.
As a consequence, a strongly convergent reduction can only contain
finitely many rewrite steps at every depth $d\in\nat$~\cite{kenn:klop:slee:vrie:1995}.

\begin{theorem}[Refined Compression Lemma]
  \label{thm:compression}
  Let $\atrs$ be a left-linear \iTRS. 
  Let $\aseq \funin s \to^\alpha_R t$ be a rewrite sequence, 
  $d$ the minimal depth of a step in $\aseq$, 
  and $n$ the number of steps at depth $d$ in $\aseq$.
  Then there exists a rewrite sequence $\aseq' \funin s \to^{\le \omega}_R t$
  in which all steps take place at depth $\ge d$, and where precisely
  $n$ steps contract redexes at depth~$d$.
\end{theorem}

\begin{proof}
  We proceed by transfinite induction on the ordinal length $\alpha$
  of rewrite sequences $\aseq \funin s \to^\alpha_R t$
  with $d$ the minimal depth of a step, 
  and $n$ the number of steps at depth $d$, in $\aseq$.
  
  In case that $\alpha = 0$ 
  nothing needs to be shown. 
  
  Suppose $\alpha$ is a successor ordinal.
  Then $\alpha = \beta + 1$ for some ordinal $\beta$, and
  $\aseq$ is of the form $s \to^{\beta} s' \to t$.
  Applying the induction hypothesis to $s \to^{\beta} s'$ 
  yields a rewrite sequence $s \to^{\gamma} s'$ of length $\gamma\le\omega$
  that contains the same number of steps at depth $d$, and no steps
  at depth less than $d$.
  
  If $\gamma < \omega$, then $s \to^\gamma s' \to t$ is a rewrite sequence of
  length $\gamma + 1 < \omega$, in which all steps take place at depth
  $\ge d$ and precisely $n$ steps 
  at depth $d$.
  
  If $\gamma = \omega$, we obtain a rewrite sequence of the form 
  $s = s_0 \to s_1 \to \cdots \to^\omega s_\omega \to t$.
  Let $\ell \to r \in R$ be the rule applied in the final step $s_\omega \to t$,
  that is, $s_\omega = \contextfill{\acontext}{\subst{\asubst}{\ell}} \to \contextfill{\acontext}{\subst{\asubst}{r}} = t$ 
  for some context $\acontext$ and substitution $\asubst$.
  Moreover, let $d_h$ be the depth of the hole in $\acontext$, and $d_p$ the depth of the pattern of $\ell$.
  Since the reduction $s_0 \to^\omega s_\omega$ is strongly convergent, 
  there exists $n \in \nat$ such that all rewrite steps in $\xi : s_n \to^\omega s_\omega$ have depth $> d_h + d_p $,
  and hence are below the pattern of the redex contracted in the last step
  $s_\omega \to t$.
  As a consequence of this fact and left-linearity, there exists a context $\bcontext$ and a substitution $\bsubst$
  such that $s_n = \contextfill{\bcontext}{\subst{\bsubst}{\ell}}$.
  Since the rewrite sequence 
  $\xi : s_n = \contextfill{\bcontext}{\subst{\bsubst}{\ell}} 
     \to^\omega \contextfill{\acontext}{\subst{\asubst}{\ell}} = s_\omega$
  consists only of steps at depth $> d_h + d_p$,
  it follows that:
  \begin{itemize}
  \item there exists a rewrite sequence
        $\vartheta : \contextfill{\bcontext}{\contexthole} \to^{\le \omega} \contextfill{\acontext}{\contexthole}$
        at depth $> d_h + d_p$, and
  \item there exist rewrite sequences
        $\vartheta_x : \funap{\bsubst}{x} \to^{\le \omega} \funap{\asubst}{x}$ for all $x \in \vars{\ell}$.
  \end{itemize}
  We now prepend the final step $s_\omega \to t$ to $s_n$, that is:
  $s_n = \contextfill{\bcontext}{\subst{\bsubst}{\ell}} \to \contextfill{\bcontext}{\subst{\bsubst}{r}}$.
  Even if the term $r$ is infinite, this creates at most $\omega$-many copies of subterms $\funap{\bsubst}{x}$ 
  with reduction sequences $\vartheta_x : \funap{\bsubst}{x} \to^{\le \omega} \funap{\asubst}{x}$ of length $\le \omega$.
  Since the rewrite sequences $\vartheta$ and $\vartheta_x$ for $x \in \vars{\ell}$ are in disjoint (parallel) subterms,
  there exists an interleaving 
  $\contextfill{\bcontext}{\subst{\bsubst}{r}} \to^{\le \omega} \contextfill{\acontext}{\subst{\asubst}{r}}$
  of length at most $\omega$ (the idea is similar to establishing countability of $\omega^2$ by dovetailing).
  We obtain a rewrite sequence
  $\aseq' \funin s \to^{\le \omega} t$,
  since
  $s \to^n s_n = \contextfill{\bcontext}{\subst{\bsubst}{\ell}}
   \to \contextfill{\bcontext}{\subst{\bsubst}{r}} \to^{\le \omega} \contextfill{\acontext}{\subst{\asubst}{r}} = t$.

  It remains to be shown that $\aseq'$ contains only steps at depth $\ge d$, 
  and that it has the same number of steps as the original sequence $\aseq$ 
  at depth $d$.
  This follows from the induction hypothesis and
  the fact that all steps in $s_n \to^\omega s_\omega$ 
  have depth $> d_h + d_p$ and thus also all steps of the interleaving
  $\contextfill{\bcontext}{\subst{\bsubst}{r}} \to^{\le \omega} \contextfill{\acontext}{\subst{\asubst}{r}}$
  have depth $> d_h + d_p - d_p = d_h \ge d$ 
  (the application of $\ell \to r$ can lift steps 
   by at most the pattern depth $d_p$ of $\ell$).
  
  \input{figs/compression}
  
  Finally, suppose that $\alpha$ is a limit ordinal $> \omega$.
  We refer to Figure~\ref{fig:compression} for a sketch of the proof.
  Since $\aseq$ is strongly convergent, only a finite number of steps
  take place at depth $d$. Hence there exists $\beta < \alpha$ such that
  $s_{\beta}$ is the target of the last step at depth $d$ in $\aseq$.
  We have $s \to^\beta s_\beta \to^{\le \alpha} t$ 
  and all rewrite steps in $s_\beta \to^{\le \alpha} t$ are at depth $> d$.
  By induction hypothesis there exists 
  a rewrite sequence $\bseq \funin s \to^{\le \omega} s_\beta$
  containing an equal amount of steps at depth $d$ as $s \to^\beta s_\beta$.
  Consider the last step of depth $d$ in $\bseq$\,. 
  This step has a finite index $n < \omega$.
  Thus we have $s \to^* s_{n} \to^{\le \alpha} t$, 
  and all steps in $s_n \to^{\le \alpha} t$ are at depth $> d$. 
  By successively applying this argument to $s_n \to^{\le \alpha} t$ 
  we construct finite initial segments $s \to^* s_n$ 
  with strictly increasing minimal rewrite depth $d$. 
  Concatenating these finite initial segments 
  yields a reduction $s \to^{\le \omega} t$
  containing as many steps at depth $d$ as the original sequence.
\end{proof}

With this refined compression lemma at hand, we now prove that
also divergent rewrite sequences can be compressed to length less than or equal
to $\omega$. 

\begin{corollary}
  \label{cor:comp:div:seqs}
  Let $\atrs$ be a left-linear \iTRS. 
  For every divergent rewrite sequence $\aseq \funin s \to^\alpha_R$ of length $\alpha$
  there exists a divergent rewrite sequence
  $\aseq' \funin s \to^{\le\omega}_R$ of length at most $\omega$.
\end{corollary}

\begin{proof}
  Let $\aseq \funin s \to^\alpha_R$ be a divergent rewrite sequence.
  Then there exist $k\in\nat$ such that infinitely many steps in $\aseq$
  take place at depth $k$.
  Let $d$ be the minimum of all numbers $k$ with that property.
  Let $\beta$ be the index of the last step above depth $d$ in 
  $\aseq$.
  Then  $\aseq$ can be written as
  $\aseq \funin s \to^\beta s_\beta \to^{\le \alpha}$,
  where
  $s_\beta \to^{\le \alpha}$ consists only of steps at depth $\ge d$,
  among which there are infinitely many steps at depth $d$.
  Now by Theorem~\ref{thm:compression} the rewrite sequence
  $s \to^\beta s_\beta$ can be compressed to a rewrite sequence
  $s \to^{\le \omega} s_\beta$.
  Let $n$ be the index of the last step of depth $\le d$ in the rewrite sequence $s \to^{\le \omega} s_\beta$.
  Then $s \to^* s_n \to^{\le \omega} s_\beta \to^{\le \alpha}$,
  and $s_n \to^{\le \omega} s_\beta \to^{\le \alpha}$ contains only steps at depth $\ge d$.
  Thus all steps with depth less than $d$ take place in the finite prefix
  $s \to^* s_n$.

  Now consider the rewrite sequence 
  $\aseq_1 \funin s_n \relcomp{\to^{\le \omega}}{\to^{\le \alpha}}$,
  say $\aseq_1 \funin s_n \to^\gamma$ for short,
  containing infinitely many steps at depth $d$.
  Let $\gamma'$ be the index of the first step at depth $d$ in $\aseq_1$.
  Then $\aseq_1 \funin s_n \to^{\gamma'} u \to^{\le \gamma}$ for some term $u$
  and $s_n \to^{\gamma'} u$ can be compressed to $s_n \to^{\le \omega} u$ 
  containing exactly one step at depth $d$.
  Now let $m$ be the index of this step, then $s_n \to^m u' \to^{\le \omega} u \to^{\le \gamma}$
  where $s_n \to^m u'$ contains one step at depth $d$.
  Repeatedly applying this construction to $u' \to^{\le \omega} u \to^{\le \gamma}$
  we obtain a rewrite sequence 
  $\aseq' \funin s \to^* s_n \to^* u' \to^* u'' \to \cdots$ 
  that contains infinitely many steps at depth $d$,
  and hence is divergent.
\end{proof}

\begin{remark}
  A slightly weaker version of Theorem~\ref{thm:compression} and 
  Corollary~\ref{cor:comp:div:seqs}, due to the second author,
  can be found in \cite{zant:2008} (see Lemma~3 and Theorem~4 there).   
  The weaker version of Theorem~\ref{thm:compression} 
  states the following:
  Every strongly convergent rewrite sequence $\aseq \funin s \to^\alpha_R t$
  with $d$ the minimal depth of its steps can be compressed 
  into a rewrite sequence $\aseq' \funin s \to^{\le\omega}_R t$ of length less or equal to $\omega$ 
  with \emph{at least as many} 
  (instead of \emph{precisely as many as} in Theorem~\ref{thm:compression})  
  steps as $\aseq$ at (minimal) depth $d$.

  We note that very closely related statements 
  have been formulated for infinitary combinatory reduction systems 
  in~\cite{kete:2008}, see Theorem~2.7 and Lemma 5.2 ibid.
\end{remark}

\section{Infinitary Confluence}\label{sec:confluence}

In Section~\ref{sec:pscountex} we have seen that the property $\UNinf$ 
fails for weakly orthogonal \iTRS{s} when collapsing rules are present,
and hence also $\CRinf$. 
Now we show that \woTRS{s} without collapsing 
rules are infinitary confluent ($\CRinf$),
and as a consequence also have the property $\UNinf$.

We adapt the projection of parallel steps in weakly orthogonal TRS{s}
from~\cite[Section~8.8.4.]{terese:2003} to infinite terms.
The basic idea is to orthogonalize the parallel steps,
and then project the orthogonalized steps.
The orthogonalization uses that overlapping redexes 
have the same effect and hence can be replaced by each other.
In case of overlaps we replace the outermost redex by the innermost one.
This is possible since
the maximal nesting depth of the union of two infinite parallel steps is at most 2,
that is, there can not be infinite chains of overlapping nested redexes
in such a union (see Example~\ref{ex:chain}).
For a treatment of infinitary multi-steps where such chains can occur,
we refer to Section~\ref{sec:multi}.
See further~\cite[Proposition~8.8.23]{terese:2003} for orthogonalization 
in the finitary case.

The ordinary notion of a redex, that is, an instance of a left-hand side
of a rule, does not suffice for the analysis of non-orthogonal
TRS{s}~\cite[Chapter~8]{terese:2003}. To see this, suppose we are given 
rules having left-hand sides $f(a,x)$, $f(x,a)$, and $a$.
What is the result of contracting the redex $f(a,a)$?
Since the term $f(a,a)$ is an instance of both the first
and the second left-hand side, that result will in general
depend on which of the two rules is applied.
Although in weakly orthogonal systems that result is unique,
the notion of overlap of redexes is problematic.
Are the three redexes in $f(a,a)$ non-overlapping?
Each \emph{pair} of redexes can be considered to be
non-overlapping, e.g., $\underline{f}(a,a)$ and $f(a,\underline{a})$
are non-overlapping redexes when the latter is seen as an
instance of the first left-hand side.
However, it is not possible to have $3$ non-overlapping redex occurrences in the term $f(a,a)$.
These observations motivate the following refined notion of redex.

\begin{definition}
  Let $\atrs$ be an iTRS, and $\atrm \in \iter{\asig}$ a term.

  A \emph{redex} in $\atrm$ is
  a pair consisting of a position $\apos$ and a
  rule $\rulstr{\alhs}{\arhs}$, such that $\trmat{\atrm}{\apos} = \subap{\sigma}{\alhs}$
  for some substitution $\sigma$.
  We call $p$ and $\rulstr{\alhs}{\arhs}$ the \emph{root} and \emph{rule} of the redex, respectively.
  The pattern of a redex $\pair{\apos}{\rulstr{\alhs}{\arhs}}$ is the set of 
  all positions $\possco{\apos}{\bpos}$ such that $\symat{\alhs}{\bpos}$ is a function symbol.

  Two sets of positions are \emph{overlapping} if they have a non-empty intersection.
  For redexes $u$ and $v$ in $t$ we say that \emph{$u$ and $v$ overlap}, denoted by $u \poverlap v$,
  if the patterns of $u$ and $v$ overlap.
  A set $U$ of redexes is called \emph{non-overlapping} if, for all $u,v\in U$ with $u\neq v$,
  $u$ does not overlap with $v$.  

  A \emph{multi-redex} in a term $t \in \iter{\asig}$ is a set of non-overlapping redexes in $t$.
\end{definition}

For a thorough study of developments we refer to \cite[Sec.~4.5.2]{terese:2003} and \cite{oost:1997}.
Here, we introduce developments in weakly orthogonal systems via labelling (underlining):

\begin{definition}\label{def:dev}
  Let $\atrs = \pair{\asig}{\asetofrules}$ be a weakly orthogonal iTRS.
  For symbols $f \in \asig$ and $\rho \in R$ we write $f^{\rho}$ for $f$ labelled with $\rho$.
  For labelled terms $t$, we write $\floor{t}$ to denote the term obtained from $t$ by dropping all labels.

  We define the \iTRS{} $\atrs^{\dev} = \pair{\asig^{\dev}}{\asetofrules^{\dev}}$ 
  where $\asig^{\dev} = \asig \cup \{f^\rho \where f \in \asig , \rho \in \asetofrules \}$
  and $\asetofrules^{\dev}$ consists of all rules 
  $\ell^{\rho} \to r$ for $\rho : \ell \to r \in R$
  where $\ell^{\rho}$ is the term obtained from $\ell$ by labelling the root-symbol of $\ell$ with $\rho$.

  Let $t, t' \in \iter{\asig}$ be terms, and $U$ a multi-redex in $t$.
  Let $t^U$ be the term obtained from $t$ by labelling
  for each redex $\pair{\apos}{\rho} \in U$ the symbol at position $\apos$ in $t$ with $\rho$.
  A \emph{development of $U$} in $t$ is a rewrite sequence 
  $ t \ired_{R} t'$ (in $\atrs$) that can be lifted to a reduction 
  $ t^U \ired_{R^{\dev}} t'' $ (in $\atrs^{\dev}$)
  such that $\floor{t''} = t'$, that is, $t'$ arises from $t''$ by dropping all labels.
  The development is called \emph{complete} if $t' \equiv t''$.
  A \emph{multi-step with respect to $U$}
  is a step ${t}\arsdev_U{t'}$ such that there exists a reduction $t^U \ired_{\atrs^{\dev}} t'$.
\end{definition}


\begin{remark}
  Let $t$ be a term, $U$ a multi-redex in $t$ and $t^U$ as in Definition~\ref{def:dev}.
  Observe that every term $s$ with $t^U \mred_{\atrs^{\dev}} s$
  has the property that every symbol occurrence labelled with a rule $\rho: \ell \to r \in R$ in $s$
  is a redex occurrence with respect to $\ell^{\rho} \to r \in \atrs^{\dev}$.
  The reason is that $U$ is a set of non-overlapping redexes,
  and $\atrs^{\dev}$ is an orthogonal iTRS.
  Therefore redex occurrences stay redex occurrences until they are contracted.
\end{remark}

Every complete development of a multi-redex $U$ ends in the same term, see~\cite{terese:2003}.
In non-collapsing, weakly orthogonal \iTRS{s},
every multi-redex $U$ has a complete development.
Multi-steps arise from complete developments, and are
uniquely determined by their starting term and a selection of redex occurrences.

\begin{definition}
  Let $\atrs$ be an iTRS, $t \in \iter{\asig}$ a term, and
  let $U$ and $V$ be sets of redexes in $t$.
  We call $U$ and $V$ \emph{orthogonal (to each other)}
  if $U \cup V$ is a multi-redex.
\end{definition}

\begin{definition}\label{def:oproj}
  Let $\atrs$ be a non-collapsing, \woTRS, and let $U$ and $V$ be orthogonal sets of redexes in a term $t$.
  For multi-steps $\dstepu{\astep}{t}{t'}{U}$ and $\dstepu{\bstep}{t}{t''}{V}$
  with respect to $U$ and $V$
  we define the projection $\project{\astep}{\bstep}$ as the 
  multi-step ${t''}\arsdev_{U'}{s}$ with respect to the set of residuals $U' = \project{U}{\bstep}$
  as defined in~\cite{terese:2003}.%
    \footnote{We refer to Def.$\,$12.5.3 in \cite{terese:2003}, 
              and note that the definition not only applies in orthogonal \iTRS{s}, 
              but also to every non-overlapping set $U$ of redexes versus a multistep $\phi$ 
              with respect to a redex set $V$ that is orthogonal to~$U$.}
  In the sequel, we sometimes write $\arsdev$ for the multi-step relation, 
  suppressing the set of redexes $U$ that induces the multi-step $\arsdev_{U}$.
\end{definition}

\begin{definition}\label{def:orthogonalization}
  An \emph{orthogonalization} 
  of a pair $\pair{\astep}{\bstep}$ of multi-steps
  $\dstepu{\astep}{s}{t_1}{U}$ and $\dstepu{\bstep}{s}{t_2}{V}$
  with respect to sets $U$ and $V$ of redexes in $s$
  is a pair $\pair{\astep'}{\bstep'}$ 
  of multi-steps $\dstepu{\astep'}{s}{t_1}{U'}$ and $\dstepu{\bstep'}{s}{t_2}{V'}$
  with respect to orthogonal sets $U'$ and $V'$ of redexes in $s$.
\end{definition}

A parallel step $\pstep{\astep}{s}{t}$ is a multi-step $\dstepu{\astep}{s}{t}{U}$ with respect to a set $U$ of parallel redexes,
that is, redexes at pairwise disjoint positions.

\begin{proposition}\label{prop:portho}
  Let $\pstep{\astep}{s}{t_1}$ and $\pstep{\bstep}{s}{t_2}$ be parallel steps in a \woTRS{}.
  Then there exists an orthogonalization $\pair{\astep'}{\bstep'}$ of $\astep$ and $\bstep$
  with the special property that $\pstep{\astep'}{s}{t_1}$ and $\pstep{\bstep'}{s}{t_2}$.
\end{proposition}

\begin{proof}
  Let $U$ be the set of parallel redex occurrences contracted in $\pstep{\astep}{s}{t_1}$,
  and $V$ the set of parallel redex occurrences contracted in $\pstep{\astep}{s}{t_1}$.
  In case of overlaps between $U$ and $V$, 
  then for every overlap we replace the outermost redex by the innermost one
  (if there are multiple inner redexes overlapping, 
  then we choose the left-most among the top-most redexes).
  If there are two redexes at the same position but with respect to different rules,
  then we replace the redex in $V$ with the one in $U$.
  See also Figure~\ref{fig:orthogonalization:parallel}.
  \begin{figure}[hpt!]
  \begin{center}
    \scalebox{.7}{\input{figs/parallel}}
  \end{center}\vspace{-2ex}
  \caption{\textit{Orthogonalization of parallel steps; the arrow indicates replacement.}}
  \label{fig:orthogonalization:parallel}
  \end{figure}
\end{proof}

\begin{definition}\label{def:pproject}
  Let $\pstep{\astep}{s}{t_1}$, $\pstep{\bstep}{s}{t_2}$ be parallel steps in a \woTRS{}.
  The \emph{weakly orthogonal projection $\project{\astep}{\bstep}$ of $\astep$ over $\bstep$}
  is defined as the orthogonal projection $\project{\astep'}{\bstep'}$
  where $\pair{\astep'}{\bstep'}$ is the orthogonalization of $\astep$ and $\bstep$ 
  given in the proof of Proposition~\ref{prop:portho}. 
\end{definition}

\begin{remark}
  The weakly orthogonal projection does not give rise to a residual system
  in the sense of~\cite{terese:2003}.
  The projection fulfils the three identities
  $\project{\astep}{\astep} \approx \unit$,
  $\project{\astep}{\unit} \approx \astep$, and
  $\project{\unit}{\astep} \approx \unit$,
  but not the \emph{cube identity}
  $\project{(\project{\astep}{\bstep})}{(\project{\cstep}{\bstep})} \approx
    \project{(\project{\astep}{\cstep})}{(\project{\bstep}{\cstep})}$.
\end{remark}

\begin{lemma}\label{lem:proj:lift}
  Let $\pstep{\astep}{s}{t_1}$, $\pstep{\bstep}{s}{t_2}$ be parallel steps in a \woTRS{} $\atrs$.
  Let $d_\astep$ and $d_\bstep$ be the minimal depth of a step in $\astep$ and $\bstep$, respectively.
  Then the minimal depth of the weakly orthogonal projections
  $\project{\astep}{\bstep}$ and $\project{\bstep}{\astep}$
  is greater or equal $\bfunap{\min}{d_\astep}{d_\bstep}$.
  If $\atrs$ contains no collapsing rules
  then the minimal depth of $\project{\astep}{\bstep}$ and $\project{\bstep}{\astep}$ 
  is greater or equal $\bfunap{\min}{d_\astep}{d_\bstep+1}$
  and $\bfunap{\min}{d_\bstep}{d_\astep+1}$, respectively.
\end{lemma}

\begin{proof}
  Immediate from the definition of the orthogonalization (for overlaps the innermost redex is chosen)
  and the fact that in the orthogonal projection
  a non-collapsing rule applied at depth $d$
  can lift nested redexes at most to depth $d+1$ (but not above).
\end{proof}

\begin{lemma}\label{lem:pdiamond}
  Parallel steps in a \woTRS{} have the diamond property.
\end{lemma}
\begin{proof}
  Consequence of Lemma~\ref{prop:portho} and the usual orthogonal projection, see Definition~\ref{def:pproject}.
\end{proof}

\begin{lemma}[Parallel Moves Lemma]\label{lem:pml}
  Let $\atrs$ be a \woTRS{}, 
  $\aseq \funin s \to^\alpha t_1$ a rewrite sequence,
  and $\pstep{\astep}{s}{t_2}$ a parallel rewrite step.
  Let $d_\aseq$ and $d_\astep$ be the minimal depth of a step in $\aseq$
  and $\astep$, respectively.
  Then there exist a term $u$, a rewrite sequence $\bseq \funin t_2 \to^{\le \omega} u$
  and a parallel step $\pstep{\bstep}{t_1}{u}$
  such that the minimal depth of the rewrite steps in $\bseq$ and $\bstep$ 
  is $\bfunap{\min}{d_\aseq}{d_\bseq}$;
  see Figure~\ref{fig:pml:collapse} (left).

  If additionally $\atrs$ contains no collapsing rules, 
  then the minimal depth of a step in $\bseq$ and $\bstep$
  is $\bfunap{\min}{d_\aseq}{d_\bseq+1}$ and $\bfunap{\min}{d_\bseq}{d_\aseq+1}$, respectively. 
  See also Figure~\ref{fig:pml} (right).
\end{lemma}
\begin{figure}[h!]
  \begin{minipage}[b]{.45\textwidth}
  \input{figs/strip-collapse}
  \end{minipage}
  %
  %
  \begin{minipage}[b]{.45\textwidth}
  \input{figs/strip}
  \end{minipage}
  \caption{\textit{Parallel Moves Lemma; with (left) and without (right) collapsing rules.}}
  \label{fig:pml:collapse}
  \label{fig:pml}
\end{figure}

\begin{proof}
  By compression we may assume $\alpha \le \omega$ in $\aseq \funin s \to^{\le \omega} t_1$ 
  (note that, the minimal depth $d$ is preserved by compression).
  Let $\aseq \funin s \equiv s_0 \to s_1 \to s_2 \to \cdots$,
  and define $\astep_0 = \astep$.
  Furthermore, let $\seqpref{\aseq}{n}$ denote the prefix of $\aseq$ of length $n$, 
  that is, $s_0 \to \cdots \to s_n$
  and let $\seqsuf{\aseq}{n}$ denote the suffix $s_n \to s_{n+1} \to \cdots$ of $\aseq$.
  We employ the projection of parallel steps to
  close the elementary diagrams with top $s_n \to s_{n+1}$ and left $\pstep{\astep_n}{s_n}{s_n'}$,
  that is, 
  we construct the projections $\astep_{n+1} = \project{\astep_n}{(s_n \to s_{n+1})}$ (right)
  and $\project{(s_n \to s_{n+1})}{\astep_n}$ (bottom).
  Then by induction on $n$ using Lemma~\ref{lem:proj:lift} there exists
  for every $1 \le n \le \alpha$
  a term $s_n'$, and parallel steps $\pstep{\astep_n}{s_n}{s_n'}$ and $s_{n-1}' \pred s_n'$.
  See Figure~\ref{fig:pml:proof} for an overview.
  \input{figs/strip-proof}

  We show that the rewrite sequence constructed at the bottom $s_0' \pred s_1' \pred \cdots$
  of Figure~\ref{fig:pml:proof} is strongly convergent,
  and that the sequence of parallel steps $\astep_0,\astep_1,\astep_2,\ldots$ has a limit
  (which is itself a parallel steps and therefore strongly convergent).

  Let $d \in \nat$ be arbitrary.
  By strong convergence of $\aseq$ there exists $n_0 \in \nat$ such that
  all steps in $\seqsuf{\aseq}{n_0}$ are at depth $\ge d$.
  Since $\astep_{n_0}$ is a parallel step there are only finitely many
  redexes $\astep_{n_0,<d} \subseteq \astep_{n_0}$ in $\astep_{n_0}$ rooted above depth $d$.
  By projection of $\astep_{n_0}$ along $\seqsuf{\aseq}{n_0}$
  no fresh redexes above depth $d$ can be created.
  The steps in $\astep_{n_0,<d}$ may be cancelled out due to overlaps,
  nevertheless, for all $m \ge n_0$ the set of steps above depth $d$ in $\astep_m$
  is a subset of $\astep_{n_0,<d}$.

  Let $p$ be the maximal depth of a left-hand side of a rule applied in $\astep_{n_0,<d}$.
  By strong convergence of $\aseq$ there exists $m_0 \ge n_0 \in \nat$ such that
  all steps in $\seqsuf{\aseq}{m_0}$ are at depth $\ge d+p$.
  As a consequence the steps $\bstep$ in $\astep_{m_0}$ rooted above depth $d$
  will stay fixed throughout the remainder of the projection.
  Then for all $m \ge m_0$ the parallel step $\astep_m$
  can be split into $\astep_m = s_m \pred_\bstep s_m'' \pred_{\astep_{m,\ge d}} s_m'$
  where $\astep_{m,\ge d}$ consists of the steps of $\astep_m$ at depth $\ge d$.
  Since $d$ was arbitrary, it follows that projection of $\astep$ over $\aseq$ has a limit.
  Moreover the steps of the projection of $\seqsuf{\aseq}{m_0}$ over $\astep_{m_0}$
  are at depth $\ge d + p - p = d$ since rules with pattern depth $\le p$ 
  can lift steps by at most by $p$.
  Again, since $d$ was arbitrary, it follows that the projection of $\aseq$ 
  over $\astep$ is strongly convergent.

  Finally, both constructed rewrite sequences (bottom and right) 
  converge towards the same limit $u$
  since all terms $\{s_m', s_m'' \where m \ge m_0\}$ coincide up to depth $d-1$
  (the terms $\{s_m \where m \ge m_0\}$ coincide up to depth $d + p -1$ 
  and the lifting effect of the steps $\astep_m$ is limited by $p$).
\end{proof}

\begin{theorem}\label{thm:cr}
  Every \woTRS{} without collapsing rules is infinitary confluent.
\end{theorem}

\input{figs/confluence}
\begin{proof}
  An overview of the proof is given in Figure~\ref{fig:confluence}.
  Let $\aseq \funin s \to^\alpha t_1$ and $\bseq \funin s \to^\beta t_2$ be two rewrite sequences.
  By compression we may assume $\alpha \le \omega$ and $\beta \le \omega$.
  Let $d$ be the minimal depth of any rewrite step in $\aseq$ and $\bseq$.
  Then $\aseq$ and $\bseq$ are of the form
  $\aseq \funin s \to^* s_1 \to^{\le\omega} t_1$
  and
  $\bseq \funin s \to^* s_2 \to^{\le\omega} t_2$
  such that all steps in $s_1 \to^{\le\omega} t_1$ and $s_2 \to^{\le\omega} t_2$ at depth $> d$.
  
  Then $s \to^* s_1$ and $s \to^* s_2$ can be joined by finitary diagram completion employing 
  the diamond property for parallel steps (Lemma~\ref{lem:pdiamond}).
  If follows that there exists a term $s'$ and finite sequences of 
  (possibly infinite) parallel steps $s_1 \pred^* s'$ and $s_2 \pred^* s'$
  all steps of which are at depth $\ge d$ (Lemma~\ref{lem:proj:lift}).
  We project 
  $s_1 \to^{\le\omega} t_1$ over $s_1 \pred^* s'$, 
  $s_2 \to^{\le\omega} t_2$ over $s_2 \pred^* s'$ 
  by repeated application of the Lemma~\ref{lem:pml},
  obtaining rewrite sequences 
  $t_1 \ired t_1'$,
  $s' \ired t_1'$,
  $t_2 \ired t_2'$, and
  $s' \ired t_2'$ with depth $\ge d$, $> d$, $\ge d$, and $> d$, respectively.
  As a consequence we have $t_1'$, $s'$ and $t_2'$ coincide up to (including) depth $d$.
  Recursively applying the construction to the rewrite sequences $s' \ired t_1'$ and $s' \ired t_2'$
  yields strongly convergent rewrite sequences 
  $t_2 \ired t_2' \ired t_2'' \ired \cdots$ and $t_1 \ired t_1' \ired t_1'' \ired \cdots$
  where the terms $t_1^{(n)}$ and $t_2^{(n)}$ coincide up to depth $d + n -1$.
  Thus these rewrite sequences converge towards the same limit $u$.
\end{proof}

We consider an example to illustrate that the absence of collapsing rules is a necessary condition 
for Theorem~\ref{thm:cr}.
\begin{example}\label{ex:collapse}
  Let $\atrs$ be an \iTRS\ over the signature $\{f,a,b\}$ consisting of the collapsing rule:
    $\bfunap{f}{x}{y} \to x$
  Then, using a self-explaining recursive notation, 
  the term $s = \bfunap{f}{\bfunap{f}{s}{b}}{a}$ rewrites in $\omega$ many steps to 
  $t_1 = \bfunap{f}{t_1}{a}$ as well as $t_2 = \bfunap{f}{t_2}{b}$ which have no common reduct.
  The \iTRS{} $\atrs$ is weakly orthogonal (even orthogonal) but not confluent.
  The same phenomenon occurs in the infinitary version of combinatory logic, 
  due to the rule $Kxy\to x$.
\end{example}

\section{The Diamond and Triangle Properties for Multi-Steps}\label{sec:multi}

Orthogonal rewrite systems possess a rich theory of
residuals originating with Church and Rosser's seminal 
paper~\cite{Chur:Ross:36}, establishing confluence
of the (orthogonal) $\lambda$-calculus with $\beta$-reduction by
means of residuals.
Intuitively, the \emph{residual} of a redex after another
is that what remains of the former after contracting the latter
(see~\cite[Section~8.7]{terese:2003} for an abstract development and references). 
Orthogonality of a rewrite system guarantees that distinct redexes
are mutually non-overlapping, giving rise to residuals in a natural
way and entailing the well-definedness of multi-redexes, in the sense
that the result of contracting a multi-redex is independent of the 
order in which the redexes in it are contracted.
As a consequence, the multi-step rewrite relation $\aarsdev$ has, 
unlike the single step rewrite relation $\aars$, good rewrite properties, 
in particular the diamond and triangle properties,
justifying the central role they play in the theory of orthogonality.
\begin{definition} \label{def:diamond:triangle}
  A rewrite relation $\aars$ on $\aset$ is said to have:
  \begin{enumerate}
  \item
    the \emph{diamond} property if 
    $\relle{\relsco{\aarsinv}{\aars}
          }{\relsco{\aars}{\aarsinv}
          }$\,;
  \item
    the \emph{angle} property if there exists a function $\atrmful$
    from $\aset$ to $\aset$, such that for all $\setin{\aobj,\bobj}{\aset}$,
    $\arsa{\aobj}{\bobj}$ implies $\arsa{\bobj}{\trmfula{\aobj}}$\,;   
  \item   
    the \emph{triangle} property if there exists a function $\atrmful$
    from $\aset$ to $\aset$, such that
    $\arsa{\aobj}{\trmfula{\aobj}}$ for all $\setin{\aobj}{\aset}$,
    and for all $\setin{\aobj,\bobj}{\aset}$,
    $\arsa{\aobj}{\bobj}$ implies $\arsa{\bobj}{\trmfula{\aobj}}$\,.
  \end{enumerate}
\end{definition}

\noindent Note that the triangle property entails the angle property which in turn entails the diamond property.
We show that multi-steps in weakly orthogonal iTRSs
without collapsing rules have both the diamond and the triangle 
properties. We proceed by first illustrating the difficulties caused
by the transitions from the finitary to the infinitary case and from
the orthogonal to the weakly orthogonal case, and then
showing how these difficulties can be overcome.
\begin{example} \label{exa:infinitary:collapse}
  Consider the orthogonal TRS with rules $\rulstr{\syma{\atrmvar}}{\atrmvar}$ and
  $\rulstr{\symb{\atrmvar}}{\atrmvar}$ collapsing $\asym$ and $\bsym$, respectively.
  Contracting the multi-redexes consisting of all $\asym$-redexes respectively 
  all $\bsym$-redexes in the infinite term $\atrm \isdefd \syma{\symb{\atrm}}$ yields the 
  infinite terms $\atrm' \isdefd \symb{\atrm'}$ and $\atrm'' \isdefd \syma{\atrm''}$,
  which do not have a common reduct.
\end{example}
In the example the problem is not so much that it is not clear what the residuals are.
They are the multi-redex consisting of all $\bsym$-redexes in $\atrm'$
and the multi-redex consisting of all $\asym$-redexes in $\atrm''$.
The problem is rather that contracting these multi-redexes would lead to
an infinite collapse in both cases and thereby an undefined result and 
common reduct.
In order to bar such examples, it suffices to exclude collapsing rules,
guaranteeing that contracting multi-redexes is productive.
\begin{example} \label{exa:partial:overlap}
  Consider the weakly orthogonal TRS with rule $\rulstr{\syma{\syma{\atrmvar}}}{\atrmvar}$.
  What should the residual be of the outermost redex in the term $\syma{\syma{\syma{\asymzer}}}$
  after the step contracting the innermost redex?
\end{example}
The problem illustrated by the example is that if redexes are partially
overlapping, here on the middle $\asym$-symbol, there is no natural 
notion of residual. Still, because of weak orthogonality we know that partially
overlapping redexes are equivalent in the sense that contracting either 
of them yields the same result. The idea is then to replace the redexes
in a multi-redex by equivalent ones in such a way that orthogonality
is restored. Such an \emph{orthogonalization} was developed for 
the finitary case in~\cite[Section~8.8.4]{terese:2003}, replacing
inside--out each redex by an equivalent one inside it. As infinite 
terms need not have innermost redexes, such an inside--out approach 
does not immediately carry over to the infinitary case. 
To overcome this difficulty we switch from inside--out to outside--in.
We even present two outside--in approaches:
a simple but non-effective one to establish the triangle property,
and a more complex but effective one to establish the diamond property.

\section*{The Triangle Property, Non-effectively}

We show that multi-steps have the triangle property for weakly orthogonal 
iTRSs without collapsing rules, hence as a consequence the diamond property.
This generalizes the same result for weakly orthogonal TRSs~\cite[Theorem~8.8.27]{terese:2003}.
Since the proof of the theorem in the finitary case employs an inside--out
approach, it does not carry over immediately to the infinitary case as observed above.
We show that, nonetheless, the infinitary case can be reduced to the
finitary case by means of a careful analysis of clusters of redexes.
The main observation is that infinite clusters can be dropped alltogether,
as illustrated in the following example.
\begin{figure}[htb]
  \begin{center}
    \input{figs/chain}
  \end{center}
  \vspace{-2ex}
  \caption{\textit{Infinite chain of overlaps.}}
  \label{fig:chain}
\end{figure}

\begin{example}\label{ex:chain} \label{exa:chain}{
  \renewcommand{\a}{\funap{A}}
  \begin{samepage}
  Consider the TRS $R$ consisting of the single rule
  $\rulstr{\a{\a{\a{\atrmvar}}}}{\a{\atrmvar}}$
  and an infinite term $A^{\omega}$ containing an infinite
  chain of overlaps as displayed in Figure~\ref{fig:chain},
  with each blue redex partially overlapping its adjacent
  green redexes and vice versa:
  \begin{center}
    \input{figs/AAA}
  \end{center}
  The main question in establishing the triangle property
  is how to construct a common reduct $\trmfula{(A^{\omega})}$ 
  of all possible multi-steps from $A^{\omega}$.
  \end{samepage}
  
  To that end, observe that by weak orthogonality contracting some redex
  in the chain has exactly the same effect as contracting any
  other redex in it. By choosing a sequence of redexes ever
  deeper in the term, and noting that contracting a redex 
  leaves its context intact, it follows that in fact the whole
  chain must be left intact by contracting a redex in it. 
  That is, we can simply define
  $\trmfula{(A^{\omega})} \isdefd A^{\omega}$.
}\end{example}
The notion of chain is covered by the 
notion of cluster~\cite[Definition~4.31]{kete:klop:oost:2004a}.
\begin{definition}\label{def:cluster}
  A \emph{cluster} is a non-empty set of redexes which forms 
  a connected component with respect to the overlap relation $\poverlap$.
  The \emph{pattern} of a cluster is the union of the patterns of its redexes.
  A cluster is said to be \emph{infinite} if the set of
  root-positions of the redexes in it, is infinite.
  A cluster is a \emph{\Ycluster} if it contains a pair of redexes at parallel positions (Figure~\ref{fig:cases:orthogonalization}, cases (ii) and (iv));
  otherwise it is an \emph{\Icluster} (Figure~\ref{fig:cases:orthogonalization}, cases (i) and (iii)). 
\end{definition}
A notion for a set of positions (for example, overlap) is extended to a 
cluster via the union of the positions in the patterns of the redexes in 
it. Note that for any cluster, there is a least (topmost) position 
overlapped by it, which we call its \emph{root}.
We use $\acls$, $\bcls$ to range over clusters.
\begin{example} \label{exa:cluster}
  Figure~\ref{fig:chain} displays a single infinite \Icluster.
\end{example}
Clusters are to weakly orthogonal iTRSs 
what redexes are to orthogonal iTRSs.
\begin{lemma}[Cluster Redex] \label{lem:cluster:step}
  If $\acls$ and $\bcls$ are clusters in a term $\atrm$, 
  then:
  \begin{enumerate}
  \item
    there is a unique $\btrm$ such that 
    for some and all $\setin{\ardx}{\acls}$, $\iarsa{\ardx}{\atrm}{\btrm}$,
    and for any position $\apos$ not overlapping $\acls$, 
    there is a unique $\aPos$ such that
    $\poseq{\resa{\apos}{\acls}}{\aPos}$; and
  \item
    either $\acls$ and $\bcls$ are identical or
    they are non-overlapping.
 \end{enumerate}  
\end{lemma}
\begin{proof} \mbox{}
  \begin{enumerate}
  \item
    Suppose $\iarsa{\ardx}{\atrm}{\btrm}$ and $\iarsa{\brdx}{\atrm}{\ctrm}$
    for $\setin{\ardx,\brdx}{\acls}$, and let $\apos$ be a position not overlapping $\acls$
    By definition of cluster there is a sequence  
    $\rdxeq{\ardx}{\airdx{\natone}},\ldots,\rdxeq{\airdx{\anat}}{\brdx}$ of redexes in $\acls$,
    such that consecutive elements have overlap.
    By weak orthogonality, these pairwise induce steps having 
    the same targets and the same descendant relation for positions
    not overlapping them, from which we conclude by transitivity;
  \item
    If $\acls$ and $\bcls$ are overlapping, then
    they contain redexes which are overlapping, hence
    belong to the same equivalence class. \qedhere
  \end{enumerate}
\end{proof}
The Cluster Redex Lemma justifies speaking of a \emph{cluster}-step 
contracting $\acls$, denoted by $\aiars{\acls}$, unambiguously
inducing a residual relation on other clusters and cluster-steps.
\begin{remark} \label{rem:cluster}
  Descendants of positions \emph{within} a cluster may depend on the redex 
  contracted~\cite[Section~9.3.1]{terese:2003}\cite[Example~4.55]{kete:klop:oost:2004a}. 
  For example, the position
  $\posone$ has either zero or one descendants along the step 
  $\arsa{f(a)}{a}$ in the weakly orthogonal TRS with rules
  $\rulstr{f(x)}{x}$, $\rulstr{f(a)}{a}$. 
\end{remark}
As noted in Example~\ref{exa:chain}, 
certain arrangements of redexes in a cluster trivialise it.
\begin{definition} \label{def:trivial}
  A \emph{trivial} cluster is either a \Ycluster{} or an infinite \Icluster.
\end{definition}
Note that triviality of clusters is established on the basis of left-hand sides 
of rules only, and that the non-trivial clusters are exactly the finite \Icluster{s}.
Given a trivial cluster, the corresponding cluster-step is trivial
in the sense that its source is equal to its target, i.e.\ it
is inert from an outside perspective.
The following lemma formalizes this and generalizes earlier observations that 
\Ycluster{s} and their special case, Takahashi-configurations, are 
trivial~\cite[Proposition~9.3.5]{terese:2003}\cite[Remark~4.38]{kete:klop:oost:2004a}.
\begin{lemma}[Trivial Cluster]\label{lem:trivial}
  If $\acls$ is a trivial cluster, 
  then $\iarsa{\acls}{t}{s}$ implies $\trmeq{\atrm}{\btrm}$,
  and for any position $\apos$ not overlapping $\acls$, 
  $\poseq{\resa{\apos}{\acls}}{\setstr{\apos}}$.
\end{lemma}
\begin{proof}
  Suppose  $\iarsa{\acls}{\atrm}{\btrm}$.
  
  If $\acls$ is a \Ycluster, then it contains redexes
  $\ardx$, $\brdx$ whose roots are incomparable.
  By the previous lemma, for some term $\btrm$, both 
  $\iarsa{\ardx}{\atrm}{\btrm}$ and $\iarsa{\brdx}{\atrm}{\btrm}$.
  Since the former leaves the context of $\ardx$ and the subterm
  the subterm at the root of $\brdx$ untouched and, \emph{mutatis mutandis},
  the latter leaves the context of $\brdx$ and the subterm
  at the root of $\ardx$ untouched, we conclude from incomparability
  of their roots and the Cluster Redex Lemma that in fact the whole 
  term must be left unchanged and positions outside the cluster untouched.  
  
  If $\acls$ is an infinite \Icluster, then for any context
  with its hole at the root-path of $\acls$,
  that is, at the path through all roots of redexes in $\acls$,
  there is a redex $\ardx$ below it in $\acls$.
  Since contracting $\ardx$ leaves the context untouched, 
  we conclude from the assumption that the root-path is infinite and
  the Cluster Redex Lemma, that in fact the whole
  term must be left unchanged and positions outside the cluster untouched.
\end{proof}
The above proof displays typical `cluster-reasoning': a property of a cluster-step 
is established as a consequence of a property of \emph{some} step in the cluster.
\begin{remark} \label{rem:trivial}
  The Trivial Cluster Lemma does not imply that redexes in trivial clusters 
  are necessarily due to trivial \emph{rules}, 
  i.e.\ rules of the form $\ell \to r$ with $\ell = r$~\cite[p.508, middle]{terese:2003}.
  To wit, let $\atrs$ consist of the following (non-trivial) rules:
  \begin{align}
    \tfunap{f}{\bfunap{g}{x}{y}}{z}{\bfunap{g}{a}{a}} 
    &\to \tfunap{f}{\bfunap{g}{y}{x}}{z}{\bfunap{g}{a}{a}} \tag{$\rho_1$}\\
    \tfunap{f}{\bfunap{g}{a}{a}}{z}{\bfunap{g}{x}{y}} 
    &\to \tfunap{f}{\bfunap{g}{a}{a}}{z}{\bfunap{g}{y}{x}} \tag{$\rho_2$}\\
    \bfunap{g}{x}{y} 
    &\to \bfunap{g}{y}{x} \tag{$\rho_3$} 
  \end{align}
  We consider the term $\tfunap{f}{\bfunap{g}{a}{a}}{t}{\bfunap{g}{a}{a}}$
  which contains both a $\rho_1$-redex and a $\rho_2$-redex at the root,
  a $\rho_3$-redex at disjoint positions $1$ and $3$. These redexes form a \Ycluster.
  Also, the trivial infinite \Icluster{} of Example~\ref{exa:chain} is due 
  to the non-trivial rule {\renewcommand{\a}{\funap{A}}
  $\rulstr{\a{\a{\a{\atrmvar}}}}{\a{\atrmvar}}$}.
\end{remark}
It is always safe to drop steps in trivial clusters from a
multi-step without changing its outcome.
\begin{lemma}\label{lem:noeffect}
  Let $R$ be a \woTRS{}, $t \in \iter{\asig}$ a term.
  Let $U$ be a multi-redex in $t$, and $V\subseteq U$ such that 
  every redex in $V$ is contained in a trivial cluster of $t$.
  Then the multi-step with respect to $U{\setminus}V$ results
  in the same term as the multi-step with respect to $U$.
\end{lemma}
\begin{proof}
  We reduce in the complete development first all redexes in trivial clusters: 
  by Lemma~\ref{lem:trivial} this leaves the term as well as all redexes not in 
  trivial clusters untouched. As a consequence, the result of the complete development 
  (multi-step) depends only on the redexes not in trivial clusters.
\end{proof}
The above allows for a cluster-wise definition by cases of the map
required for the triangle property: trivial clusters can simply be mapped
to themselves, and since non-trivial clusters are finite we may proceed
for them analogously (but outside--in instead of inside--out)
to the finite case~\cite[Theorem~8.8.27]{terese:2003}.

The \emph{tail} of a redex in an \Icluster\ $\acls$ is its maximal position on the root-path of $\acls$
(that is, the position of the last pattern symbol of $\acls$ along the root-path of $\acls$).
\begin{definition} \label{def:full:multi:redex}
  The \emph{full} multi-redex $\rdxfula{\acls}$ of a cluster $\acls$ in a term $\atrm$ is
  defined by case-distinction as follows:
  \renewcommand{\descriptionlabel}[1]{\hspace{\labelsep}{{\normalfont{#1}}}} 
  \begin{description}
  \item[\emph{Case~1: $\acls$ is trivial}.] 
    Then $\rdxfula{\acls} = \setemp$.
  \item[\emph{Case~2: $\acls$ is non-trivial}.] 
    Then $\acls$ is a finite \Icluster,
    and $\rdxfula{\acls}$ is defined by repeating the following procedure 
    until no further selections are possible:
    Select a redex that is below the already selected ones,
    such that its tail is minimal (topmost among the remaining).
  \end{description}
  The \emph{full} multi-redex $\rdxfula{\atrm}$ of $\atrm$ is the union
  of the full multi-redexes for all its clusters.
  We write $\trmfula{\atrm}$ to denote a term obtained by developing $\rdxfula{\atrm}$.
\end{definition}
The procedure is the (top--bottom) mirrored version of the procedure 
in~\cite[Proposition~8.8.23]{terese:2003}. Mirroring works
because we are considering finite \Icluster{s}, for these only
their root-paths are relevant, and paths are clearly mirrorable.
\begin{theorem}[Triangle]\label{thm:nptriangle} \label{thm:angle}
  In every weakly orthogonal TRS without collapsing rules 
  the multi-step rewrite relation $\aarsdev$ has the triangle property.
 \end{theorem}
\begin{proof}
  We show that for any development $\iarsdeva{\aRdx}{\atrm}{\btrm}$,
  it holds  $\arsdeva{\btrm}{\trmfula{\atrm}}$.
   
  We first show that we may assume without loss of generality that $\aRdx$ only 
  contains redexes that are not part of trivial clusters.
  Write $\aRdx$ as $\setsum{\aiRdx{\natone}}{\aiRdx{\nattwo}}$
  with $\aiRdx{\natone}$ consisting of all redexes in $\aRdx$ 
  contained in trivial clusters. 
  Developing $\aRdx$ by contracting $\aiRdx{\natone}$ gives rise to 
  $\iarsdeva{\aiRdx{\natone}}{\atrm}{\iarsdeva{\resa{\aRdx}{\aiRdx{\natone}}}{\ctrm}{\btrm}}$.
  By the Trivial Cluster Lemma $\trmeq{\atrm}{\ctrm}$ and
  $\rdxeq{\resa{\aRdx}{\aiRdx{\natone}}}{\aiRdx{\nattwo}}$.
  
  Next we show that $\aRdx$ may be covered by the $\rdxfula{\atrm}$.
  Write $\aRdx$ as $\setSum{\acls}{\aiRdx{\acls}}$ such that each
  $\aiRdx{\acls}$ is contained in the non-trivial finite I-cluster $\acls$.
  Since $\acls$ is non-trivial, $\rdxfula{\atrm}$ contains
  the full multi-redex of $\acls$. 
  This allows to \emph{cover} $\aRdx$, in the sense that we can define
  an injective mapping $\ainj$ mapping every redex
  $\setin{\ardx}{\aiRdx{\acls}}$ to a redex 
  $\setin{\inja{\ardx}}{\rdxfula{\atrm}}$ such that 
  $\ardx$ overlaps the tail of $\inja{\ardx}$.

  Finally, consider a development of $\aRdx$ contracting
  redexes in outside--in order. This development is mapped by $\ainj$ 
  to a development of $\inja{\aRdx}$ performing exactly the same
  steps, which is still converging since each
  redex in $\ardx$ is covered by a redex overlapping it,
  and which therefore has the same target, 
  i.e.\ $\iarsdeva{\inja{\aRdx}}{\atrm}{\btrm}$.
  Since $\inja{\aRdx}$ is contained in $\rdxfula{\atrm}$ we
  conclude by completely developing the residuals of the
  latter after the former 
  $\iarsdeva{\resa{\rdxfula{\atrm}}{\inja{\aRdx}}}{\btrm}{\trmfula{\atrm}}$.

  Note that since $\arsdeva{\atrm}{\trmfula{\atrm}}$ holds by 
  definition, multi-steps even have the triangle property.
\end{proof}

As a direct consequence we obtain, the following corollary.
\begin{corollary}[Diamond]\label{cor:diamond}
  In every weakly orthogonal TRS without collapsing rules 
  the multi-step rewrite relation $\aarsdev$ has the diamond property.
\end{corollary}

\section*{The Triangle Property, Effectively}

Theorem~\ref{thm:angle} and Corollary~\ref{cor:diamond} show that 
weakly orthogonal iTRSs without collapsing rules have both
the triangle property and the diamond property.
Still, the above is somewhat unsatisfactory in that it does not yield a construction
to obtain a common reduct of two multi-steps from a given term, even if these steps
are given effectively. The reason is that the definition of the full multi-redex 
of a given term employs a case distinction on whether a cluster is finite or not,
an undecidable property in general. Here we remedy that and present an effective
orthogonalization procedure. To that end, we first present the idea of \emph{orthogonalizing} 
a set of redexes and recapitulate the concrete orthogonalization procedure for weakly orthogonal 
TRSs of~\cite[Section~8.8.4]{terese:2003}, and next show, by a careful analysis of clusters, 
that it can be extended to an effective procedure for weakly orthogonal iTRSs without collapsing rules.

In a peak $\relap{\relsco{\aiarsdevinv{U}}{\aiarsdev{V}}}{\atrm}{\btrm}$ in a 
left-linear TRS, the union $\setlub{U}{V}$ of the multi-redexes $U$ and $V$ may be non-orthogonal
(there may be overlaps) despite that the multi-redexes $U$ and $V$ themselves are orthogonal.
As a consequence, in general no common reduct of $\atrm$ and $\btrm$ 
can be found and confluence is lost for such TRSs.
However, in some cases and in particular in the case of weakly orthogonal TRSs,
the non-orthogonality is more apparent than real, in the sense that the peak can be replaced by
another equivalent\footnote{%
Equivalent in the sense of relating the same two terms $t$ and $s$.}
one $\relap{\relsco{\aiarsdevinv{\rdxort{U}}}{\aiarsdev{\rdxort{V}}}}{\atrm}{\btrm}$
such that $\setlub{\rdxort{U}}{\rdxort{V}}$ \emph{is} orthogonal.
As a consequence, a common reduct of $t$ and $s$ can be reached via the usual orthogonal projections
$\relap{\relsco{\aiarsdev{\project{\rdxort{V}}{\rdxort{U}}}}{\aiarsdevinv{\project{\rdxort{U}}{\rdxort{V}}}}}{\atrm}{\btrm}$ 
of \emph{those}, now orthogonal, multi-steps, and confluence is regained.
In such cases, the function $\srdxort$ will be called an \emph{orthogonalization}.

We let the orthogonalization of a set of redexes $W$ in a term $t$,
depend on a subset $U$ of $W$ of redexes that `have already been orthogonalized'.
This reflects that our orthogonalization algorithm will proceed incrementally, 
initially setting $U$ to the empty set imposing no orthogonality constraints, 
incrementing it in each iteration of the main loop, 
expressing that an ever growing initial segment of the 
(potentially) infinite set of redexes has been orthogonalized, until in the 
limit an orthogonalization of the whole set $W$ is obtained.
\begin{definition}
  A \emph{$U$-orthogonalization} for a subset $U$ of a given set of redexes $W$ in a term $t$
  of a TRS is a partial function $\srdxort$ from $W$ to itself such that
  \begin{itemize}
  \item
    the image $\rdxort{U}$ of $U$ is an orthogonal set of redexes; and
  \item
    $\iarsdeva{\rdxort{V}}{\atrm}{\btrm}$
    for every multi-redex $\setle{V}{W}$ such that $\iarsdeva{V}{\atrm}{\btrm}$.
  \end{itemize}
  where we have used superscripting to indicate the application 
  of $\srdxort$ lifted to subsets of $W$.
  In case $U = W$ we simply speak of orthogonalization.
\end{definition}
The first condition expresses that $\srdxort$ orthogonalizes the subset $U$
and the second condition that orthogonalization is consistent also with 
redexes outside $U$, in that any orthogonal subset of $W$, the multi-redex $V$,
is mapped to a set of redexes that is orthogonal again, as implicitly expressed by 
$\iarsdeva{\rdxort{V}}{\atrm}{\btrm}$.
For instance, the identity function is an $\emptyset$-orthogonalization of any
set of redexes, and an orthogonalization of any orthogonal set of redexes.

We say a TRS \emph{admits} orthogonalization if for every term $t$, every set of redexes $W$, 
and every subset $U$ of $W$, there exists a $U$-orthogonalization of $W$ in $t$.
If this holds when the sets $W$ are restricted to unions of pairs of multi-redexes, 
we say the TRS admits \emph{binary} orthogonalization.
\begin{theorem}[Diamond and Triangle by Orthogonalization]\label{thm:diamond:effective}
  The multi-step relation $\aarsdev$ of a left-linear TRS
  has the diamond property, if the TRS admits binary orthogonalization,
  and has the triangle property, if the TRS admits orthogonalization.
\end{theorem}
\begin{proof}
  \mbox{}
  \begin{itemize}
  \item
    The proof of the diamond property is as outlined above:
    Suppose $\iarsdevinva{U}{\btrm}{\iarsdeva{V}{\atrm}{\ctrm}}$.
    As the TRS is assumed to admit binary orthogonalization,
    there exists an orthogonalization $\srdxort$ of the union $W = \setlub{U}{V}$ 
    of the multi-redexes $U$,$V$ in $t$.
    so $\iarsdevinva{\rdxort{U}}{\btrm}{\iarsdeva{\rdxort{V}}{\atrm}{\ctrm}}$.
    By the Infinite Developments Lemma we conclude to
    $\iarsdeva{\project{\rdxort{V}}{\rdxort{U}}}{\btrm
}{\iarsdevinva{\project{\rdxort{U}}{\rdxort{V}}}{\atrm'
                                               }{\ctrm
                                               }}$ for some term $\atrm'$.
  \item
    To see the triangle property holds, let $\srdxort$ be an
    an orthogonalization of the set of all redexes of the term $t$;
    it exists since the TRS is assumed to admit orthogonalization. 
    Denoting (for convenience) both the set of all redexes 
    and the result of contracting it by $\rdxort{t}$, 
    we have $\iarsdeva{\rdxort{t}}{t}{\rdxort{t}}$.
    For an arbitrary multi-step $\iarsdeva{U}{t}{s}$,
    $\setle{U}{\rdxort{t}}$, hence
    $\iarsdeva{\rdxort{U}}{t}{s}$ and by the Infinite Developments Lemma
    $\iarsdeva{\project{\rdxort{t}}{\rdxort{U}}}{s}{\rdxort{t}}$ as desired. \qedhere
  \end{itemize}
\end{proof}

\noindent Although this result only depends on left-linearity,
some TRSs do not admit an orthogonalisation even if confluent,
e.g.\ consider $\setstr{\rulstr{a}{b},\rulstr{a}{c},\rulstr{b}{c}}$,
but others such as weakly orthogonal TRSs\footnote{%
The proof that so-called development-closed TRSs are confluent~\cite{terese:2003}
could also be seen as providing a stepwise procedure for orthogonalizing any peak of multi-steps,
each time replacing a non-orthogonal peak by an equivalent `more orthogonal' one,
but there the source of the peak may change.
}
do and these will be of interest here.
\begin{remark}
  Obviously if a TRS admits orthogonalization, then it admits binary orthogonalization.
  In view of Corollary~\ref{cor:effective:diamond:triangle} below, it may be
  interesting to investigate under which conditions (an effective version of) the converse holds.
  
  Not every TRS that admits orthogonalization is weakly othogonal,
  two typical examples being the `left-reducible' TRS\footnote{%
Call a rule \emph{left-reducible} if its left-hand side is reducible with respect to the other rules.
} with rules
  $\setstr{\rulstr{f(a)}{g(b)},\rulstr{f(x)}{g(x)},\rulstr{a}{b}}$, 
  and the `feebly orthogonal' TRS\footnote{%
Call a critical peak $s_1 \leftarrow \arsa{t}{s_2}$ \emph{feeble}
if the cardinality of the set $\setstr{s_1,t,s_2}$ is at most $2$.
} with rules
  $\setstr{\rulstr{g(f(a,x))}{b},\rulstr{f(x,a)}{f(x,x)}}$.
\end{remark}
In the special case of a TRS that is weakly orthogonal,
if a redex overlaps another one, both are equivalent, so either can be replaced by the other.
Then the challenge in developing an orthogonalization for a pair of multi-redexes
is to make these replacements consistently, i.e.\ in such a way that no new overlaps are created. 
\begin{example} \label{exa:orthogonalization}
  Consider Figure~\ref{fig:orthogonalization}.
  The set $\setstr{1,2,3,4,5}$ combining the multi-redexes of the peak 
  $\relsco{\aiarsdevinv{1,2}}{\aiarsdev{3,4,5}}$ of multi-steps, is not orthogonal,
  e.g.\ the redexes $2$ and $3$ partially overlap each other. 
  When attempting to orthogonalize it, these overlaps have to be resolved
  but one has to be careful since e.g.\ either replacing the redex $2$ by $3$ 
  or replacing $3$ by $2$ would create new conflicts
  (between $3$ and $1$ and between $2$ and $5$ respectively).
\end{example}
\begin{figure}[hpt!]
  \begin{center}
  \scalebox{.8}{\input{figs/ortho}}
  \end{center}\vspace{-2ex}
  \caption{\textit{Orthogonalization in a weakly orthogonal TRS.}}
  \label{fig:orthogonalization}
\end{figure}
Creation of new conflicts can be avoided by proceeding in inside--out fashion,
resolving partial overlaps by replacing the outer by the inner redex~\cite[Theorem~8.8.23]{terese:2003}.
\begin{example}
  Reconsider Figure~\ref{fig:orthogonalization} and 
  apply the orthogonalization of~\cite[Theorem~8.8.23]{terese:2003}. 
  We start at the bottom of the tree. 
  The first partial overlap we find is between the redexes $2$ and $5$; 
  this is removed by replacing the outer redex $2$ by the inner redex $5$,
  i.e.\ setting $\srdxort(2) = \srdxort(5) = 5$.
  Then the partial overlap between $2$ and $3$ has also disappeared.
  The only remaining partial overlap is between the redexes $3$ and $1$. 
  Hence we replace the outer redex $3$ by the inner redex $1$,
  i.e.\ $\srdxort(1) = \srdxort(3) = 1$.
  Finally, as $4$ does not overlap any other redex we set $\srdxort(4) = 4$.
  As result $\srdxort$ maps the non-orthogonal set $\setstr{1,2,3,4,5}$ 
  to the orthogonal set $\setstr{1,4,5}$,
  and applying it to the peak $\relsco{\aiarsdevinv{1,2}}{\aiarsdev{3,4,5}}$
  yields the equivalent orthogonal peak $\relsco{\aiarsdevinv{1,5}}{\aiarsdev{1,4,5}}$,
  as desired.
\end{example} 
As an inside--out orthogonalization procedure obviously cannot work on infinite terms,
here we will proceed dually, in outside--in fashion.
Roughly speaking, we start at the top of the term and replace overlapping redexes with the outermost one.
However, care has to be taken in situations as depicted in Figure~\ref{fig:orthogonalization},
where it seems that the \emph{two} inner redexes $1$ and $2$ would both 
need to be replaced by the \emph{single} outer redex $3$. 
The key observation to overcome this problem is that in such cases all redexes $1$,$2$,$3$ 
belong to the same \Ycluster{} and, by Lemma~\ref{lem:noeffect}, contracting 
it only yields a trivial step, i.e.\ a step from the term to itself. 
Hence, upon detection, it is safe to simply discard these redexes.
Before formally defining the corresponding orthogonalization procedure,
we illustrate it by an example showing four typical\footnote{%
In fact, these cases `cover' all possibilities arising during the orthogonalization
of the union of two multi-redexes, as needed when establishing the effective 
version of the diamond property.} 
cases arising during orthogonalization.
\begin{example}
  Consider orthogonalizing the set $\setlub{U}{V}$ in the four cases as displayed in 
  Figure~\ref{fig:cases:orthogonalization}, 
  with $U$ the (blue) multi-redex containing $u$ and (possibly) $m$
  and $V$ the (green) multi-redex containing $v$ and (possibly) $w$, 
  and with $u$ a topmost (that is, having minimal depth) redex in $U$
  and $v$ a topmost redex in $V$ overlapping $u$ from the inside,
  see 
  \begin{figure}[hpt!]
  \vspace{-1ex}
  \begin{center}
          \input{figs/cases}
  \end{center}
  \vspace{-2ex}
  \caption{\textit{Four typical cases for the orthogonalization algorithm.}}
  \label{fig:cases:orthogonalization}
  \vspace{-1ex}
  \end{figure}
  \begin{enumerate}[label=(\roman*)]
   \item If $v$ is the only redex in $V$ that overlaps with $u$,
       then we replace $v$ by $u$.
  \end{enumerate}
  \noindent
  Otherwise 
  we pick a redex $w \in V$, $w \ne v$ and $w$ overlaps $u$.
  \begin{enumerate}[label=(\roman*)]
  \setcounter{enumi}{1}
  \item Assume that $v$ and $w$ are at parallel (disjoint) positions.
      Then $u$, $v$ and $w$ belong to a \Ycluster{} and can be dropped from $U$ and $V$ by Lemma~\ref{lem:noeffect}.
  \end{enumerate}
  \noindent
  Otherwise, $v$ and $w$ are not disjoint, and then $w$ must be nested inside $v$.
  \begin{enumerate}[label=(\roman*)]
  \setcounter{enumi}{2}
  \item If $u$ is the only redex from $U$ overlapping $v$, then we can replace $u$ by $v$.
  \item In the remaining case there must be a redex $m \in U$, such that $m \ne u$ and $m$ overlaps with the redex $v$, 
      see case (iv) of Figure~\ref{fig:cases:orthogonalization}. 
      Then $u$, $m$, $v$ and $w$ are contained in a \Ycluster{} again,\footnote{%
This uses convexity of patterns to establish 
that $v$ cannot `tunnel through' $w$ to touch~$m$, cf.~\cite{kete:klop:oost:2004a}.
}
      hence can be dropped from $U$ and $V$, as in case (ii).
  \end{enumerate}
\end{example}

\noindent Cases (ii) and (iv) of the example are both dealt with in the then--branch
of the orthogonalization algorithm below,
and cases (i) and (iii) in the else--branch.

\subsection*{The Orthogonalization Algorithm}

The orthogonalization algorithm for weakly orthogonal iTRSs
is given in Figure~\ref{fig:orthogonalization:procedure}
It computes for any given set $W$ of (possibly overlapping) redexes in a term $t$ its orthogonalization.
It does so by each time picking a topmost redex and considering the set of
all redexes that are first or second degree overlapping with it, i.e.\ that overlap the
redex itself, or that overlap a redex that overlaps the redex. 
Then the algorithm distinguishes cases on whether this set contains parallel redexes or not,
i.e.\ on whether or not the redexes can be \emph{locally} seen to belong to a \Ycluster.
If the set contains parallel redexes, all redexes in it belong to a \Ycluster{} and can be discarded.
Otherwise, their roots are located on a path called the root-path,
i.e.\ \emph{locally} the redexes seem to belong to a \Icluster{}, 
a redex `highest' on that root-path is selected and 
all redexes overlapping it are mapped onto it.
\begin{figure}
\begin{algorithmic}[1]
  \Require $W$ is a set of redexes in a term $\atrm$;
  \State $C \gets W$;
  \State $\srdxort$ is the partial function on the empty domain;
  \While{$C$ is non-empty}
    \State $x \gets$ a redex at minimal depth in $C$;
    \State $X \gets$ the set of redexes in $C$ that overlap $x$;
    \State $X \gets$ the set of redexes in $C$ that overlap some redex in $X$;
    \If{$X$ contains parallel redexes}
      \State $\srdxort$ is set to be undefined on $X$;
    \Else
      \State $x \gets$ a redex in $X$ that is above every redex in $X$ that is 
      below some redex in $X$;
      \State $X \gets$ the set of redexes in $X$ overlapping $x$;
      \State $\srdxort$ is set to $x$ on $X$;
    \EndIf
    \State $C \gets \setdif{C}{X}$;
  \EndWhile
  \Ensure 
    the partial function defined as $\srdxort$ on $\setdif{W}{C}$ and
    the identity on $C$, is a $(\setdif{W}{C})$-orthogonalization of $W$ in $t$,
    $C$ neither root-touches $\setdif{W}{C}$ nor
    touches $\rdxort{(\setdif{W}{C})}$.
\end{algorithmic}
\caption{The orthogonalization algorithm for weakly orthogonal non-collapsing iTRSs}
  \label{fig:orthogonalization:procedure}
\end{figure}
\begin{remark}
  Visually one can think of reading a symbol that is 
  either an $\msf{I}$ or a $\msf{Y}$, by starting at the 
  base of the stem and going upward;
  our algorithm then pretends to be reading an $\msf{I}$ until
  we are forced into giving this up because we \emph{locally}
  detect that the stem forks, and have to admit that we have
  been reading a $\msf{Y}$ all along.
  The local detection of forks is what makes the algorithm
  effective. The fact that steps in \Ycluster{}s must
  be trivial is what makes this pretense to be harmless.
\end{remark} 
The intuition for the variables used in the algorithm is that 
$\setdif{W}{C}$ represents the prefix of redexes that have already been orthogonalized, 
$x$ is a (candidate) representative of a (candidate) set $X$ of overlapping redexes,
and $C$ is the suffix of $W$ of redexes that still have to be orthogonalized;
cf.\ Figure~\ref{fig:line:twelve} where a typical state during a run is presented.
As the condition (Ensure) holds after each iteration of the while--loop 
for the suffix $C$, we will refer to it as the \emph{invariant} below.
It is a post-condition only in the limit, when the suffix $C$ is empty.
The algorithm and its pre- and post-condition (Require and Ensure)
all make use of the prefix order on positions lifted 
in the following ways to (sets of) redexes.
\begin{definition}
  Two redexes are said to \emph{overlap} if (the sets of positions of) their patterns do.
  A redex $\ardx$ is said to be \emph{root-above} a redex $\brdx$ if
  the root of $\ardx$ is $\sposlt$-related to the root of $\brdx$,
  and \emph{above} it if in addition their patterns do not overlap.
  \emph{Root-below} is the converse of root-above and \emph{below} is the converse of above.
  
  Two sets of redexes are said to \emph{overlap} if some redex
  in the one overlaps a redex in the other.  
  A set of redexes $U$ \emph{touches} a set of redexes $V$,
  if there is some redex in $U$ such that its root is
  $\sposle$-related to a position in the pattern of some redex in $V$.
  If in addition the latter holds when restricting to root positions,
  then $U$ \emph{root-touches} $V$.
\end{definition}
Note that if $U$ does not touch $V$, then $U$ also does not root-touch $V$ but not vice versa.
\begin{remark}
  There are various equivalent ways to define that $U$ touches $V$.
  One would be to say that $U$ overlaps the prefix-closure of $V$.
  Another would be to say that the subterms contracted by $U$
  overlap $V$. These alternative definitions illustrate that the
  notion of `touch' also covers the case that a redex in $U$ is above a redex in $V$.
\end{remark}
Below we will confuse $\srdxort$ in the orthogonalization algorithm with its extension 
to the whole of $W$ (by mapping elements in $C$ to themselves) in its invariant.
\begin{example}
  To illustrate the algorithm we apply it to the earlier examples.
  \begin{enumerate}[label=(\roman*)]
  \item    
    For ease of reference, let $u_i$ ($v_i$) be the $i$th blue (green) redex from the top
    of the infinite chain $W = \setstr{ u_1,v_1,u_2,v_2, \ldots}$ of blue and green redexes,
    in Figure~9.    
    Executing the algorithm successively leads to to $C = W$,
    $x = u_1$, and first $X = \setstr{ u_1,v_1 }$ and next $X = \setstr{ u_1,v_1,u_2 }$.
    Then since $X$ does not contain parallel redexes and only $u_1$ is above
    all redexes that are below some redex, i.e.\ above $u_2$,
    we let $\srdxort$ map the redexes overlapping $u_1$, i.e.\ $u_1$ itself and $v_1$,
    both to $u_1$.
    
    In the next iteration of the loop $C = \setstr{ u_2,v_2, \ldots}$,
    and one proceeds analogously, resulting in that 
    $\srdxort$ maps $u_2$,$v_2$ both to $u_2$.
    
    In the limit, $\srdxort$ maps $u_i$ and $v_i$ to $u_i$, for all $i$,
    i.e.\ $\srdxort$ maps $W$ onto the set $\setstr{u_1,u_2,\ldots}$ (of blue redexes),
    which is seen to constitute an orthogonal subset of $W$ indeed.
  \item
    Consider Figure~\ref{fig:cases:orthogonalization} for $W$ the subset of $\setstr{u,m,v,w}$ appropriate to each case.
   \begin{enumerate}[label=(\roman*)]
    \item
      This leads successively to $C = \setstr{ u,v }$, $x = u$ and $X = \setstr{ u,v }$.
      Then since $X$ does not contain parallel redexes, 
      $\srdxort$ will select one of $u$,$v$ to map both to.
      Thus $\rdxort{\setstr{u,v}}$ is a singleton, hence orthogonal.
    \item
      This leads successively to $C = \setstr{ u,v,w }$, $x = u$ and $X = \setstr{ u,v,w }$.
      Then since $X$ contains the parallel redexes $v$,$w$,
      the map $\srdxort$ is taken to be undefined on $X$.    
      Thus $\rdxort{\setstr{u,v,w}} = \emptyset$;
    \item
      This leads successively to $C = \setstr{ u,v,w }$, $x = u$ and $X = \setstr{ u,v,w }$.
      Then since $X$ does not contain parallel redexes, and only $v$ is above all redexes
      that are below some redex, i.e.\ above $w$, 
      we let $\srdxort$ map both redexes overlapping $v$, i.e.\ $u$ and $v$, to $v$.      
      In the next iteration of the loop we successively have
      $C = \setstr{ w }$, $x = w$ and $X = \setstr{ w }$, leading
      to $w$ being mapped onto itself by $\srdxort$.      
      Thus $\rdxort{\setstr{ u,v,w }} = \setstr{ v,w }$;
    \item
      This leads successively to $C = \setstr{ u,v,w }$, $x = u$ and first $X = \setstr{ u,v,w }$,
      and next $X = \setstr{ u,v,w,m }$.
      Then since $X$ contains the parallel redexes $w$,$m$,
      the map $\srdxort$ is taken to be undefined on $X$, so 
      $\rdxort{\setstr{ u,v,w,m }} = \emptyset$.
    \end{enumerate} 
  \end{enumerate}
\end{example}

\noindent Correctness of this intuition is established by the following theorem, 
the proof of which depends on the Finite Jump Developments 
Theorem~\cite[Proposition~12.5.9]{terese:2003}, expressing (among others)
that all ways of developing a multi-redex result in the same term.
The theorem applies since developments of multi-redexes in a left-linear, not-necessarily
orthogonal, iTRS can be seen as developments of a suitable labelled version of
it that \emph{is} orthogonal (see Remark~6.3), and non-collapsingness guarantees
it has \emph{finite jumps} (see~\cite{terese:2003}).
\begin{theorem}
  The orthogonalization algorithm is correct in the sense that
  it produces an orthogonalization $\srdxort$ of the set $W$ of redexes in $t$.
\end{theorem}
\begin{proof}
  We first show that the algorithm is partially correct,
  i.e.\ the invariant is ensured in each iteration of the while loop,
  and then that the algorithm is productive,
  i.e.\ that $\srdxort$ is an orthogonalization of an ever growing prefix of $W$.
  This shows correctness as the invariant ensures that
  in the limit $\srdxort$ is an orthogonalization of the whole of $W$, as desired.

  For ease of reference we use a variable indexed by a line number
  to refer to the value of that variable just before executing the 
  statement on that line in the orthogonalization algorithm.
  For instance, $X_8$ is value of the variable $X$ just before executing
  the statement on line 8 (then $X_8$ contains parallel redexes).
  We follow the structure of the algorithm.
  \begin{itemize}[label=(3--15)]
  \item[(1--2)]
    When entering the while loop for the first time, the invariant is trivially ensured:
    Since $C_3 = W$ by line 1, and $\srdxort_3$ is the partial function on the empty domain by line 2,
    the orthogonalization is the identity function on the whole of $W$ (`it doesn't orthogonalize anything yet').
    Since both the orthogonalized prefix $\setdif{W}{C_3}$ and 
    its orthogonalization $(\setdif{W}{C_3})^{\srdxort_3}$ are empty, 
    there's nothing to be touched in them.
  \item[(3--15)]  
    Supposing the invariant is ensured when starting an iteration of the while loop,
    i.e.\ $\srdxort_4$ is a $(\setdif{W}{C_4})$-orthogonalization,
    $C_4$ neither root-touches $\setdif{W}{C_4}$ nor
    touches $\rdxort{(\setdif{W}{C_4})}$,
    we have to show it is also ensured at the end of that iteration,
    i.e.\ $\srdxort_{15}$ is a $(\setdif{W}{C_{15}})$-orthogonalization,
    $C_{15}$ neither root-touches $\setdif{W}{C_{15}}$ nor
    touches $\rdxort{(\setdif{W}{C_{15}})}$.
    \begin{itemize}
    \item[(4)]
      Executing line 4, we have $x_5$ is a redex at minimal depth in $C_4$ (it exists by $C_4$ being non-empty).
      Thus, no redex $\ardx$ in $C_4$ is above $x_5$ and $\ardx$ overlaps $x_5$,
      then the root of $x_5$ is $\sposle$-related to the root of $\ardx$.
    \item[(5)]
      Executing line 5, we have that 
      $X_6$ is the set of redexes in $C_4$ that overlap $x_5$.
      We claim that $\setdif{C_4}{X_6}$ does not root-touch $X_6$, hence also 
      does not root-touch $\setdif{W}{(\setdif{C_4}{X_6})}$.
      For a proof by contradiction of the first part of the claim,
      suppose there were a redex $\ardx$ in $\setdif{C_4}{X_6}$ the
      root of which is $\sposle$-related to the root of a redex $\brdx$ in $X_6$.
      Since by line 4, the root of $x_5$ is $\sposle$-related to that of $\brdx$.
      By convexity of patterns, therefore $\ardx$ is either above $x_5$
      or overlaps it. This yields a contradiction, in the former case
      with minimality of $x_5$ and in the latter case with $\ardx$ not being an element of $X_6$.
      For a proof of the second part of the claim, note that $\setle{X_6}{C_4}$,
      so $\setdif{W}{(\setdif{C_4}{X_6})} = \setlub{(\setdif{W}{C_4})}{X_6}$.
      so it suffices to show $\setdif{C_4}{X}$ touches neither $\setdif{W}{C_4}$ nor $X_6$.
      The former follows by the invariant for $C_4$ and the latter by the first part of the claim.
    \item[(6)]
      Executing line 6, we have that 
      $X_7$ is the set of redexes in $C_4$ that overlap some redex in $X_6$.
      We claim that $\setdif{C_4}{X_7}$ does not root-touch $X_7$, hence also 
      does not root-touch $\setdif{W}{(\setdif{C_4}{X_7})}$.
      For a proof by contradiction of the first part of the claim,
      suppose there were a redex $\ardx$ in $\setdif{C_4}{X_7}$ the
      root of which is $\sposle$-related to the root of a redex $\brdx$ in $\setdif{X_7}$.
      Since $\brdx$ overlaps some redex $\brdx'$ in $X_6$, $\ardx$ 
      is either above $\brdx'$ or overlaps it. 
      This yields a contradiction, in the former case with that $\setdif{C_4}{X_6}$
      does not root-touch $X_6$, in the latter case with $\ardx$ not being an element of $X_7$.
      The second part of the claim, follows as for (the second part of the claim of) line 5. 
    \item[(7--14)]
      We verify the invariant holds after both branches of the if--then--else on line 7,
      i.e.\ by distinguishing cases on whether or not $X_7$ contains parallel redexes.
      \begin{itemize}
      \item[(8,14)]
        Suppose $X_7$ contains parallel redexes. We verify the conditions of the invariant.
        First note that since the redexes in $X_7$ belong to (one and) the same cluster,
        a cluster that contains parallel redexes,
        $X_7$ in fact is a \Ycluster{}.
        
        To show that $C_{15}$ does not root-touch $\setdif{W}{C_{15}}$,
        note that by line 14, $C_{15} = \setdif{C_4}{X_7}$ 
        so we conclude by the above claim for $X_7$.
        
        To show that $C_{15}$ does not touch $(\setdif{W}{C_{15}})^{\srdxort_{15}}$,
        note that the former is contained in $C_{4}$ and the latter
        identical to $(\setdif{W}{C_{4}})^{\srdxort_{4}}$ since
        $(\setdif{W}{C_{15}})^{\srdxort_{15}} =
         \setlub{(\setdif{W}{C_4})^{\srdxort_{15}}}{(X_7^{\srdxort_{15}})}$
        and by line 8, $\srdxort_{15}$ is $\srdxort_{4}$ on $\setdif{W}{C_{4}}$
        and undefined on $X_7$.
        Hence we conclude by the invariant for $C_4$.
        
        To see that $\srdxort_{15}$ is a $(\setdif{W}{C_{15}})$-orthogonalization,
        we use that $(\setdif{W}{C_{15}})^{\srdxort_{15}} =
        (\setdif{W}{C_4})^{\srdxort_4}$ as noted above. 
        As the latter set is a multi-redex by the invariant, so is the former.
        Next suppose $V$ is a multi-redex contained in $W$ such that $\iarsdeva{V}{t}{s}$.
        The maps $\srdxort_{15}$ and $\srdxort_{4}$ only differ on the \Ycluster{},
        mapping the elements of $X_7$ to undefined respectively themselves,
        but by Lemma~7.10 the corresponding steps are trivial and omitting them
        from \Ycluster{}s does not make a difference. 
        Hence we conclude by the invariant for $\srdxort_4$.
      \item[(10--12,14)]
        Suppose the cluster $X_7$ does not contain parallel redexes.
        Then it has a root-path, i.e.\ a path through the roots of all redexes in $X_7$,
        so these redexes are linearly ordered by the prefix order $\sposle$ on their roots.
        This case is illustrated in Figure~\ref{fig:line:twelve} for $X = \setstr{x,m,o}$
        with these redexes linearly ordered as $\posle{x}{\posle{m}{o}}$.
\begin{figure}[h!]
  \begin{center}
  \begin{tikzpicture}
    [inner sep=1.5mm,
     node distance=5mm,
     terminal/.style={
           rectangle,rounded corners=2.25mm,minimum size=4mm,
           very thick,draw=black!50,top color=white,bottom color=black!20,
     },
     redexA/.style={draw=dblue,fill=fblue,opacity=\oblue,thick},
     redexB/.style={draw=dgreen,fill=fgreen,opacity=\ogreen,thick},
     >=stealth,xscale=.75,yscale=.75]
  
    \begin{scope}[xshift=15cm,yshift=0cm]
      \draw (0,0) -- (4,8) -- (8,0) -- cycle;
      \draw[dashed] (1,4.5) to[bend right=20] (7,4.5);
      \draw[dashed] (.3,2.5) to[bend right=20] (7.7,2.5);
      \fill[redexA,xshift=4cm,yshift=7.8cm,scale=.7] (0,0) -- (-1.4,-2.8) -- (0.85,-1.7) -- cycle;
      \fill[redexB,xshift=4.5cm,yshift=6cm,scale=.6] (0,0) -- (-0.7,-1.2) -- (1.5,-2.6) -- cycle;
      \fill[redexA,xshift=3.3cm,yshift=5cm,scale=.7] (0,0) -- (-1,-1) -- (1,-1) -- cycle;

      \fill[redexA,xshift=3.3cm,yshift=3.85cm,scale=.9] (0,0) -- (-1.4,-1.6) -- (1.2,-2.2) -- cycle;
      \fill[redexB,xshift=3.0cm,yshift=2.7cm,scale=.7] (0,0) -- (-0.7,-1.8) -- (1.5,-2.6) -- cycle;
      \fill[redexA,xshift=2.9cm,yshift=1.7cm,scale=.8] (0,0) -- (-.7,-1.5) -- (1,-1.3) -- cycle;

      \fill[redexB,xshift=5.7cm,yshift=1.8cm,scale=.6] (0,0) -- (-0.7,-1.8) -- (2.2,-1.9) -- cycle;
      \fill[redexA,xshift=5.6cm,yshift=1.1cm,scale=.6] (0,0) -- (-.7,-1.5) -- (.3,-1.3) -- cycle;
      \fill[redexA,xshift=6.2cm,yshift=1.1cm,scale=.6] (0,0) -- (-.2,-1.5) -- (1.3,-1.3) -- cycle;

      \node [anchor=west] at (7.5,1.5) {$\setdif{C}{X}$, remainder};
      \node [anchor=west] at (6.5,3.6) {$X = \setstr{x,m}$, $\rdxort{X} = x$, orthogonalization};
      \node [anchor=west] at (5.3,6) {$\rdxort{(\setdif{W}{C})}$, orthogonal, overlaps resolved};
      \node [terminal] at (3.2,3) {$x$};
      \node [terminal] at (4.0,1.2) {$m$};
      \node [terminal] at (2.9,0.7) {$o$};
      
      \node [anchor=east] at (1.05,4.3) {no touch};
      \node [anchor=east] at (.35,2.3) {no root-touch};

    \end{scope}
  \end{tikzpicture}
  \caption{\textit{State of orthogonalization algorithm at end of line 12.}}
  \label{fig:line:twelve}
  \end{center}
\end{figure}
        \begin{itemize}
        \item[(10)]
          Executing line 10 entails, we claim, that $x_{11}$ exists and $x_{11}$ is an element of $X_6$.
          To that end, let $L$ be the set of redexes in $X_7$ that are below some redex in $X_7$ 
          (one may think of $L$ as the set of `lower bounds').       
          If $L$ is non-empty then let $\ardx$ be a redex in $L$ the root of which is $\sposle$-least;
          it exists by these roots all being on the root-path, but there may be several such. 
          Now take $x_{11}$ to be any redex in $X_7$ that is above $\ardx$ in case $L$ is non-empty.
          To see that $x_{11}$ is an element of $X_6$ suppose that, 
          to the contrary, $x_{11}$ were an element of $\setdif{X_7}{X_6}$,
          i.e.\ that would overlap some redex in $X_6$ that overlaps $x_5$, but that itself would not overlap $x_5$.
          Then by convexity of patterns it would be below $x_5$, 
          hence itself be an element of $L$, so be below itself; a contradiction.
        \item[(11)]
          Executing line 11 entails, we claim, that all elements of $X_{12}$ overlap each other,
          so in particular $x_5$ (from below, by minimality of $x_5$) hence $\setle{X_{12}}{X_6}$, and that
          $\setdif{C_4}{X_{12}}$ does not root-touch $X_{12}$.
          To see that all elements of $X_{12}$ overlap each other, 
          let $\ardx$,$\brdx$ be arbitrary redexes in $X_{12}$ for which we may assume
          without loss of generality that the root of $\ardx$ is $\sposle$-related to that of $\brdx$,
          by linearity. Thus if $\brdx$ were not overlapping $\ardx$, it would be below it.
          Therefore, it would be in the set $L$ and be below $x_{11}$ . 
          But since $\brdx$ was assumed an element of $X_{12}$ it overlap $x_{11}$; a contradiction.
          To see that $\setdif{C_4}{X_{12}}$ does not root-touch $X_{12}$,
          note that since by the claim for line 5, $\setdif{C_4}{X_6}$ does not root-touch $X_6$,
          it suffices by $\setle{X_{12}}{X_6}$ to show that 
          $\setdif{X_6}{X_{12}}$ does not root-touch $X_{12}$. 
          This holds since if a redex $\ardx$ in $\setdif{X_6}{X_{12}}$ were to root-touch an element
          of $X_{12}$ then it would overlap it, hence by the first part of the claim overlap $x_{11}$,
          hence be an element of $X_{12}$; contradiction.
        \item[(12,14)]
          To show that $C_{15}$ does not root-touch $\setdif{W}{C_{15}}$,
          observe that as by line 14, $C_{15} = \setdif{C_4}{X_{12}}$
          and by the above $\setle{\setle{X_{12}}{X_7}}{C_4}$, 
          we have that 
          $\setdif{W}{C_{15}} = \setdif{W}{(\setdif{C_4}{X_{12}})} = \setlub{(\setdif{W}{C_4})}{X_{12}}$
          with the latter a disjoint union.
          Since by the invariant $C_4$ does not root-touch $\setdif{W}{C_4}$,
          $\setdif{C_4}{X_{12}}$ does not do so either, and it remains
          to show that the latter does not root-touch $X_{12}$, which follows from the claim for line 11.
        
          To show that $C_{15}$ does not touch $(\setdif{W}{C_{15}})^{\srdxort_{15}}$,
          we use again $C_{15} = \setdif{C_4}{X_{12}}$
          and $(\setdif{W}{C_{15}})^{\srdxort_{15}} =
          \setlub{(\setdif{W}{C_4})^{\srdxort_{15}}}{X_{12}^{\srdxort_{15}}} =
          \setlub{(\setdif{W}{C_4})^{\srdxort_{4}}}{\setstr{x_{11}}}$.
          By the invariant $C_4$ does not touch $(\setdif{W}{C_4})^{\srdxort_{4}}$,
          so certainly $\setdif{C_4}{X_{12}}$ doesn't either, and it remains
          to show that the latter does not touch $\setstr{x_{11}}$.
          Suppose to the contrary that the root of some redex $\ardx$ in $\setdif{C_4}{X_{12}}$
          were $\sposle$-related to a position in the pattern of $x_{11}$.
          Then the redex $\ardx$ would either overlap $x_{11}$ or be above it.
          But $\ardx$ cannot overlap $x_{11}$ as then it would be in $X_{12}$ as well,
          and it cannot be above it as this would contradict the choice of $x_{11}$.
        
          To see that $\srdxort_{15}$ is a $(\setdif{W}{C_{15}})$-orthogonalization,
          we use $C_{15} = \setdif{C_4}{X_{12}}$,
          $\setdif{W}{C_{15}} = \setlub{(\setdif{W}{C_4})}{X_{12}}$, and
          $(\setdif{W}{C_{15}})^{\srdxort_{15}} =  \setlub{(\setdif{W}{C_4})^{\srdxort_{4}}}{\setstr{x_{11}}}$.
          Since the left part is orthogonal by the invariant, and the right part by being a singleton,
          it suffices that $x_{11}$ be orthogonal to $(\setdif{W}{C_4})^{\srdxort_{4}}$ which follows
          from the invariant as $x_{11}$ is an element of $C_4$ hence below $(\setdif{W}{C_4})^{\srdxort_{4}}$.
          
          Finally, suppose $V$ is a multi-redex contained in $W$ such that $\iarsdeva{V}{t}{s}$.
          We may partition $V$ into the multi-redexes
          $V_1 = \setglb{V}{(\setdif{C_4}{X_{12}})}$,
          $V_2 = \setglb{V}{X_{12}}$ and
          $V_3 = \setglb{V}{(\setdif{W}{C_4})}$.
          By the claim for line 11, all redexes in $X_{12}$ overlap each other,
          so $V_2$ being a multi-redex, it can only be the empty set or a singleton set.
          By the invariant for $\srdxort_4$,
          $\iarsdeva{V^{\srdxort_4}}{t}{s}$,
          Hence if $V_2$ is the empty set, then $V^{\srdxort_{15}} = V^{\srdxort_4}$
          and we conclude immediately.
          Otherwise $V_2$ is a singleton set, say $\setstr{v}$, and  
          $\srdxort_{15}$ differs from $\srdxort_4$ only in that the
          former maps $v$ to $x_{11}$ whereas the latter maps it to itself.
          Thus,
          $V^{\srdxort_4} = \setlub{V_1^{\srdxort_{15}}}{\setlub{\setstr{v}}{V_3^{\srdxort_{15}}}}$
          which by the Finite Jump Developments Theorem may be developed as 
          \[ \relap{\relsco{\aiarsdev{V_1^{\srdxort_{15}}}}{\relsco{\aiarsdev{\setstr{v}}}{\aiarsdev{V_3^{\srdxort_{15}}}}}}{t}{s} \]
          as the redexes in $V_1$ do not touch redexes in $V_2$, which in turn do not touch redexes in $V_3$.
          Per construction $v$ overlaps $x_{11}$ hence by weak orthogonality\footnote{%
This is the only place where weak orthogonality is used in the proof.
}         both induce the same 
          step and still do so after contracting $V_1^{\srdxort_{15}}$,
          since $\setdif{C_4}{X_{12}}$ does not overlap $X_{12}$ and in particular does not overlap $x_{11}$.
          Therefore 
          \[ \relap{\relsco{\aiarsdev{V_1^{\srdxort_{15}}}}{\relsco{\aiarsdev{\setstr{x_{11}}}}{\aiarsdev{V_3^{\srdxort_{15}}}}}}{t}{s} \]
          which by the Finite Jump Developments Theorem again, just is a development of the multi-redex
          $V^{\srdxort_{15}}$.
        \end{itemize}
      \end{itemize}
    \end{itemize}
  \end{itemize}
  It remains to show the algorithm is productive. This follows from that in each iteration of the
  loop a redex of minimal depth is selected, all redexes at that position (and usually more) 
  are removed from $C$ and we assume our signature to have finite arities, so only finitely
  many iterations take place at any given depth.
\end{proof}
\begin{corollary} \label{cor:effective:diamond:triangle}
  For weakly orthogonal non-collapsing TRSs, the orthogonalization algorithm is effective
  \begin{itemize}
  \item
    for the diamond property, if left-hand sides of rules are finite;
  \item
    for the triangle property, if any given redex may only be overlapped by finitely many other redexes;
  \end{itemize}    
\end{corollary}
\begin{proof}
  The construction in the algorithm is based on computing a set $X$ of all redexes that are first or 
  second degree overlapping with a given redex $x$. 
  
  In the case of the diamond property the orthogonalization is
  performed relative to a set $W$ which is the union of two multi-redexes $U$ and $V$, and finiteness of left-hand sides
  then guarantees that $X$ itself is finite, since if, say $x$ is in $U$ then only finitely many redexes,
  all in $V$ except for the redex itself, can be overlapping with it, and in turn only finitely many redexes,
  all in $U$ except for the redexes themselves, can be overlapping with those.
  
  In the case of the triangle property, the ambient set $W$ with respect which 
  orthogonalization takes place is the set of all redexes of a given term,
  and the condition then again guarantees that the set $X$ can be produced.  
\end{proof}

\section{Conclusions}\label{sec:conclusion}

We have shown the failure of $\UNinf$ for weakly orthogonal iTRSs
in the presence of two collapsing rules.
For \woTRS{s} without collapsing rules we proved that $\CRinf$ (and hence $\UNinf$) holds,
and that this result is optimal in the sense that
allowing only one collapsing rule is able to invalidate $\CRinf$.
For these results we have refined two well-known theorems of infinitary rewriting 
with respect to the minimal depth of the steps involved:
\begin{enumerate}
  \item 
    the compression lemma (see Theorem~\ref{thm:compression}), and
  \item 
    the parallel moves lemma (see Lemma~\ref{lem:pml}).
\end{enumerate}  
The refined version of the compression lemma is employed to establish compression also 
for divergent reductions (see Corollary~\ref{cor:comp:div:seqs}).
The proof of this theorem uses a slightly simpler construction than the proof of the compression lemma in~\cite{terese:2003}
(compare Figure~\ref{fig:compression} with~\cite[Figure 12.8]{terese:2003}).
The refined compression lemma is also used to establish the refined version of the parallel moves lemma,
which in turn is used in the proof of infinitary confluence of weakly orthogonal rewrite systems 
without collapsing rules (Theorem~\ref{thm:cr}).

Furthermore, we have shown that infinitary developments in \woTRS{s} without collapsing rules
have the diamond property. In general this property fails already in the presence
of just one collapsing rule.

Apart from this diamond property in itself, 
for which our paper does not yet give an application, 
we point out that here the employed \emph{technique of orthogonalization} 
is the notable contribution. 
Indeed we envisage future elaborations that establish cofinal reduction stategies 
(for finite rewrite sequences on possibly infinite terms) 
in the current setting of infinitary rewriting with weakly orthogonal systems, 
and we expect that such applications will crucially hinge upon 
the use of the orthogonalization technique as presented.

The following table summarizes the results of this paper 
(coloured green) next to known results (black):
\begin{center}
\scalebox{0.7}{{\large \input{figs/table}}}
\end{center}
\medskip

\noindent
The nc-WOTRSs are weakly orthogonal TRSs without collapsing rules;
likewise 1c-WOTRSs have one collapsing rule. 
The fe-OCRSs are fully extended orthogonal CRSs, see~\cite{kete:simo:2009},
and WOCRSs are weakly orthogonal CRSs~\cite{klop:1980}.

The properties of infinitary $\lambda$-calculus summarised in this table concern 
the infinitary calculus arising from the 
standard depth measure where the depth of a symbol occurrence is the length of its position (often referred to as metric $111$).
There are variants of infinitary $\lambda$-calculus
based on different depth measures and corresponding metrics, see further~\cite{endr:hend:klop:2012}.
For these variants the properties can differ, 
for example infinitary $\lambda\beta\eta$ with depth measure $001$
has the properties UN$^{\infty}$ and NF$^{\infty}$ as a consequence of results in~\cite{seve:vrie:2002}.

The failure of $\UNinf$ for two collapsing rules raises the following question, 
as indicated in the table above:
\begin{question}\label{q:one}
  Does $\UNinf$ hold for \woTRS{s} with \emph{one} collapsing rule?
\end{question}

\footnotetext[2]{%
  Beware: in~\cite{beth:klop:vrij:2000} a counterexample is given to the Parallel Moves Lemma PML
  for $\lambda\beta\eta$, but that pertains to the stronger (classical) version of PML
  where the `parallel move' has to consist of contractions of `residuals' 
  of the originally contracted redex.
}

\bibliographystyle{alpha}
\bibliography{main}

\end{document}

%% file: packages.tex
\usepackage{amsfonts} 
\usepackage{latexsym}
\usepackage{amsmath}
\usepackage{amssymb}
\usepackage{amstext}
\usepackage{comment}
\usepackage{color}

\usepackage{algpseudocode}
\usepackage{algorithm}

\usepackage{stmaryrd}
\usepackage{hyperref}
\usepackage{wrapfig}

\usepackage{tikz}
\usetikzlibrary{arrows,automata,backgrounds,calc,chains,fadings,folding,decorations.fractals,fit,patterns,positioning,mindmap,shadows,shapes.geometric,shapes.symbols,through,trees,decorations.markings,decorations.text}

%% file: macros.tex
\definecolor{firebrick}{RGB}{175,25,25}

\theoremstyle{plain}
\newtheorem{theorem}[thm]{Theorem} 
\newtheorem{corollary}[thm]{Corollary}
\newtheorem{lemma}[thm]{Lemma}
\newtheorem{proposition}[thm]{Proposition}

\theoremstyle{definition}
\newtheorem{definition}[thm]{Definition}
\newtheorem{remark}[thm]{Remark}
\newtheorem{example}[thm]{Example}

\renewcommand{\paragraph}[1]{\textit{#1}}

\colorlet{dblue}{blue!75!black}
\colorlet{fblue}{blue!90}
\newcommand{\oblue}{.5}
\colorlet{dgreen}{green!50!black}
\colorlet{fgreen}{green!45}
\newcommand{\ogreen}{.7}

\makeatletter
\def\moverlay{\mathpalette\mov@rlay}
\def\mov@rlay#1#2{\leavevmode\vtop{%
   \baselineskip\z@skip \lineskiplimit-\maxdimen
   \ialign{\hfil$#1##$\hfil\cr#2\crcr}}}
\makeatother

\theoremstyle{plain}
\theoremstyle{definition}
\newtheorem{question}[thm]{Question}

\newcommand{\SN}{\mrm{SN}}
\newcommand{\WN}{\mrm{WN}}
\newcommand{\UNred}[1]{\mrm{UN}_{#1}}
\newcommand{\UN}{\mrm{UN}}
\newcommand{\CR}{\mrm{CR}}
\newcommand{\CONF}{\mrm{CF}}
\newcommand{\NF}{\mrm{NF}}
\newcommand{\NFred}[1]{\mrm{NF}_{#1}}
\newcommand{\SNinf}{\SN^{\infinity}}
\newcommand{\WNinf}{\WN^{\infinity}}
\newcommand{\UNredinf}[1]{\UNred{#1}^{\infinity}}
\newcommand{\UNinf}{\UN^{\infinity}}
\newcommand{\CRinf}{\CR^{\infinity}}
\newcommand{\CONFinf}{\CONF^{\infinity}}
\newcommand{\NFinf}{\NF^{\infinity}}
\newcommand{\NFredinf}[1]{\NFred{#1}^{\infinity}}

\newcommand{\mit}{\mathit}
\newcommand{\mrm}{\mathrm}

\newcommand{\mcl}{\mathcal}
\newcommand{\msf}{\mathsf}

\newcommand{\mbb}{\mathbb}

\newcommand{\infinity}{\infty}

\newcommand{\thsp}{-1.87ex}
\newcommand{\threeheadrightarrow}{{\twoheadrightarrow\hspace*\thsp\twoheadrightarrow}}
\newcommand{\threeheadleftarrow}{{\twoheadleftarrow\hspace*\thsp\twoheadleftarrow}}

\numberwithin{equation}{section}

\newcommand{\nix}{}

\newcommand{\binap}[3]{#2\mathbin{#1}#3}
\newcommand{\funap}[2]{#1(#2)}
\newcommand{\fap}{\funap}

\newcommand{\super}[2]{#1^{#2}}

\newcommand{\bfunap}[3]{\funap{#1}{#2,#3}}
\newcommand{\tfunap}[4]{\bfunap{#1}{#2}{#3,#4}}
\newcommand{\relap}[3]{#2\mathrel{#1}#3}
\newcommand{\relcomp}[2]{#1\cdot#2}

\newcommand{\mylam}[2]{\lambda#1.#2}

\newcommand{\lsubst}[3]{#1[#2{:=}#3]}

\newcommand{\pairlft}{{\langle}}                
\newcommand{\pairrgt}{{\rangle}}                
\newcommand{\pairsep}{{,\,}}                    
\newcommand{\pairstr}[1]{\pairlft#1\pairrgt}    
\newcommand{\pair}[2]{\pairstr{#1\pairsep#2}}   

\newcommand{\sdom}{\mit{dom}}
\newcommand{\dom}{\funap{\sdom}}

\newcommand{\sred}{{\rightarrow}}
\newcommand{\red}{\mathrel{\sred}}
\newcommand{\sredi}{{\leftarrow}}
\newcommand{\redi}{\mathrel{\sredi}}

\newcommand{\smred}{{\twoheadrightarrow}}
\newcommand{\mred}{\mathrel{\smred}}

\newcommand{\sired}{\threeheadrightarrow}
\newcommand{\ired}{\mathrel{\sired}}
\newcommand{\infred}{\ired}
\newcommand{\siredi}{{\threeheadleftarrow}}
\newcommand{\iredi}{\mathrel{\siredi}}

\newcommand{\siconvR}[1]{{=^{\infinity}_{#1}}}
\newcommand{\iconvR}[1]{\mathrel{\siconvR{#1}}}

\newcommand{\spred}{\wash{\,\,\,\,\tikz\draw (0,.5ex) -- (0,-.6ex) (.3ex,.5ex) -- (.3ex,-.6ex);}\mathord{\longrightarrow}}
\newcommand{\pred}{\mathrel{\spred}}
\newcommand{\pstep}[3]{#1\funin #2\pred #3}

\newcommand{\aset}{A}

\newcommand{\setemp}{{\varnothing}}
\renewcommand{\emptyset}{\setemp}
\newcommand{\ssetin}{{\in}}
\newcommand{\setin}{\relap{\ssetin}}

\newcommand{\ssetle}{{\subseteq}}
\newcommand{\setle}{\relap{\ssetle}}

\newcommand{\setwhr}{\mathrel{|}}
\newcommand{\where}{\setwhr}

\newcommand{\sfunin}{{:}}
\newcommand{\funin}{\mathrel{\sfunin}}

\newcommand{\aseq}{\kappa}
\newcommand{\bseq}{\xi}

\newcommand{\seqpref}[2]{#1_{\le #2}}
\newcommand{\seqsuf}[2]{#1_{\ge #2}}

\newcommand{\wordemp}{\varepsilon}

\newcommand{\lstlength}[1]{|#1|}

\newcommand{\nat}{\mbb{N}}
\newcommand{\integers}{\mathbb{Z}}
\newcommand{\ints}{\integers}

\newcommand{\asig}{\Sigma}
\newcommand{\asetofrules}{R}
\newcommand{\avars}{\mathcal{X}}
\newcommand{\ster}{\mit{Ter}}
\newcommand{\ter}{\funap{\ster}}

\newcommand{\svar}{\mit{Var}}

\newcommand{\vars}{\funap{\svar}}

\newcommand{\atrs}{\mcl{R}}

\newcommand{\subtrmat}[2]{#1|_{#2}}
\newcommand{\repsubtrmat}[3]{#1[#2]_{#3}}
\newcommand{\posemp}{\epsilon}
\newcommand{\apos}{p}
\newcommand{\bpos}{q}
\newcommand{\posconcat}[2]{#1#2}
\newcommand{\pos}{\funap{\mathcal{P}\!os}}

\newcommand{\asubst}{\sigma}
\newcommand{\bsubst}{\tau}

\newcommand{\subst}[2]{#2\hspace*{0.5pt}#1}

\newcommand{\acxt}{C}
\newcommand{\bcxt}{D}
\newcommand{\cxthole}{\Box}
\newcommand{\cxtap}[2]{#1[#2]}

\newcommand{\acontext}{\acxt}
\newcommand{\bcontext}{\bcxt}
\newcommand{\contexthole}{\cxthole}
\newcommand{\contextfill}{\cxtap}

\newcommand{\sdefdby}{{:}{=}}
\newcommand{\defdby}{\mathrel{\sdefdby}}

\newcommand{\ssynteq}{{\equiv}}
\newcommand{\synteq}{\mathrel{\ssynteq}}

\renewcommand{\implies}{\Rightarrow}

\newcommand{\slogequiv}{\Leftrightarrow}
\newcommand{\logequiv}{\mathrel{\slogequiv}}

\newcommand{\slogand}{\wedge}
\newcommand {\logand}{\mathrel{\slogand}}

\newcommand{\sarity}{\mit{ar}}
\newcommand{\arity}{\fap{\sarity}}

\newcommand{\siter}{\ster^\infinity}
\newcommand{\iter}{\funap{\siter}}

\newcommand{\woTRS}{weakly orthogonal \iTRS{}}
\newcommand{\iTRS}{{iTRS}}

\newcommand{\spre}{\msf{P}}
\newcommand{\pre}{\funap{\spre}}
\newcommand{\ssuc}{\msf{S}}
\newcommand{\suc}{\funap{\ssuc}}

\newcommand{\aars}{\mcl{A}}

\newcommand{\astep}{\phi}
\newcommand{\bstep}{\psi}
\newcommand{\cstep}{\chi}

\newcommand{\project}[2]{#1/#2}
\newcommand{\unit}{1}

\newcommand{\poverlap}{\leftrightsquigarrow}

\newcommand{\treeap}{\cdot}

\makeatletter

\def\foreachR#1#2#3{%
  \@nullTest@foreach{#1}{#2}#3\@end@foreach}

\def\@nullTest@foreach#1#2{%
  \@ifnextchar\@end@foreach{#1\@swallow}%
    {\@prime@foreach{#1}{#2}}}
\def\@prime@foreach#1#2#3\@end@foreach{%
  \@test@foreach{#1}{#2}#3,\@end@foreach}

\def\@test@foreach#1#2{%
  \@ifnextchar\@end@foreach%
     {#1\@swallow}%
     {\@more@foreach{#1}{#2}}}

\def\@swallow#1{}

\def\@more@foreach#1#2#3,#4\@end@foreach{%
  #2{#3}{\@test@foreach{#1}{#2}#4\@end@foreach}}

\makeatother

\newcommand{\psword}{$\spre\ssuc$-word}

\newcommand{\tsp}{\hspace{.05em}}
\newcommand{\pstolam}[1]{\llparenthesis{\tsp#1\tsp}\rrparenthesis}
\newcommand{\spstolam}{\pstolam{\_}}
\newcommand{\pstolami}[2]{\pstolam{#1}_{#2}}
\newcommand{\spstolami}{\pstolami{\_}}

\newcommand{\ilb}{\lambda^\infinity\beta}
\newcommand{\ilbe}{\lambda^\infinity\beta\eta}

\newcommand{\snorm}[1]{\|#1\|_{\ssuc}}
\newcommand{\pnorm}[1]{\|#1\|_{\spre}}

\newcommand{\pscountex}{\psi}

\newcommand{\BNFis}{\mathrel{{:}{:}{=}}}
\newcommand{\coBNFis}{\super{\BNFis}{\mrm{co}}}
\newcommand{\BNFor}{\mathrel{|}}

\newcommand{\lamter}{\iter{\lambda}}

\newcommand{\abs}[1]{{\left|{#1}\right|}}

\newcommand{\nth}{\funap}

\newcommand{\nb}{\nobreakdash}
\newcommand{\nbd}{\nobreakdash-{\hspace*{0pt}}}
\newcommand{\Icluster}{$\msf{I}$\nbd cluster}
\newcommand{\Ycluster}{$\msf{Y}$\hspace*{-1.25pt}\nbd cluster}

\def\ardx{\rho}
\def\brdx{\kappa}
\def\airdx{\ip{\ardx}}
\def\aRdx{R}

\def\aiRdx{\ip{\aRdx}}
\def\acls{\mathcal{C}}
\def\bcls{\mathcal{D}}
\def\atrm{t}
\def\aitrm{\ip{\atrm}}
\def\btrm{s}
\def\ctrm{u}
\def\natone{{1}}
\def\nattwo{{2}}
\def\anat{{n}}
\def\sposeq{{=}}
\def\poseq{\relap{\sposeq}}
\def\aPos{P}
\def\rulstr#1#2{#1\mathbin{\to}#2}
\def\ares{{/}}
\def\resa#1#2{#1\ares#2}
\def\posone{1}
\def\aars{{\to}}
\def\arsa{\relap{\aars}}
\def\alhs{\ell}
\def\arhs{r}
\def\trmat#1#2{#1|_{#2}}
\def\subap#1#2{#2^{#1}}
\def\symat{\funap}
\def\possco#1#2{#1#2}
\def\aiars{\ip{\aars}}
\def\iarsa#1{\relap{\aiars{#1}}}
\def\ip#1#2{#1_{#2}}
\def\srdxeq{{=}}
\def\rdxeq{\relap{\srdxeq}}
\def\strmeq{{=}}
\def\trmeq{\relap{\strmeq}}
\def\setstr#1{\{#1\}}
\def\rdxfula{\funap{\mathbb{\aRdx}}}

\def\aarsdev{{\sarsdev}}
\def\arsdeva{\relap{\aarsdev}}
\def\aiarsdev{\ip{\aarsdev}}
\def\iarsdeva#1{\relap{\aiarsdev{#1}}}
\def\atrmful{{\bullet}}
\def\trmfula#1{{#1}^\atrmful}
\def\ssetlub{{\cup}}
\def\setlub{\binap{\ssetlub}}
\def\ssetglb{{\cap}}
\def\setglb{\binap{\ssetglb}}
\def\ssetsum{{\uplus}}
\def\setsum{\binap{\ssetsum}}
\def\setSum#1#2{\biguplus_{#1}#2}
\def\ainj{\iota}
\def\inja{\funap{\ainj}}
\def\aidx{i}
\def\sposle{{\leq}}
\def\posle{\relap{\sposle}}
\def\sposlt{{<}}

\catcode`\@=11
\newcommand{\wash}{\mathpalette\mywash}
\newcommand{\mywash}[2]{\setbox0=\hbox{$\m@th#1{#2}$}\wd0=0pt\box0}
\catcode`\@=12

\newcommand{\sarsdev}{\wash{\,\,\,\circ}\mathord{\longrightarrow}}
\newcommand{\arsdev}{\mathrel{\sarsdev}}

\newcommand{\dstep}[3]{#1 : #2 \arsdev #3}
\newcommand{\dstepu}[4]{#1 : #2 \arsdev_{#4} #3}

\newcommand{\mulstepi}[3]{#1 \mulstepredi{#3} #2}

\newcommand{\smulstepred}{\sarsdev}

\newcommand{\simulstepred}{\wash{\,\,\,\circ}\mathord{\longleftarrow}}

\newcommand{\smulstepredi}[1]{\smulstepred_{#1}}
\newcommand{\mulstepredi}[1]{\mathrel{\smulstepredi{#1}}}

\newcommand{\dev}{\triangleright}
\newcommand{\floor}[1]{\lfloor#1\rfloor}

\def\ardx{\rho}
\def\brdx{\kappa}
\def\airdx{\ip{\ardx}}
\def\aRdx{R}

\def\aiRdx{\ip{\aRdx}}
\def\acls{c}
\def\bcls{d}
\def\atrm{t}
\def\aitrm{\ip{\atrm}}

\def\btrm{s}
\def\ctrm{u}
\def\natone{{1}}
\def\nattwo{{2}}
\def\anat{{n}}
\def\sposeq{{=}}
\def\poseq{\relap{\sposeq}}
\def\aPos{P}
\def\rulstr#1#2{#1\mathbin{\to}#2}
\def\ares{{/}}
\def\resa#1#2{#1\ares#2}
\def\posone{1}
\def\aars{{\to}}
\def\arsa{\relap{\aars}}
\def\alhs{\ell}
\def\arhs{r}
\def\trmat#1#2{#1|_{#2}}
\def\subap#1#2{#2#1}
\def\symat{\funap}
\def\possco#1#2{#1#2}
\def\aiars{\ip{\aars}}
\def\iarsa#1{\relap{\aiars{#1}}}
\def\ip#1#2{#1_{#2}}
\def\srdxeq{{=}}
\def\rdxeq{\relap{\srdxeq}}
\def\srdxort{{\bot}}
\def\rdxort#1{#1^{\srdxort}}

\def\strmeq{{=}}
\def\trmeq{\relap{\strmeq}}
\def\setstr#1{\{#1\}}
\def\rdxfula{\funap{\mathbb{\aRdx}}}

\def\aarsdev{{\sarsdev}}
\def\arsdeva{\relap{\aarsdev}}
\def\aiarsdev{\ip{\aarsdev}}
\def\iarsdeva#1{\relap{\aiarsdev{#1}}}
\def\aarsdevinv{\simulstepred}

\def\aiarsdevinv{\ip{\aarsdevinv}}
\def\iarsdevinva#1{\relap{\aiarsdevinv{#1}}}
\def\atrmful{{\bullet}}
\def\trmfula#1{{#1}^\atrmful}

\def\ssetlub{{\cup}}
\def\setlub{\binap{\ssetlub}}
\def\ssetsum{{\uplus}}
\def\setsum{\binap{\ssetsum}}
\def\setSum#1#2{\biguplus_{#1}#2}
\def\ainj{\iota}
\def\inja{\funap{\ainj}}
\def\aidx{i}
\def\ssetdif{{\setminus}}
\def\setdif#1#2{#1\ssetdif#2}

\newcommand{\ametric}{d}
\newcommand{\metric}{\bfunap{\ametric}}

\newcommand{\redrat}[2]{\mathrel{\sred_{#1,#2}}}
\newcommand{\redr}[1]{\mathrel{\sred_{#1}}}

\newcommand{\pto}{\rightharpoonup}

\definecolor{magenta}{rgb}{0.255,0,0.255}
\definecolor{lemonchiffon}{rgb}{0.255,0.250,0.205}
\definecolor{goldenrod}{rgb}{0.218,0.165,0.032}
\definecolor{azure}{rgb}{0.94,1.00,1.00}
\definecolor{blue}{rgb}{0,0,0.5}
\definecolor{brown}{rgb}{.75,.25,.25}
   \definecolor{cyan}{rgb}{0.25,0.88,0.82}
\definecolor{chocolate}{rgb}{0.82,0.41,0.12}
\definecolor{darkcyan}{rgb}{0.5,0,1}
\definecolor{darkgreen}{rgb}{0,0.39,0}
\definecolor{darkmagenta}{rgb}{0.5,0,0.5}
\definecolor{forestgreen}{rgb}{0.13,0.55,0.13}
\definecolor{lightcyan}{rgb}{0.88,1.00,1.00}
\definecolor{lightpink}{rgb}{1.00,0.71,0.76}
\definecolor{lightyellow}{rgb}{1.00,1.00,0.88}
\definecolor{lightgoldenrod}{rgb}{0.83,0.97,0.51}
\definecolor{lightgoldenrodyellow}{rgb}{0.98,0.98,0.82}
\definecolor{lightskyblue}{rgb}{0.53,0.81,0.98}
\definecolor{moccasin}{rgb}{1.00,0.89,0.71}
\definecolor{magenta}{rgb}{1,0,1}
\definecolor{navyblue}{rgb}{0,0,0.5}
\definecolor{orange}{rgb}{1.0,0.65,0.0}
\definecolor{orangered}{rgb}{1.0,0.27,0.0}
\definecolor{palegreen}{rgb}{0.60,0.98,0.60}
\definecolor{powderblue}{rgb}{0.69,0.88,0.90}
\definecolor{purple}{rgb}{1,0.5,1}
\definecolor{springgreen}{rgb}{0.0,1.0,0.5}
\definecolor{turquoise}{rgb}{0.25,0.88,0.82}
\definecolor{snow}{rgb}{1.00,0.98,0.98}
\definecolor{tan}{rgb}{0.82,0.71,0.55}

\def\aarsinv{{\leftarrow}}
\def\srelsco{{\cdot}}
\def\relsco{\binap{\srelsco}}
\def\srelle{{\subseteq}}
\def\relle{\relap{\srelle}}
\def\aobj{{a}}
\def\bobj{{b}}

\def\atrmvar{x}

\def\asym{f}
\def\bsym{g}

\def\asymzer{\aobj}

\def\syma{\funap{\asym}}
\def\symb{\funap{\bsym}}

\def\isdefd{=}

%% file: figs/SPgraph.tex
\begin{tikzpicture}
  [inner sep=1.5mm,
   node distance=5mm,
   terminal/.style={
     rectangle,rounded corners=2.25mm,minimum size=4mm,
     very thick,draw=black!50,top color=white,bottom color=black!20,
   },
   redexA/.style={draw=dblue,fill=fblue,opacity=\oblue,thick},
   redexB/.style={draw=dgreen,fill=fgreen,opacity=\ogreen,thick},
   >=stealth,xscale=.75,yscale=.75]

  \pgfmathsetseed{1}
  \begin{scope}[thick]
  \draw (0cm,0cm) -- (8.45cm,0cm); 
  \draw (0cm,-3.8cm) -- (0cm,3.8cm) node [at end,above,yshift=2mm] {$\funap{\mrm{sum}}{w,n}$};
  \end{scope}
  \node at (8.45cm,0mm) 
                        [xshift=4mm] {$+\infty$};
  \node at (-0.1cm,3.8cm) [left] {$+\infty$};
  \node at (-0.1cm,-3.8cm) [left] {$-\infty$};
  \node at (-0.1cm,0cm) [left] {$0$};

  \begin{scope}[yscale=-1]
  \draw (0,0) -- (0.2,.5) -- (0.6,-.5) -- (1,1);
  \newcommand{\expo}{1.08}
  \newcommand{\expog}{1.03}
  \coordinate (up) at (1,1);
  \coordinate (down) at (1,-1);
  \foreach \i in {1,...,60} {
    \coordinate (upLast) at (up);
    \coordinate (up) at ($(up) + (.6/\expo^\i,.1/\expog^\i)$);
    \coordinate (downLast) at (down);
    \coordinate (down) at ($(down) + (.6/\expo^\i,-.1/\expog^\i)$);
    \draw [opacity=0.3+0.7/\expo^\i] (upLast) -- ($(downLast)!.5!(down)$) -- (up);
    \begin{scope}[line width=0.8mm,gray]
    \draw (upLast) -- (up);
    \draw (downLast) -- (down);
    \end{scope}
  }
  
  \node at (0,0) {$\bullet$};
  \node at (0.2,0.5) {$\bullet$};
  \node at (0.4,0) {$\bullet$};
  \node at (0.6,-0.5) {$\bullet$};
  \node at (0.733,0) {$\bullet$};
  \node at (0.866,0.5) {$\bullet$};
  \node at (1,1) {$\bullet$};
  \end{scope}
  
  \draw [->] (1cm,-2cm) -- +(1cm,0cm) node [at end,below] {$n$};
\end{tikzpicture}

%% file: figs/venn.tex
\begin{tikzpicture}[thick]
  \draw [gray,line width=.5mm] (1.25cm,2.75cm) circle (1cm);
  \node [anchor=north] at (1.25cm,3.75cm) {RA};
  \begin{scope}
  \clip (0cm,0cm) -- (2cm,3cm) -- (-2cm,3cm) -- cycle;
  \fill [gray!50,xshift=2.5cm] (0cm,0cm) -- (2cm,3cm) -- (-2cm,3cm) -- cycle;
  \end{scope}
  \begin{scope}
  \fill [pattern=dots] (0cm,0cm) -- (.6666cm,1cm) -- (-.6666cm,1cm) -- cycle;
  \draw (0cm,0cm) -- (2cm,3cm) -- (-2cm,3cm) -- cycle;
  \node (Sw) at (0cm,0cm) {$\bullet$} node [below] {$\ssuc^\omega$};
  \end{scope}
  \begin{scope}[xshift=2.5cm]
  \fill [pattern=dots] (0cm,0cm) -- (.6666cm,1cm) -- (-.6666cm,1cm) -- cycle;
  \draw (0cm,0cm) -- (2cm,3cm) -- (-2cm,3cm) -- cycle;
  \node at (0cm,0cm) {$\bullet$} node [below] {$\spre^\omega$};
  \end{scope}
  \draw [gray,line width=.5mm] (-2cm,1cm) -- (4.5cm,1cm) node [at end,below,black] {$\SNinf$};
  
  \node (snorm) [anchor=east] at (0.2cm,2.7cm) {$\snorm{w} = \infty$}; 
  \node (pnorm) [anchor=west] at (2.3cm,2.7cm) {$\pnorm{w} = \infty$}; 
  
  \node (rasi) at (0.7cm,2.2cm) {$\bullet_{\rasi}$};
  \node (rafn) at (1.8cm,3.2cm) {$\bullet_{\rafn}$};
  \node at ($(rasi)!.5!(rafn)$) {$\bullet_{\pscountex}$};
  \node at ($(rasi)!-0.5!(rafn)$) {$\bullet_{\rasi'}$};
\end{tikzpicture}

%% file: figs/betaeta.tex
\newcommand{\tarrow}[5]{
  \begin{scope}[>=angle 90,->]
  \draw ($(#1)!#3!(#2)$) -- ($(#1)!#4!(#2)$);
  \draw [shorten >= 1.5mm] ($(#1)!#3!(#2)$) -- ($(#1)!#4!(#2)$);
  \draw [shorten >= 3mm] ($(#1)!#3!(#2)$) -- ($(#1)!#4!(#2)$) node [midway,above,inner sep=2mm] {#5};
  \end{scope}
}%
\begin{figure}[h!]
  \begin{center}
  \begin{tikzpicture}[level distance=6mm,inner sep=1mm,sibling distance=10mm,
                      app/.style={grow=-90,level distance=3.5mm},
                      abs/.style={grow=-90,level distance=6mm},
                      inner sep=0.2mm,
                      c0/.style={blue},
                      c1/.style={blue!70!red},
                      c2/.style={blue!50!red},
                      c3/.style={blue!30!red},
                      c4/.style={blue!00!red}]
    \node [c0] (root) {$\treeap$} [app,c0]
      child [c2] { node [c2] {$\lambda x_0$} [abs,c0]
        child [c4] { node [c4] {$\lambda x_1$} [abs,c2]
          child [c4] { node [c4] {$\treeap$} [app,c4]
            child [c2] { node [c2] {$\treeap$} [app,c2]
              child [c0] { node [c0] {$\treeap$} [app,c0]
                child [c1] { node [c1] {$\lambda x_{-1}$} [abs,c0]
                  child [c2] { node [c2] {$\lambda x_0$} [abs,c1]
                    child [c3] { node [c3] {$\lambda x_1$} [abs,c2]
                      child [c4] { node [c4] {$\lambda x_2$} [abs,c3]
                        child [c4] { node [c4] {$\treeap$} [app,c4]
                          child [c3] { node [c3] {$\treeap$} [app,c3]
                            child [c0] { node [c0] {$\treeap$} [app,c2]
                              child [c0] { node [c0] {$\treeap$} [app,c1]
                                child [c0] { node [c0] {$\treeap$} [app,c0]
                                  child [c0] { node [c0] {$\vdots$} }
                                  child [c0] { node [c0,inner sep=0.5mm] {$x_{-2}$} }
                                }
                                child [c1] { node [c1,inner sep=0.5mm] {$x_{-1}$} }
                              }
                              child [c2] { node [c2] {$x_{0}$} }
                            }
                            child [c3] { node [c3] {$x_1$} }
                          }
                          child [c4] { node [c4] {$x_2$} }
                        }
                      }
                    }
                  }
                }
                child [c0] { node [c0,inner sep=0.5mm] {$x_{-1}$} }
              }
              child [c2] { node [c2] {$x_0$} }
            }
            child [c4] { node [c4] {$x_1$} }
          }
        }
      }
      child [c0] { node [c0] {$x_0$} }
      ;
     
    \begin{scope}
    \node [c4] (left) [left of=root,node distance=40mm,yshift=-1cm] {$\lambda x_1$} 
      child [c4] { node [c4] {$\lambda x_2$} } [abs] 
      child [c4,level distance=7mm] { node [c4,inner sep=-1mm] {$\vdots$} }
      ;
    \end{scope}

    \node [c0] (right) [right of=root,node distance=40mm,yshift=-1cm] {$\treeap$} [app,c0]
      child [c0] { node [c0] {$\treeap$} [app,c0]
        child [c0] { node [c0] {$\treeap$} [app,c0]
          child { node [c0] {$\vdots$} }
          child [c0] { node [c0,inner sep=0.5mm] {$x_{-2}$} }
        }
        child [c0] { node [c0,inner sep=0.5mm] {$x_{-1}$ } }
      }
      child [c0] { node [c0] {$x_{0}$} }
      ;
    
    \begin{scope}[blue]
    \tarrow{root}{left}{.3}{0.85}{$\beta$}
    \end{scope}
    \begin{scope}[red]
    \tarrow{root}{right}{.3}{0.85}{$\eta$}
    \end{scope}
  \end{tikzpicture}
  \caption{\textit{Counterexample to unique normal forms in $\ilbe$.}}
  \label{fig:betaeta}
  \end{center}
\end{figure}%

%% file: figs/ps2lam.tex
\begin{tikzpicture}[inner sep=1mm,>=stealth]

  \begin{scope}[node distance=15mm]
  \node (a) {$\spre\ssuc w$};
  \node (b) [right of=a,node distance=35mm] {$(\mylam{x_i}{\pstolami{w}{i}}) \, x_i$};
  \node (c) [below of=a] {$w$};
  \node (d) [below of=b] {$\pstolami{w}{i}$};
  \end{scope}

  \draw [->] (a) -- (b) node [midway,above] {$\spstolami{i}$};
  \draw [->] (a) -- (c) node [midway,left] {$\spre\ssuc$};
  \draw [->] (c) -- (d) node [midway,above] {$\spstolami{i}$};
  \draw [->] (b) -- (d) node [midway,left] {$\beta$};

\end{tikzpicture}
\quad 
\begin{tikzpicture}[inner sep=1mm,>=stealth]

  \begin{scope}[node distance=15mm]
  \node (a) {$\ssuc\spre w$};
  \node (b) [right of=a,node distance=35mm] {$\mylam{x_{i+1}}{\pstolami{w}{i}\,x_{i+1}}$};
  \node (c) [below of=a] {$w$};
  \node (d) [below of=b] {$\pstolami{w}{i}$};
  \end{scope}

  \draw [->] (a) -- (b) node [midway,above] {$\spstolami{i}$};
  \draw [->] (a) -- (c) node [midway,left] {$\ssuc\spre$};
  \draw [->] (c) -- (d) node [midway,above] {$\spstolami{i}$};
  \draw [->] (b) -- (d) node [midway,left] {$\eta$};

\end{tikzpicture}

%% file: figs/compression.tex
\begin{figure}[hpt!]
\begin{center}
\begin{tikzpicture}
  [inner sep=1mm,
   node distance=45mm,
   terminal/.style={
   rectangle,rounded corners=2.25mm,minimum size=4mm,
   very thick,draw=black!50,top color=white,bottom color=black!20,
   },
   >=stealth,
   red/.style={thick,->>},
   infred/.style={
   thick,
   shorten >= 1mm,
   decoration={
     markings,
     mark=at position -3mm with {\arrow{stealth}},
     mark=at position -1.5mm with {\arrow{stealth}},
     mark=at position -0.001 with {\arrow{stealth}}
   },
   postaction={decorate}
   }]

  \node (s) {$s$};
  \node (s') [right of=s] {};
  \node (sb) [node distance=1mm,below of=s] {\hphantom{$s$}};
  \node (s'b) [node distance=1mm,below of=s'] {};
  \node (t) [right of=s'] {$t$};

  \draw [infred] (s) -- node [at end,above] {$\alpha$} (t);
  \draw [thick] ($(s') + (0,.5ex)$) -- ($(s'b) + (0,-.5ex)$);
  \draw [thick,dotted] ($(s') + (0,2.5ex)$) -- 
    node [midway,yshift=5.2ex] {last step of depth $d$} 
    node [midway,yshift=2.7ex] {$\beta < \alpha$}
    ($(s'b) + (0,-.5ex)$);
  \draw [infred] (s) to [bend right=45] 
    node [at end,below,xshift=-.5ex,yshift=-1.4ex] {$\le \omega$}
    node [midway] (d) {}
    (s');

  \node at ($(s)!.5!(s')$) [anchor=south] {$\ge d$};
  \node at ($(s')!.5!(t)$) [anchor=south] {$> d$};

  \draw [thick] ($(d) + (0,.5ex)$) -- ($(d) + (0,-.5ex)$);
  \draw [thick,dotted] ($(d) + (0,-2.5ex)$) -- 
    node [midway,yshift=-4.7ex] {last step of depth $d$}
    node [midway,yshift=-2.7ex] {$n < \omega$}
    ($(d) + (0,.5ex)$);

  \draw [infred] (d) to [bend right=35] node [at end,below,xshift=-.5ex,yshift=-1.4ex] {$\le \omega$} (t);
  \node at ($(s)!.5!(s')$) [yshift=-3ex] {IH};
  \node at ($(s')!.5!(t)$) [xshift=-4ex,yshift=-3.8ex] {continue};
  \node at ($(s')!.5!(t)$) [xshift=-4ex,yshift=-6ex] {with $d+1$, \ldots};
\end{tikzpicture}
\end{center}\vspace{-3ex}
\caption{\textit{Compression Lemma, in case $\alpha$ is a limit ordinal.}}
\label{fig:compression}
\end{figure}

%% file: figs/parallel.tex
\begin{tikzpicture}
  [inner sep=1.5mm,
   node distance=5mm,
   terminal/.style={
     rectangle,rounded corners=2.25mm,minimum size=4mm,
     very thick,draw=black!50,top color=white,bottom color=black!20,
   },
   redexA/.style={draw=dblue,fill=fblue,opacity=\oblue,thick},
   redexB/.style={draw=dgreen,fill=fgreen,opacity=\ogreen,thick},
   >=stealth,xscale=.75,yscale=.75]

  \begin{scope}[xscale=-1]
    \draw (2,5) -- (-2.5,0) -- (6.5,0) -- cycle;

    \begin{scope}[xshift=-.8cm]
      \fill[redexA,xshift=3.5cm,yshift=2.1cm] (1,.7) -- (.2,0) -- (1.8,0) -- cycle;
      \fill[redexB,xshift=1.5cm,yshift=.5cm] (2.6,1.8) -- (1.9,0.3) -- (2.9,0.5) -- cycle;
      \fill[redexB,xshift=2.3cm,yshift=0cm] (2.6,1.8) -- (2.3,1.1) -- (3.3,0.3) -- cycle;
      \draw [very thick,->] (4.5,2.5) .. controls (4.9,2.7) and (4.9,1.6) .. (4.1,1.5);
    \end{scope}

    \begin{scope}[xshift=-.8cm,yshift=-.3cm]
      \fill[redexB] (2,3.8) -- (1,1.7) -- (2.65,2.5) -- cycle;
      \fill[redexA] (1.4,2.2) -- (.7,.9) -- (2.1,.9) -- cycle;
      \fill[redexA] (2.4,2.8) -- (2.2,1.3) -- (3.25,1.1) -- cycle;
    \draw [very thick,->] (2,2.8) .. controls (1.5,2.6) and (2.0,1.7) .. (2.6,1.6);
    \end{scope}

    \fill[redexA,xshift=-2.5cm,yshift=-1.4cm] (2.3,3.8) -- (.5,1.7) -- (2.15,2.5) -- cycle;

  \end{scope}

\end{tikzpicture}

%% file: figs/strip-collapse.tex
\begin{center}\hspace{1cm}
\begin{tikzpicture}
  [inner sep=1mm,
   node distance=12mm,
   terminal/.style={
   rectangle,rounded corners=2.25mm,minimum size=4mm,
   very thick,draw=black!50,top color=white,bottom color=black!20,
   },
   >=stealth,
   red/.style={thick,->>},
   infred/.style={
   thick,
   shorten >= 1mm,
   decoration={
     markings,
     mark=at position -3mm with {\arrow{stealth}},
     mark=at position -1.5mm with {\arrow{stealth}},
     mark=at position 1 with {\arrow{stealth}}
   },
   postaction={decorate}
   }]

  \node (s) {$s$};
  \node (t1) [node distance=40mm,right of=s] {$t_1$};
  \node (t2) [below of=s] {$t_2$};
  \node (u) [below of=t1] {$u$};

  \draw [infred] (s) -- (t1);
  \draw [red,->] (s) -- (t2);
  \draw [infred,dashed] (t2) -- (u);
  \draw [red,->] (t1) -- (u);

  \draw [thick] ($(s)!.5!(t2) + (-.75ex,-.20ex)$) -- ($(s)!.5!(t2) + (.75ex,-.20ex)$);
  \draw [thick] ($(s)!.5!(t2) + (-.75ex,.20ex)$) -- ($(s)!.5!(t2) + (.75ex,.20ex)$);
  \draw [thick] ($(t1)!.5!(u) + (-.75ex,-.20ex)$) -- ($(t1)!.5!(u) + (.75ex,-.20ex)$);
  \draw [thick] ($(t1)!.5!(u) + (-.75ex,.20ex)$) -- ($(t1)!.5!(u) + (.75ex,.20ex)$);

  \node at ($(s)!.5!(t1)$) [anchor=south] {$\ge d_\aseq$};
  \node at ($(s)!.5!(t2)$) [anchor=east,xshift=-.7ex] {$\ge d_\bseq$};
  \node at ($(t2)!.5!(u)$) [anchor=north] {$\ge \bfunap{\min}{d_\aseq}{d_\bseq}$};
  \node at ($(t1)!.5!(u)$) [anchor=east,xshift=-.7ex] {$\ge \bfunap{\min}{d_\aseq}{d_\bseq}$};
\end{tikzpicture}
\end{center}\vspace{-3ex}

%% file: figs/strip.tex
\begin{center}\hspace{2cm}
\begin{tikzpicture}
  [inner sep=1mm,
   node distance=12mm,
   terminal/.style={
   rectangle,rounded corners=2.25mm,minimum size=4mm,
   very thick,draw=black!50,top color=white,bottom color=black!20,
   },
   >=stealth,
   red/.style={thick,->>},
   infred/.style={
   thick,
   shorten >= 1mm,
   decoration={
     markings,
     mark=at position -3mm with {\arrow{stealth}},
     mark=at position -1.5mm with {\arrow{stealth}},
     mark=at position 1 with {\arrow{stealth}}
   },
   postaction={decorate}
   }]

  \node (s) {$s$};
  \node (t1) [node distance=40mm,right of=s] {$t_1$};
  \node (t2) [below of=s] {$t_2$};
  \node (u) [below of=t1] {$u$};

  \draw [infred] (s) -- (t1);
  \draw [red,->] (s) -- (t2);
  \draw [infred,dashed] (t2) -- (u);
  \draw [red,->] (t1) -- (u);

  \draw [thick] ($(s)!.5!(t2) + (-.75ex,-.20ex)$) -- ($(s)!.5!(t2) + (.75ex,-.20ex)$);
  \draw [thick] ($(s)!.5!(t2) + (-.75ex,.20ex)$) -- ($(s)!.5!(t2) + (.75ex,.20ex)$);
  \draw [thick] ($(t1)!.5!(u) + (-.75ex,-.20ex)$) -- ($(t1)!.5!(u) + (.75ex,-.20ex)$);
  \draw [thick] ($(t1)!.5!(u) + (-.75ex,.20ex)$) -- ($(t1)!.5!(u) + (.75ex,.20ex)$);

  \node at ($(s)!.5!(t1)$) [anchor=south] {$\ge d_\aseq$};
  \node at ($(s)!.5!(t2)$) [anchor=east,xshift=-.7ex] {$\ge d_\bseq$};
  \node at ($(t2)!.5!(u)$) [anchor=north] {$\ge \bfunap{\min}{d_\aseq}{d_\bseq+1}$};
  \node at ($(t1)!.5!(u)$) [anchor=east,xshift=-.7ex] {$\ge \bfunap{\min}{d_\bseq}{d_\aseq+1}$};
\end{tikzpicture}
\end{center}\vspace{-3ex}

%% file: figs/strip-proof.tex
\begin{figure}[hpt!]
\begin{center}
\begin{tikzpicture}
  [inner sep=1mm,
  node distance=25mm,
  terminal/.style={
	rectangle,rounded corners=2.25mm,minimum size=4mm,
	very thick,draw=black!50,top color=white,bottom color=black!20,
  },
  >=stealth,
  red/.style={thick,->>},
  infred/.style={
	thick,
	shorten >= 1mm,
	decoration={
	  markings,
	  mark=at position -3mm with {\arrow{stealth}},
	  mark=at position -1.5mm with {\arrow{stealth}},
	  mark=at position 1 with {\arrow{stealth}}
	},
	postaction={decorate}
  }]

  \node (s) {$s \equiv s_0$};
  \node (s1) [node distance=13mm,right of=s] {$s_1$};
  \node (s2) [node distance=10mm,right of=s1] {$\ldots$};
  \node (sn0) [node distance=10mm,right of=s2] {$s_{n_0}$};
  \node (sm0) [node distance=30mm,right of=sn0] {$s_{m_0}$};
  \node (t1) [node distance=30mm,right of=sm0] {$t_1$};
  \node (t2) [below of=s] {$t_2 \equiv s_0'$};
  \node (s1') [below of=s1] {$s_1'$};
  \node (s2') [below of=s2] {$\ldots$};
  \node (sn0') [below of=sn0] {$s_{n_0}'$};
  \node (sm0') [below of=sm0] {$s_{m_0}'$};
  \node (u) [below of=t1] {$u$};

  \node (sn0'') [node distance=12mm,below of=sn0] {$s_{n_0}''$};
  \node (sm0'') [node distance=12mm,below of=sm0] {$s_{m_0}''$};
  \node (t1'') [node distance=12mm,below of=t1] {$t_1''$};

  \draw [red,->] (s) -- (s1);
  \draw [red,->] (s1) -- (s2);
  \draw [red,->] (s2) -- (sn0);
  \draw [red,->>] (sn0) -- (sm0) node [midway,above] {$\ge d$};
  \draw [infred] (sm0) -- (t1) node [midway,above] {$\ge d+p$};

  \draw [red,->] (s) -- (t2) node [midway,left,xshift=-.5ex] {$\astep = \astep_0$};
  \draw [thick] ($(s)!.5!(t2) + (-.75ex,-.20ex)$) -- ($(s)!.5!(t2) + (.75ex,-.20ex)$);
  \draw [thick] ($(s)!.5!(t2) + (-.75ex,.20ex)$) -- ($(s)!.5!(t2) + (.75ex,.20ex)$);

  \draw [red,->] (s1) -- (s1') node [midway,right,xshift=.5ex] {$\astep_1$};
  \draw [thick] ($(s1)!.5!(s1') + (-.75ex,-.20ex)$) -- ($(s1)!.5!(s1') + (.75ex,-.20ex)$);
  \draw [thick] ($(s1)!.5!(s1') + (-.75ex,.20ex)$) -- ($(s1)!.5!(s1') + (.75ex,.20ex)$);

  \draw [red,->] (sn0) -- (sn0'') node [midway,right,xshift=.5ex] {$\astep_{n_0,<d}$};
  \draw [thick] ($(sn0)!.5!(sn0'') + (-.75ex,-.20ex)$) -- ($(sn0)!.5!(sn0'') + (.75ex,-.20ex)$);
  \draw [thick] ($(sn0)!.5!(sn0'') + (-.75ex,.20ex)$) -- ($(sn0)!.5!(sn0'') + (.75ex,.20ex)$);
  \draw [red,->] (sn0'') -- (sn0') node [midway,right,xshift=.5ex] {$\astep_{n_0,\ge d}$};
  \draw [thick] ($(sn0'')!.5!(sn0') + (-.75ex,-.20ex)$) -- ($(sn0'')!.5!(sn0') + (.75ex,-.20ex)$);
  \draw [thick] ($(sn0'')!.5!(sn0') + (-.75ex,.20ex)$) -- ($(sn0'')!.5!(sn0') + (.75ex,.20ex)$);

  \draw [red,->] (sm0) -- (sm0'') node [midway,right,xshift=.5ex] {$\bstep \subseteq \astep_{n_0,<d}$};
  \draw [thick] ($(sm0)!.5!(sm0'') + (-.75ex,-.20ex)$) -- ($(sm0)!.5!(sm0'') + (.75ex,-.20ex)$);
  \draw [thick] ($(sm0)!.5!(sm0'') + (-.75ex,.20ex)$) -- ($(sm0)!.5!(sm0'') + (.75ex,.20ex)$);
  \draw [red,->] (sm0'') -- (sm0') node [midway,right,xshift=.5ex] {$\astep_{m_0,\ge d}$};
  \draw [thick] ($(sm0'')!.5!(sm0') + (-.75ex,-.20ex)$) -- ($(sm0'')!.5!(sm0') + (.75ex,-.20ex)$);
  \draw [thick] ($(sm0'')!.5!(sm0') + (-.75ex,.20ex)$) -- ($(sm0'')!.5!(sm0') + (.75ex,.20ex)$);

  \draw [red,->] (t1) -- (t1'') node [midway,right,xshift=.5ex] {$\bstep$};
  \draw [thick] ($(t1)!.5!(t1'') + (-.75ex,-.20ex)$) -- ($(t1)!.5!(t1'') + (.75ex,-.20ex)$);
  \draw [thick] ($(t1)!.5!(t1'') + (-.75ex,.20ex)$) -- ($(t1)!.5!(t1'') + (.75ex,.20ex)$);
  \draw [red,->,dashed] (t1'') -- (u) node [midway,right,xshift=.5ex] {$\ge d$};
  \draw [thick] ($(t1'')!.5!(u) + (-.75ex,-.20ex)$) -- ($(t1'')!.5!(u) + (.75ex,-.20ex)$);
  \draw [thick] ($(t1'')!.5!(u) + (-.75ex,.20ex)$) -- ($(t1'')!.5!(u) + (.75ex,.20ex)$);

  \draw [red,->] (t2) -- (s1');
  \draw [red,->] (s1') -- (s2');
  \draw [red,->] (s2') -- (sn0');
  \draw [infred,dashed] (sn0') -- (sm0');
  \draw [infred,dashed] (sm0'') -- (t1'')  node [midway,above] {$\ge d$};
  \draw [infred,dashed] (sm0') -- (u)  node [midway,above] {$\ge d$};
\end{tikzpicture}
\end{center}\vspace{-3ex}
\caption{\textit{Parallel Moves Lemma, proof overview.}}
\label{fig:pml:proof}
\end{figure}

%% file: figs/confluence.tex
\begin{figure}[hpt!]
\begin{center}
\begin{tikzpicture}
  [inner sep=1mm,
   node distance=16mm,
   terminal/.style={
     rectangle,rounded corners=2.25mm,minimum size=4mm,
     very thick,draw=black!50,top color=white,bottom color=black!20,
   },
   >=stealth,
   red/.style={thick,->>},
   infred/.style={
     thick,
     shorten >= 1mm,
     decoration={
       markings,
       mark=at position -3mm with {\arrow{stealth}},
       mark=at position -1.5mm with {\arrow{stealth}},
       mark=at position 1 with {\arrow{stealth}}
     },
     postaction={decorate}
   }]

  \node (s) {$s$};
  \node (s1) [node distance=40mm,right of=s] {$s_1$};
  \node (t1) [node distance=60mm,right of=s1] {$t_1$};
  \node (s2) [below of=s] {$s_2$};
  \node (t2) [below of=s2] {$t_2$};
  \node (s') [below of=s1] {$s'$};
  \node (t1') [below of=t1] {$t_1'$};
  \node (t2') [below of=s'] {$t_2'$};
  \node (u) [below of=t1'] {$u$};

  \draw [red] (s) -- (s1);
  \draw [infred] (s1) -- (t1);
  \draw [red] (s) -- (s2);
  \draw [infred] (s2) -- (t2);

  \draw [red,dashed] (s1) -- (s');
  \draw [thick] ($(s1)!.5!(s') + (-.75ex,-.20ex)$) -- ($(s1)!.5!(s') + (.75ex,-.20ex)$);
  \draw [thick] ($(s1)!.5!(s') + (-.75ex,.20ex)$) -- ($(s1)!.5!(s') + (.75ex,.20ex)$);

  \draw [red,dashed] (s2) -- (s');
  \draw [thick] ($(s2)!.5!(s') + (-.20ex,-.75ex)$) -- ($(s2)!.5!(s') + (-.20ex,.75ex)$);
  \draw [thick] ($(s2)!.5!(s') + (.20ex,-.75ex)$) -- ($(s2)!.5!(s') + (.20ex,.75ex)$);

  \draw [infred,dashed] (s') -- (t1');
  \draw [infred,dashed] (t1) -- (t1');
  \draw [infred,dashed] (s') -- (t2');
  \draw [infred,dashed] (t2) -- (t2');
  \draw [infred,dashed] (t1') -- (u);
  \draw [infred,dashed] (t2') -- (u);

  \node at ($(s)!.5!(s1)$) [anchor=south] {$\ge d$};
  \node at ($(s)!.5!(s2)$) [anchor=east] {$\ge d$};
  \node at ($(s1)!.5!(t1)$) [anchor=south] {$>d$};
  \node at ($(s2)!.5!(t2)$) [anchor=east] {$>d$};
  \node at ($(s2)!.5!(s')$) [anchor=south,yshift=.5ex] {$\ge d$};
  \node at ($(s1)!.5!(s')$) [anchor=west,xshift=.5ex] {$\ge d$};

  \node at ($(s')!.5!(t1')$) [anchor=south] {$> d$};
  \node at ($(s')!.5!(t2')$) [anchor=west] {$> d$};

  \node at ($(t2)!.5!(t2')$) [anchor=south] {$\ge d$};
  \node at ($(t1)!.5!(t1')$) [anchor=west] {$\ge d$};

  \node at ($(s)!.5!(s')$) {finitary diagram};
  \node at ($(s')!.5!(t1)$) {PML (Lemma~\ref{lem:pml})};
  \node at ($(s')!.5!(t2)$) {PML (Lemma~\ref{lem:pml})};
  \node at ($(s')!.5!(u)$) {\parbox{3cm}{\centerline{repeat construction}\centerline{with $d+1$}}};
\end{tikzpicture}
\end{center}\vspace{-3ex}
\caption{\textit{Infinitary confluence.}}
\label{fig:confluence}
\end{figure}

%% file: figs/chain.tex
\begin{tikzpicture}
  [inner sep=1.5mm,
   node distance=5mm,
   terminal/.style={
     rectangle,rounded corners=2.25mm,minimum size=4mm,
     very thick,draw=black!50,top color=white,bottom color=black!20,
   },
   redexA/.style={draw=dblue,fill=fblue,opacity=\oblue,thick},
   redexB/.style={draw=dgreen,fill=fgreen,opacity=\ogreen,thick},
   >=stealth,xscale=.75,yscale=.75]

  \pgfmathsetseed{1}

  \newcommand{\expo}{1.3}
  \coordinate (current point) at (2,3.8);
  \foreach \i in {1,...,10} {
    \fill[redexA,{opacity=min(1/\expo^\i,.5)}] (current point) -- ($(current point) + 1/\expo^\i*(-.5,-1)$) -- ($(current point) + 1/\expo^\i*(.5,-1)$) -- cycle;
    \coordinate (current point) at ($(current point) + (.45*rand/\expo^\i,-.8/\expo^\i)$);
    \fill[redexB,{opacity=min(1/\expo^\i,.5)}] (current point) -- ($(current point) + 1/\expo^\i*(-.5,-1)$) -- ($(current point) + 1/\expo^\i*(.5,-1)$) -- cycle;
    \coordinate (current point) at ($(current point) + (.45*rand/\expo^\i,-.8/\expo^\i)$);
  }
\end{tikzpicture}

%% file: figs/AAA.tex
\begin{tikzpicture}
  [inner sep=1.5mm,
  node distance=5mm,
  terminal/.style={
    rectangle,rounded corners=2.25mm,minimum size=4mm,
    very thick,draw=black!50,top color=white,bottom color=black!20,
  },
  redexA/.style={draw=dblue,fill=fblue,opacity=\oblue,thick},
  redexB/.style={draw=dgreen,fill=fgreen,opacity=\ogreen,thick},
  >=stealth,xscale=.75,yscale=.75]

  \coordinate (current point) at (2,3.8);
  \foreach \i in {1,...,15} {
    \node (n\i) [anchor=west] at (current point) {$A($};
    \coordinate (current point) at ($(current point) + (.55,0)$);
  }
  \node [anchor=west] at (current point) {$\ldots)\ldots)$};

  \coordinate (current point) at (2,3.8);
  \foreach \i in {1,...,4} {
    \draw [redexA] ($(current point) + (.2,.35)$) rectangle ($(current point) + (1.6,.55)$);
    \coordinate (current point) at ($(current point) + (1.1,0)$);
    \draw [redexB] ($(current point) + (.2,-.25)$) rectangle ($(current point) + (1.6,-.45)$);
    \coordinate (current point) at ($(current point) + (1.1,0)$);
  }
\end{tikzpicture}

%% file: figs/ortho.tex
\begin{tikzpicture}
  [inner sep=1.5mm,
   node distance=5mm,
   terminal/.style={
     rectangle,rounded corners=2.25mm,minimum size=4mm,
     very thick,draw=black!50,top color=white,bottom color=black!20,
   },
   redexA/.style={draw=dblue,fill=fblue,opacity=\oblue,thick},
   redexB/.style={draw=dgreen,fill=fgreen,opacity=\ogreen,thick},
   >=stealth,xscale=.75,yscale=.75]

  \begin{scope}[xshift=0,yshift=0]
    \draw (0,0) -- (2,4) -- (4,0) -- cycle;
    \fill[redexA] (1.4,2.2) -- (.7,.9) -- (2.1,.9) -- cycle;
    \node [terminal] at (1.4,1.3) {1};
    \fill[redexA] (2.4,2.8) -- (2.2,1.3) -- (3.25,1.1) -- cycle;
    \node [terminal] at (2.6,1.8) {2};
  \end{scope}
  \node at (4.5,2) {$\cup$};
  \begin{scope}[xshift=5cm,yshift=0]
    \draw (0,0) -- (2,4) -- (4,0) -- cycle;
    \fill[redexB] (2,3.8) -- (1,1.7) -- (2.65,2.5) -- cycle;
    \node [terminal] at (2,2.8) {3};
    \fill[redexB] (2.6,1.8) -- (2.3,0.3) -- (3.3,0.3) -- cycle;
    \node [terminal] at (2.9,1) {5};
    \fill[redexB] (1,.7) -- (.2,0) -- (1.8,0) -- cycle;
    \node [terminal] at (1,.1) {4};
  \end{scope}
  \node at (10,2) {$=$};
  \begin{scope}[xshift=11cm,yshift=0]
    \draw (0,0) -- (2,4) -- (4,0) -- cycle;
    \fill[redexB] (2,3.8) -- (1,1.7) -- (2.65,2.5) -- cycle;
    \fill[redexA] (1.4,2.2) -- (.7,.9) -- (2.1,.9) -- cycle;
    \fill[redexA] (2.4,2.8) -- (2.2,1.3) -- (3.25,1.1) -- cycle;
    \fill[redexB] (2.6,1.8) -- (2.3,0.3) -- (3.3,0.3) -- cycle;
    \fill[redexB] (1,.7) -- (.2,0) -- (1.8,0) -- cycle;
  \end{scope}
\end{tikzpicture}

%% file: figs/cases.tex
\begin{tikzpicture}
  [inner sep=1.5mm,
   node distance=5mm,
   terminal/.style={
         rectangle,rounded corners=2.25mm,minimum size=4mm,
         very thick,draw=black!50,top color=white,bottom color=black!20,
   },
   redexA/.style={draw=dblue,fill=fblue,opacity=\oblue,thick},
   redexB/.style={draw=dgreen,fill=fgreen,opacity=\ogreen,thick},
   >=stealth,xscale=.75,yscale=.75]

  \begin{scope}[xshift=0,yshift=0]
        \draw (0,0) -- (2,4) -- (4,0) -- cycle;
        \fill[redexA] (2,3.8) -- (.7,1.2) -- (2.75,2.3) -- cycle;
        \node [terminal] at (2,2.9) {$u$};
        \fill[redexB] (1.8,2.5) -- (1.5,1.0) -- (2.7,0.8) -- cycle;
        \node [terminal] at (2,1.4) {$v$};
        \node at (2,-.5) {case (i)};
  \end{scope}
  \begin{scope}[xshift=5cm,yshift=0]
        \draw (0,0) -- (2,4) -- (4,0) -- cycle;
        \fill[redexA] (2,3.8) -- (.7,1.2) -- (2.75,2.3) -- cycle;
        \node [terminal] at (2,2.9) {$u$};
        \fill[redexB] (1.3,2.1) -- (1.0,0.6) -- (2.2,0.4) -- cycle;
        \node [terminal] at (1.5,1.0) {$v$};
        \fill[redexB] (2.4,2.4) -- (2.2,0.8) -- (3.2,0.8) -- cycle;
        \node [terminal] at (2.7,1.5) {$w$};
        \node at (2,-.5) {case (ii)};
  \end{scope}
  \begin{scope}[xshift=10cm,yshift=0cm]
        \draw (0,0) -- (2,4) -- (4,0) -- cycle;
        \fill[redexA] (2,3.8) -- (.6,1.0) -- (2.85,2.1) -- cycle;
        \node [terminal] at (2,2.9) {$u$};
        \fill[redexB] (1.5,2.7) -- (1.1,1.8) -- (2.8,1.5) -- cycle;
        \node [terminal] at (1.65,2.1) {$v$};
        \fill[redexB] (1.4,1.65) -- (0.8,0.2) -- (2.5,0.4) -- cycle;
        \node [terminal] at (1.55,0.8) {$w$};
        \node at (2,-.5) {case (iii)};
  \end{scope}
  \begin{scope}[xshift=15cm,yshift=0cm]
        \draw (0,0) -- (2,4) -- (4,0) -- cycle;
        \fill[redexA] (2,3.8) -- (.6,1.0) -- (2.85,2.1) -- cycle;
        \node [terminal] at (2,2.9) {$u$};
        \fill[redexB] (1.5,2.7) -- (1.1,1.8) -- (2.8,1.5) -- cycle;
        \node [terminal] at (1.65,2.1) {$v$};
        \fill[redexA] (2.5,1.7) -- (2.1,1.0) -- (2.85,1.0) -- cycle;
        \node [terminal] at (3,.9) {$m$};
        \fill[redexB] (1.4,1.65) -- (0.8,0.2) -- (2.5,0.4) -- cycle;
        \node [terminal] at (1.55,0.8) {$w$};
        \node at (2,-.5) {case (iv)};
  \end{scope}
\end{tikzpicture}

%% file: figs/table.tex
{
\definecolor{mblue}{RGB}{186,178,255}
\definecolor{mgreen}{RGB}{0,128,0}
\newcommand{\yes}{yes}
\newcommand{\no}{no}
\newcommand{\ouryes}{\textcolor{mgreen}{\textbf{\yes}}}
\newcommand{\ourno}{\textcolor{mgreen}{\textbf{\no}}}
\newcommand{\ourq}{\textcolor{mgreen}{\textbf{?}}}
\newcommand{\yesorno}{\yes/\no}
\begin{tikzpicture}[thick,xscale=1.5,yscale=0.7]
  \foreach \x in {0,...,7} {
    \foreach \y in {0,...,7} {
      \draw (\x cm,\y cm) rectangle +(1cm,1cm);
    }
  }
  \foreach \x in {0,...,3} { \foreach \y in {2,...,7} { \node at (0.5+\x,0.5+\y) {\yes}; } }
  \foreach \x in {1,...,3} { \foreach \y in {0,...,1} { \node at (0.5+\x,0.5+\y) {\yes}; } }
  \foreach \x/\y in {6/2,6/3,7/3} { \node at (0.5+\x,0.5+\y) {\yes}; }

  \foreach \x in {4,...,5} { \foreach \y in {0,...,3} { \node at (0.5+\x,0.5+\y) {\no}; } }
  \foreach \x in {6,...,7} { \foreach \y in {0,...,1} { \node at (0.5+\x,0.5+\y) {\ourno}; } }
  
  \foreach \x/\y/\r in {4/4/\ouryes,5/4/\no,6/4/\ourq,7/4/\ourq,
                        4/5/\ouryes,5/5/\ouryes,6/5/\ouryes,7/5/\ouryes,
                        4/6/\ouryes,5/6/\no,6/6/\ourno,7/6/\ourno,
                        4/7/\yes,5/7/\no,6/7/\yes,7/7/\yes,
                        0/0/\yes,0/1/\yes\footnotemark[2],
                        7/2/\yes} { \node at (0.5+\x,0.5+\y) {\r}; } 
  
  \foreach \x in {0,...,7} { \foreach \y in {8} { \draw [very thick,fill=mblue](\x cm,\y cm) rectangle +(1cm,1cm); } }
  \foreach \x/\r in {0/PML,1/CR,2/UN,3/NF,4/PML$^{\infty}$,5/CR$^{\infty}$,6/UN$^{\infty}$,7/NF$^{\infty}$} { \node at (0.5+\x,8.5) {\r}; }

  \foreach \x in {-1} { \foreach \y in {0,...,7} { \draw [very thick,fill=mblue](-0.5cm + \x cm,\y cm) rectangle +(1.5cm,1cm); } }
  \foreach \y/\r in {0/WOCRS,1/$\lambda\beta\eta$,2/fe-OCRS,3/$\lambda\beta$,4/1c-WOTRS,5/nc-WOTRS,6/WOTRS,7/OTRS} { \node at (-0.75,0.5+\y) {\r}; }

  \begin{scope}[line width=0.7mm,opacity=0.5]
  \draw (4cm,0cm) -- (4cm,9.8cm);
  \draw (-2cm,4cm) -- (8cm,4cm);
  \end{scope}
  \node at (2cm,9.5cm) {\textit{finitary}};
  \node at (6cm,9.5cm) {\textit{infinitary}};
  \node [rotate=90] at (-1.8cm,2cm) {\textit{higher-order}};
  \node [rotate=90] at (-1.8cm,6cm) {\textit{first-order}};
\end{tikzpicture}
}

%% file: main.bbl
\newcommand{\etalchar}[1]{$^{#1}$}
\begin{thebibliography}{KKSdV95}

\bibitem[BKV00]{beth:klop:vrij:2000}
I.~Bethke, J.W. Klop, and de~Vrijer, R.C.
\newblock {Descendants and Origins in Term Rewriting}.
\newblock {\em Information and Computation}, 159(1--2):59--124, 2000.

\bibitem[CR36]{Chur:Ross:36}
A.~Church and J.B. Rosser.
\newblock Some properties of conversion.
\newblock {\em Transactions of the American Mathematical Society}, 39:472--482,
  1936.

\bibitem[DKP91]{ders:kapl:plai:1991}
N.~Dershowitz, S.~Kaplan, and D.~A. Plaisted.
\newblock {Rewrite, Rewrite, Rewrite, Rewrite, Rewrite, \ldots}.
\newblock {\em Theoretical Computer Science}, 83(1):71--96, 1991.

\bibitem[EGH{\etalchar{+}}10]{endr:grab:hend:klop:oost:2010}
J.~Endrullis, C.~Grabmayer, D.~Hendriks, J.W. Klop, and V.~van Oostrom.
\newblock {Unique Normal Forms in Infinitary Weakly Orthogonal Rewriting}.
\newblock In C.~Lynch, editor, {\em Proc.\ Conf.\ on Rewriting Techniques and
  Applications (RTA~2010)}, volume~6 of {\em Leibniz International Proceedings
  in Informatics (LIPIcs)}, pages 85--102, Dagstuhl, Germany, 2010. Schloss
  Dagstuhl--Leibniz-Zentrum f\"{u}r Informatik.

\bibitem[EHH{\etalchar{+}}13]{endr:hans:hend:polo:silv:2013}
J.~Endrullis, H.H. Hansen, D.~Hendriks, A.~Polonsky, and A.~Silva.
\newblock {A Coinductive Treatment of Infinitary Rewriting}.
\newblock {\em CoRR}, abs/1306.6224, 2013.

\bibitem[EHK12]{endr:hend:klop:2012}
J.~Endrullis, D.~Hendriks, and J.W. Klop.
\newblock {Highlights in Infinitary Rewriting and Lambda Calculus}.
\newblock {\em Theoretical Computer Science}, 464:48--71, 2012.

\bibitem[KdV05]{klop:vrij:2005}
J.W. Klop and R.C. de~Vrijer.
\newblock {Infinitary Normalization}.
\newblock In {\em We Will Show Them: Essays in Honour of Dov Gabbay}, volume~2,
  pages 169--192. {C}ollege {P}ublications, 2005.

\bibitem[Ket06]{kete:2006}
J.~Ketema.
\newblock {\em {B\"{o}hm-Like Trees for Rewriting}}.
\newblock PhD thesis, Vrije Universiteit Amsterdam, 2006.

\bibitem[Ket08]{kete:2008}
J.~Ketema.
\newblock {On Normalisation of Infinitary Combinatory Reduction Systems}.
\newblock In A.~Voronkov, editor, {\em Proc.\ Conf.\ on Rewriting Techniques
  and Applications (RTA~2008)}, volume 5117 of {\em LNCS}, pages 172--186,
  2008.

\bibitem[KKSdV95]{kenn:klop:slee:vrie:1995}
R.~Kennaway, J.W. Klop, M.R. Sleep, and F.-J. de~Vries.
\newblock {Transfinite Reductions in Orthogonal Term Rewriting Systems}.
\newblock {\em Information and Computation}, 119(1):{18--38}, 1995.

\bibitem[KKvO04]{kete:klop:oost:2004a}
J.~Ketema, J.W. Klop, and V.~van Oostrom.
\newblock {Vicious Circles in Rewriting Systems}.
\newblock CKI Preprint~52, Universiteit Utrecht, 2004.
\newblock Available at \url{http://www.phil.uu.nl/preprints/aips/}.

\bibitem[Klo80]{klop:1980}
J.~W. Klop.
\newblock {\em {Combinatory Reduction Systems}}, volume 127 of {\em
  Mathematical centre tracts}.
\newblock Mathematisch Centrum, 1980.

\bibitem[KS09]{kete:simo:2009}
J.~Ketema and J.G. Simonsen.
\newblock {Infinitary Combinatory Reduction Systems: Confluence}.
\newblock {\em LMCS}, 5(4):1--29, 2009.

\bibitem[Oos97]{oost:1997}
V.~van Oostrom.
\newblock Finite family developments.
\newblock In Hubert Comon, editor, {\em Proc.\ Conf.\ on Rewriting Techniques
  and Applications (RTA~1997)}, volume 1232 of {\em Lecture Notes in Computer
  Science}, pages 308--322. Springer, 1997.

\bibitem[SdV02]{seve:vrie:2002}
P.~Severi and F.-J. de~Vries.
\newblock {An Extensional B{\"o}hm Model}.
\newblock In S.~Tison, editor, {\em Proc.\ Conf.\ on Rewriting Techniques and
  Applications (RTA~2002)}, volume 2378 of {\em LNCS}, pages 159--173, 2002.

\bibitem[SdV05]{seve:vrie:2005}
P.~Severi and F.-J. de~Vries.
\newblock {Continuity and Discontinuity in Lambda Calculus}.
\newblock In {\em TLCA 2005}, volume 3461 of {\em LNCS}, pages 369--385, 2005.

\bibitem[Ter03]{terese:2003}
Terese.
\newblock {\em {Term Rewriting Systems}}, volume~55 of {\em Cambridge Tracts in
  Theoretical Computer Science}.
\newblock Cambridge University Press, 2003.

\bibitem[Zan08]{zant:2008}
H.~Zantema.
\newblock {Normalization of Infinite Terms}.
\newblock In A.~Voronkov, editor, {\em Proc.\ Conf.\ on Rewriting Techniques
  and Applications (RTA~2008)}, volume 5117 of {\em LNCS}, pages 441--455,
  2008.

\end{thebibliography}
